\documentclass{article}

\usepackage[round]{natbib}
\usepackage{amsmath}
\usepackage{amssymb}
\usepackage{amsthm}
\usepackage{xspace}
\usepackage{thmtools}
\usepackage{thm-restate}
\usepackage{hyperref}

\declaretheorem[name=Theorem,parent=section]{restatableTheorem}
\declaretheorem[name=Lemma,sibling=restatableTheorem]{restatableLemma}
\declaretheorem[name=Proposition,sibling=restatableTheorem]{restatableProposition}
\declaretheorem[name=Corollary,sibling=restatableTheorem]{restatableCorollary}

\declaretheorem[name=Definition,sibling=restatableTheorem]{restatableDefinition}
\declaretheorem[name=Claim,sibling=restatableTheorem]{claim}
\declaretheorem[name=Example,sibling=restatableTheorem]{example}

\let\shortcite\citeyearpar
\newcommand{\complexityclassname}[1]{\ensuremath{\mathrm{#1}}}
\newcommand{\PTIME}{\complexityclassname{P}}
\newcommand{\NP}{\complexityclassname{NP}}
\newcommand{\mathtext}[1]{\ensuremath{\mathrm{\text{#1}}}}
\newcommand{\set}[1]{\{#1\}}
\newcommand{\suchthat}{\ensuremath{\ \vert\ }}

\newcommand{\card}[1]{{ \mathopen\parallel {#1} \mathclose\parallel }}

\newcommand{\score}[1]{{{\mbox{\it{score}}(#1)}}}
\newcommand{\surpl}[1]{{{\mbox{\it{surplus}}(#1)}}}
\newcommand{\surplsub}[2]{{{\mbox{\it{surplus}}_{#1}(#2)}}}

\newcommand{\scoresub}[2]{{{\mbox{\it{score}}_{#1}(#2)}}}
\newcommand{\ccdvstar}{CCDV$^*$\xspace}
\newcommand{\dual}[1]{\ensuremath{\mathop{dual}\left(#1\right)}}
\newcommand{\RESTcandidates}{\ensuremath{REST}}
\newcommand{\scorethreedm}[1]{\ensuremath{\scoresub{3DM}{#1}}}
\newcommand{\scorefinal}[1]{\ensuremath{\scoresub{final}{#1}}}
\newcommand{\scoresetup}[1]{\ensuremath{\scoresub{setup}{#1}}}
\newcommand{\scorelosep}[1]{\ensuremath{\scoresub{lose}{#1}}}
\newcommand{\ceilfrac}[2]{\ensuremath{\left\lceil\frac{#1}{#2}\right\rceil}}

\title{Complexity Dichotomies for Unweighted Scoring Rules\thanks{Supported in part by NSF grant CCF-1101452 and by COST Action IC1205. Work done in part while H.~Schnoor visited the University of Rochester supported by an STSM grant of Cost Action IC1205. }}

\author{Edith Hemaspaandra\\
Department of Computer Science\\
Rochester Institute of Technology\\
Rochester, NY 14623, USA \\
\and Henning Schnoor \\
Institut f\"ur Informatik\\
Christian-Albrechts-Universit\"at zu Kiel \\
24098 Kiel, Germany \\
}

\begin{document}

\maketitle

\begin{abstract}
Scoring systems are an extremely important class of election systems. 
We study the complexity of manipulation, constructive control by deleting
voters (CCDV), and bribery for scoring systems.
For manipulation, 
we show that for all scoring rules with a constant number of different
coefficients, manipulation is in P. And we conjecture that
there is no dichotomy theorem.

On the other hand, we obtain dichotomy theorems for CCDV and bribery problem. 
More precisely, we show that both of these problems are easy for 1-approval, 2-approval, 
1-veto, 2-veto, 3-veto, generalized 2-veto, and $(2,1,...,1,0)$, 
and hard in all other cases. These results are the ``dual'' of the 
dichotomy theorem for the constructive control by adding 
voters (CCAV) problem
from~\cite{HemaspaandraHemaspaandraSchnoor-CCAV-AAAI-2014},
but do not at all follow from that result. In particular, proving 
hardness for CCDV is harder than for CCAV since we do not 
have control over what the controller can delete, and proving 
easiness for bribery tends to be harder than for control, since bribery can 
be viewed as control followed by manipulation.
\end{abstract}

\section*{Introduction}
Elections are an important way to make decisions, both in human 
and electronic settings. Arguably the most important class
of election systems are the scoring rules. A scoring rule is defined
by, for each number $m$ of candidates, a scoring
vector $\alpha_1 \geq \alpha_2 \geq \cdots \geq \alpha_m$.
In addition, we typically want these vectors to be somehow similar.
This is captured nicely by the notion of pure scoring rules
from~\cite{bet-dor:j:possible-winner-dichotomy} where the 
length-($m+1$) vector is obtained by adding a coefficient in
the length-$m$ vector. 
Voters have complete tie-free preferences over the candidates, and 
a candidate ranked $i$th by a voter receives a score of $\alpha_i$ 
from that voter. The winners are the candidates with the highest
score.  

We are interested in determining, for all scoring rules at once, 
which of them give rise to easy computational problems and which
of them lead to hard problems. Theorems of that form are
known as dichotomy theorems.
For weighted scoring rules, in which each voter has a weight $w$ and
counts as $w$ regular voters,
there are dichotomy theorems for
all standard manipulative actions: manipulation~\cite{hem-hem:j:dichotomy},
bribery~\cite{fal-hem-hem:j:bribery}, and 
control~\cite{DBLP:journals/jair/FaliszewskiHH15}. 
The arguably more natural unweighted case is much harder to analyze
(since in the unweighted case we can only get hardness when the number of
candidates is unbounded, whereas in the weighted case hardness already
occurs with a fixed number of candidates; since weighted dichotomy theorems
typically look at a fixed number of candidates,
the results for the unweighted cases do not at all follow from the
results for the weighted cases). 
Despite the prevalence of scoring
rules, there are only two dichotomy theorems for the unweighted case, 
namely for the possible winner
problem~\cite{bet-dor:j:possible-winner-dichotomy,bau-rot:j:possible-winner-dichotomy-final-step}
and
for the constructive control by adding voters (CCAV)
problem~\cite{HemaspaandraHemaspaandraSchnoor-CCAV-AAAI-2014}.

In this paper, we look at bribery and manipulation for unweighted
scoring rules, and, since bribery can be viewed as deleting voters
followed by a manipulation, we also look at 
the constructive control by deleting voters (CCDV) problem.

For manipulation, we show that for all scoring rules
with a constant number of different coefficients, manipulation is in P.
This subsumes all known polynomial-time results for unweighted manipulation
for scoring rules. We conjecture that there is no dichotomy theorem for manipulation.

For bribery and CCDV, we obtain a dichotomy theorem for
pure scoring rules. In particular, we show exactly when these
problems are easy (in P) and that they are hard (NP-complete) 
in all other cases. Interestingly,
our characterization is the ``dual'' of the CCAV characterization in the following sense:
For every scoring rule $f$, 
the complexity of $f$-CCDV (and of $f$-bribery)
is the same as for $\dual{f}$-CCAV, where
$\dual{f}$ is obtained from $f$ by multiplying each entry in a scoring
vector by $-1$, and reversing the order of the vector.

These results are quite surprising: 
CCDV has less structure to encode hard problems into it than CCAV, but we still obtain the same complexity characterization (modulo duality). On the other hand, bribery can be seen as the combination of CCDV and manipulation, but the complexity is the same as for CCDV.
However, in another sense bribery behaves very differently from CCDV: The complexity of bribery changes from polynomial-time solvable to NP-complete by small changes in the definition of the problem, while the complexity of the former is much more robust.

The structure of the paper is as follows: In Section~\ref{sect:prelim}, we introduce relevant definitions, including the specific problems we study in this paper. In Section~\ref{sect:results:manipulation}, we state our results on manipulation. Section~\ref{sect:ccdv bribery dichotomy} contains our dichotomy result for CCDV and bribery. Our individual complexity results for CDDV and bribery can be found in Sections~\ref{sect:results:ccdv} and~\ref{sect:bribery}, respectively. We conclude with open questions in Section~\ref{sect:open questions}. 
All proofs not contained in the main paper can be found in the appendix.

\section{Preliminaries}\label{sect:prelim}

An \emph{election} consists of a non-empty, finite set of candidates and a finite set of voters. Each voter is identified with her vote, which is simply a linear order on the set of candidates. An \emph{election system} or \emph{voting rule} is a rule that, given an election, determines the set of candidates who are winners of the election according to this rule. A \emph{scoring vector} for $m$ candidates is simply a vector $(\alpha_1,\dots,\alpha_m)$ of integer coefficients, where $\alpha_i\ge\alpha_{i+1}$ for all $1 \leq i < m$. Such a vector defines a voting rule for elections with $m$ candidates by simply awarding, for each vote in the election, $\alpha_i$ points to the candidate ranked in the $i$-th position of this vote, and defining the candidates with the most points to be the winners of the election. A \emph{scoring rule} is an election system that for each number of candidates applies an appropriate scoring vector. Such a system can be described by a \emph{generator}, which is a function $f$ such that for each $m\in\mathbb N$, $f(m)$ is a scoring vector for $m$
candidates. Well-known scoring rules are Borda (using $f(m)=(m-1,m-2,\dots,1,0)$), $k$-approval (using $f(m)=(\underbrace{1,\dots,1}_{k\mathtext{ many}},0,\dots,0)$) and $k$-veto (using $f(m)=(1,\dots,1,\underbrace{0,\dots,0}_{k\mathtext{ many}})$) for natural numbers $k$. For readability, we usually identify a generator with the election system it defines.

To capture that the elections for different numbers of candidates should use ``similar'' generators, we use the following notion \cite{bet-dor:j:possible-winner-dichotomy}: A generator $f$ as above is \emph{pure}, if for all $m\ge1$, the vector $f(m)$ can be obtained from the vector $f(m+1)$ by removing one coefficient from the sequence. In~\cite{HemaspaandraHemaspaandraSchnoor-CCAV-AAAI-2014}, it is shown that the restriction to the set of pure generators with rational numbers covers all generators in a large and reasonable class.

We define standard manipulative actions:
Manipulation~\cite{bar-tov-tri:j:manipulating,con-lan-san:j:when-hard-to-manipulate},
bribery~\cite{fal-hem-hem:j:bribery}, and 
control~\cite{bar-tov-tri:j:control},
for generators.

\begin{restatableDefinition}
  For a generator $f$, the \emph{constructive control problem for $f$ by deleting voters}, $f$-CCDV, is the following problem: Given a set of voters $V$ over a set of candidates $C$, a
  candidate $p\in C$ and a number $k$, is there a subset $V'\subseteq V$ with
  $\card{V'}\leq k$ such that $p$ is a winner of the election if the
  votes in  $V - V'$ are evaluated using $f$?
\end{restatableDefinition}

A similar problem, the \emph{constructive control problem for $f$ by adding voters}, called $f$-CCAV, asks whether $p$ can be made a winner by adding to $V$ at most $k$ voters from a given set of so-called unregistered voters.
In the \emph{manipulation problem for $f$}, we are given a set $V$ of
nonmanipulative voters and a set of manipulators, and we
ask whether $p$ can be made a winner by setting the votes of the manipulators,
with no restriction on how these votes can be chosen.
Finally, the \emph{bribery problem} for $f$ asks whether $p$ can be made
a winner by replacing up to $k$ votes in $V$ with the same number of arbitrary votes.
Clearly, for every polynomial-time uniform generator $f$, the problems $f$-CCDV, $f$-bribery, and $f$-manipulation are in \NP.

Two scoring vectors $(\alpha_1,\dots,\alpha_m)$ and $(\beta_1,\dots,\beta_m)$ are \emph{equivalent} if they describe the same election system, i.e., if for any election, they lead to the same winner set. It is easy to see~\cite{HemaspaandraHemaspaandraSchnoor-CCAV-AAAI-2014} that this is the case if and only if there are numbers $\gamma>0$ and $\delta$ such that for each $i$, $\beta_i=\gamma\alpha_i+\delta$. We say that $f_1$ and $f_2$ are \emph{ultimately equivalent} if $f_1(m)$ and $f_2(m)$ are equivalent for all but finitely many $m$. It is easy to see that in this case, CCDV, bribery, and manipulation have the same complexity for $f_1$ and $f_2$.

For algorithms, we need the function $f$ is efficiently computable. A generator $f$ is polynomial-time uniform if $f(m)$ can be computed in polynomial time, given $m$ in unary. (Given $m$ in binary, polynomial time would not suffice to even write down a sequence of $m$ numbers.)
For the remainder of this paper, a generator is always a polynomial-time uniform pure generator with rational coefficients.

\section{Manipulation}\label{sect:results:manipulation}

Our main result on manipulation is the following: Every generator $f$ for which there is a fixed, finite upper bound on the number of coefficients that are used for any number of candidates has a polynomial-time solvable manipulation problem. We mention that this result also holds for generators that are not pure (but still are polynomial-time uniform). We note that the special cases where $f$ generates $k$-approval or $k$-veto were shown in~\cite{pro-ros-zuc:j:borda}.

\begin{restatable}{restatableTheorem}{theoremmanp}\label{theorem:man-p} Let $f$ be a generator such that there is a constant $c$ such that for each number $m$ of candidates, at most $c$ different coefficients appear in the vector $f(m)$. Then $f$-manipulation can be solved in polynomial time.
\end{restatable}

We give a proof sketch for a simple special case of the theorem, namely generators $f$ of the form $f=(0,\dots,0,-\beta,-\alpha)$.
With great care, this proof sketch generalizes to the general case.

\begin{proof}(Sketch)
Consider a preferred candidate $p$, a set of candidates $C=\set{c_1,\dots,c_m,p}$, a surplus $\surpl c$ for each $c\in C$ (i.e., the value $\score{c}-\score{p}$, which can easily be computed from the election instance, as $f$ is polynomial-time uniform), and a set of $k$ manipulators. Clearly, we can assume that
all manipulators will vote $p$ first.

The obvious greedy approach of having a manipulator rank a
candidate with the largest surplus last won't always work:
If $\beta = 2$, $\alpha = 3$, the surplus of $c_1$ is 4,
the surplus of $c_2$ and $c_3$ is 3, and we have two manipulators,
the only successful manipulation is to have the manipulators vote
$\cdots > c_1 > c_2$ and $\cdots > c_1 > c_3$ and so we cannot put
$c_1$ last.

If there is a successful manipulation, then for all $i\in\set{1,\dots,m}$, there are
numbers $x_i$ and $y_i$ (the number of times $c_i$ is ranked next to last / last
by a manipulator) such that:
\begin{enumerate}
\item $x_i + y_i \leq k$,
\item $\sum_{1 \leq i \leq m} x_i = k$,
\item $\sum_{1 \leq i \leq m} y_i = k$, and
\item $\surpl{c_i} - \beta x_i - \alpha y_i \leq 0$.
\end{enumerate}

We derive an 
algorithm deciding the condition by dynamic programming. For this, we define the Boolean predicate $M$ such that
$M(k, k_\beta, k_\alpha, s_1, \ldots, s_\ell)$
is true if and only if for all $i$, $1 \leq i \leq \ell$,
there exist natural numbers $x_i$ and $y_i$ such that
\begin{enumerate}
\item $x_i + y_i \leq k$,
\item $\sum_{1 \leq i \leq \ell} x_i = k_\beta$,
\item $\sum_{1 \leq i \leq \ell} y_i = k_\alpha$, and
\item $s_i - \beta x_i - \alpha y_i \leq 0$.
\end{enumerate}
It is immediate that if there is a successful manipulation, then
$M(k, k, k, \surpl{c_1}, \ldots , \surpl{c_m})$ is true. 
It is not so easy to see that the converse holds. This is shown
by induction on $k$.
It is easy to come up with a simple ad-hoc proof for the 
simple case we are looking at here, but 
we will instead 
describe an approach that 
generalizes to the general case.

The inductive step uses the following argument.
If $M(k+1, k_\beta, k_\alpha, s_1, \ldots, s_\ell)$ is true,
let $X = \{c_i \ | \ x_i > 0\}$ and let 
$Y = \{c_i \ | \ y_i > 0\}$. Then the sequence $(X, Y)$ can be
shown to fulfill the ``marriage condition,'' which then,
by Hall's Theorem~\cite{hall}, 
implies that there
is a ``traversal,'' i.e., a sequence of distinct representatives of
this sequence of sets which then gives us a vote for one of the
manipulators. In this particular case, the traversal consists of
two distinct candidates $(c_i,c_j)$ such that $x_i > 0$ and $y_j > 0$.
Let one manipulator vote $\cdots > c_i > c_j$, subtract 1 from $x_i$ and
$y_j$, subtract $\beta$ from
$\surpl{c_i}$, and subtract $\alpha$ from from $\surpl{c_j}$.
It follows from the induction hypothesis that
the remaining $k$ manipulators can vote to make $p$ a winner.

To conclude the proof sketch, we can show by dynamic programming
that $M$ is computable in polynomial for unary $k, k_\alpha, k_\beta$
(which is sufficient to solve the manipulation problem).
This holds since:
\begin{enumerate}
\item
$M(k,k_\beta, k_\alpha)$ is true if and only if
$k_\beta = k_\alpha = 0$.
\item For $\ell \geq 1$,
$M(k,k_\beta, k_\alpha, s_1, \ldots, s_\ell)$ if and only if
there exist natural numbers $x_\ell$ and $y_\ell$ such that:
\begin{enumerate}
\item
$x_\ell + y_\ell \leq k$,
\item
$x_\ell \leq k_\beta$,
\item
$y_\ell \leq k_\alpha$,
\item
$s_\ell - \beta x_\ell - \alpha y_\ell \leq 0$, and
\item
$M(k,k_\beta - x_\ell, k_\alpha - y_\ell, s_1, \ldots, s_{\ell-1}).$
\end{enumerate}
\end{enumerate}
\end{proof}

Given Theorem~\ref{theorem:man-p} and the fact that manipulation for Borda is NP-complete~\cite{DBLP:conf/ijcai/BetzlerNW11,dav-kat-nar-wal:c:complexity-and-algorithms-for-borda},
it is natural to ask whether the manipulation problem is NP-complete for all remaining generators.
But this is extremely unlikely: 
Though our approach does not give
polynomial-time algorithms when the number of coefficients is
unbounded, it will give quasipolynomial algorithms when the
coefficients are small enough and grow slowly enough.

It is also conceivable that additional cases will be in P.
Though a general greedy approach seems unlikely (as manipulation for Borda is NP-complete),
a greedy approach for specific cases is still possible.

We conjecture that there is no dichotomy theorem for manipulation
for pure scoring rules, with 
different intermediate (between P and
NP-complete) complexities showing up.

\section{CCDV and Bribery Dichotomy}\label{sect:ccdv bribery dichotomy}

We completely characterize the complexity of $f$-CCDV and $f$-bribery for every generator $f$. For each generator $f$, these two problems are polynomial-time equivalent, and are polynomial-time solvable or NP-complete. For the cases where $f$ generates $k$-approval or $k$-veto, the complexity classification is already stated in~\cite{lin:thesis:elections}.
For CCDV, the special case where $f$ generates $1$-approval was shown
in~\cite{bar-tov-tri:j:control}.
For bribery, the
special cases where $f$ generates $1$-approval or $1$-veto were shown in~\cite{fal-hem-hem:j:bribery}.
We also note the existence of an unpublished manuscript that proves NP-hardness
for bribery for generators of the form $(\alpha, \beta, 0, \ldots, 0)$, where $\alpha$ and $\beta$ are coprimes with
$\alpha > \beta \geq 1$~\cite{car-kak-kar-woe:unpub:bribery}.

\begin{restatable}{restatableTheorem}{theoremccdvandbriberydichotomy}\label{theorem:ccdv dichotomy}
 Let $f$ be a pure, polynomial-time uniform generator. If $f$ is ultimately equivalent to one of the following generators, then $f$-CCDV and $f$-bribery can be solved in polynomial time:
  \begin{enumerate}
  \item $f_1=(1,\dots, 1, 0, 0, 0)$ ($3$-veto),
  \item $f_2=(1, 0, \dots,0)$ ($1$-approval),
  \item $f_3=(1, 1, 0, \dots, 0)$ ($2$-approval),
  \item for some $\alpha\ge\beta \ge 0$, $f_4=(0, \ldots, 0, -\beta, -\alpha)$ (this includes triviality, $1$-veto, and $2$-veto),
  \item $f_5=(2,1,\dots,1,0)$.
 \end{enumerate}
Otherwise, $f$-CCDV and $f$-bribery are NP-complete.
\end{restatable}

This dichotomy theorem mirrors the one obtained for CCAV in~\cite{HemaspaandraHemaspaandraSchnoor-CCAV-AAAI-2014} (modulo duality, see below). For the relationship of CCDV and CCAV, this implies that the difficulty of implementing ``setup votes'' (see below) does not have any influence on the complexity of our decision problems for the class of generators we study. The below proof of Theorem~\ref{theorem:0 dots 0 -gamma -beta -alpha CCDV} is an example of a non-trivial implementation of these setup votes.

In particular, our results imply that the complexity of CCDV is ``robust'' in the following sense: The complexity of CCDV does not depend on whether we add a bit to each voter stating whether she can be deleted or not. This will be made formal below in our discussion of \ccdvstar (defined below).

The situation is different for bribery:
Generalizing the bribery problem to allow marking some voters as ``unbribable'' increases the complexity to NP-complete for some generators. As an example of this phenomenon, we state the following result. (The version of bribery defined here may be of independent interest, but is only used here to highlight the differences between CCDV and bribery.)

\begin{restatableTheorem}\label{theorem:bribery with unbribable voters np complete 0 dots 0 -1 -2}
The variation of the bribery problem for $(0, \dots, 0, -1, -2)$ where each voter has a bit that states whether this
voter can be bribed or not is NP-complete.
\end{restatableTheorem}

\section{Control by Deleting Voters}\label{sect:results:ccdv}

In this section we give an overview over our proof of CCDV-part of Theorem~\ref{theorem:ccdv dichotomy}. In Section~\ref{sect:ccav and ccdv}, we show that our CCDV polynomial-time cases easily follow from a relationship to CCAV, whose complexity was studied in~\cite{HemaspaandraHemaspaandraSchnoor-CCAV-AAAI-2014}. Sections~\ref{sect:ccdv few coefficients general cases} and~\ref{sect:ccdv many coefficients} then contain our hardness results for CCDV.

\subsection{Relationship between CCAV and CCDV and CCDV polynomial time cases}\label{sect:ccav and ccdv}

CCDV and CCAV are closely related as follows: For a generator $f$, let $\dual f$ be the generator obtained from $f$ by multiplying each coefficient with $-1$ and reversing the order of the coefficients (to maintain monotonicity). Removing a vote in an $f$-election has the same effect as adding the same vote in a $\dual f$-election. Using this observation, one can show that $f$-CCDV reduces to $\dual f$-CCAV for all generators $f$. We mention that a similar relation holds
for weighted $k$-approval and $k$-veto
elections~\cite{DBLP:journals/jair/FaliszewskiHH15}.

The other direction of this relationship does not follow so easily, as there is an important difference between CCDV and CCAV: In CCAV, the set of voters is partitioned into a set of \emph{registered} voters and a set of \emph{potential} voters, where the controller's actions can only influence the potential voters. This provides the problem with additional structure, as the controller cannot modify the registered voters. We often call these registered votes, which the controller cannot influence anymore, \emph{setup votes}, as these allow us to set up the scenario of an NP-hard problem in hardness proofs.

The CCDV problem does not have a corresponding structure; here every vote may be (potentially) deleted by the controller. This makes it harder to construct the above-mentioned setup votes: Since we cannot simply ``forbid'' the controller to delete certain votes, hardness proofs for CCDV need to ``setup'' the relevant scenario with votes that are designed to be ``unattractive'' to delete for the controller.

To obtain CCDV hardness results, it is therefore natural to consider the following analog to the CCAV problem: \ccdvstar is a version of CCDV providing the additional structure that CCAV has. In \ccdvstar, the set of votes is partitioned into a set $R$ of voters that cannot be deleted, and voters $D$ that can be deleted. From the above discussion, it follows that the complexities of CCAV and \ccdvstar are related with the following duality, which, together with the polynomial-time results obtained for CCAV in~\cite{HemaspaandraHemaspaandraSchnoor-CCAV-AAAI-2014} immediately implies the polynomial-time CCDV cases of Theorem~\ref{theorem:ccdv dichotomy}.

\begin{restatable}{restatableProposition}{propccdvstarpolyequivccav}\label{prop:ccdvstar poly equiv ccav}
 For every generator $f$, $f$-\ccdvstar and $\dual f$-CCAV are polynomially equivalent.
\end{restatable}

Proposition~\ref{prop:ccdvstar poly equiv ccav} is not completely trivial, since the reductions must convert election instances maintaining the relative points of the candidates. However, this is done using standard constructions.

From Theorem~\ref{theorem:ccdv dichotomy} and the results in~\cite{HemaspaandraHemaspaandraSchnoor-CCAV-AAAI-2014}, it follows that 
$f$-CCDV and $f$-\ccdvstar always have the same complexity. In fact, our algorithms for CCDV and \ccdvstar do not take the structure of the ``registered'' votes into account, but can work with scores for the candidates that do not come from any set of votes. For bribery, the situation is quite different, see the above Theorem~\ref{theorem:bribery with unbribable voters np complete 0 dots 0 -1 -2}.

All polynomial-time cases for CCDV follow from the above relationship in a straight-forward manner. This is not surprising, since CCDV in the above-discussed sense has less structure than CCAV, and thus easiness results for CCAV translate to (dual) CCDV. The interesting part of our dichotomy is the converse: If CCAV is NP-hard for some generator $f$, then CCDV is hard for $\dual f$ as well. 

A natural approach for the proof of the dichotomy theorem, suggested by Proposition~\ref{prop:ccdvstar poly equiv ccav}, is to show that $f$-\ccdvstar always reduces to $f$-CCDV. While this does in fact follow, proving a generic reduction from $f$-\ccdvstar to $f$-CCDV for all generators $f$ seems to be difficult, due to the additional structure provided by \ccdvstar. 

Our proof of the CCDV part of Theorem~\ref{theorem:ccdv dichotomy} therefore uses a case distinction to obtain $f$-CCDV hardness for each remaining pure, polynomial-time uniform generator $f$.

We note that due to the relationship between CCDV and CCAV, all CCAV-hardness results in~\cite{HemaspaandraHemaspaandraSchnoor-CCAV-AAAI-2014} easily follow from the results obtained in the current paper. However, our proofs make use of the results and proofs from~\cite{HemaspaandraHemaspaandraSchnoor-CCAV-AAAI-2014}.

\subsection{CCDV hardness: ``few coefficients''}\label{sect:ccdv few coefficients general cases}

We first consider generators with ``few'' different coefficients, i.e., generators of the form $f=(\alpha_1,\alpha_2,\alpha_3,\dots,\alpha_3,\alpha_4,\alpha_5,\alpha_6)$ for rationals $\alpha_1,\dots,\alpha_6$.\footnote{This only uniquely defines $f$ for elections with at least $6$ candidates, however $f$ is uniquely defined up to ultimate equivalence, which is sufficient for the complexity analysis.}
Using equivalence-preserving transformations, we can assume that all $\alpha_i$ are nonnegative integers, and that their greatest common divisor is $1$. Note that a generator of this form is trivially polynomial-time uniform. 

\subsubsection{Reductions from \ccdvstar}\label{sect:ccdv hardness:generic ccdvstar reductions}

A general reduction from $f$-CCDV to $f$-\ccdvstar does not seem feasible, as discussed above. However, there are cases where hardness of $f$-\ccdvstar leads to hardness of $f$-CCDV with a direct proof. 
The following two results (Theorems~\ref{theorem:alpha 0 dots 0 alpha < beta ccdv} and~\ref{theorem:alpha 0 dots 0 beta < alpha ccdv}) are proven in this way.

\begin{restatable}{restatableTheorem}{theoremccdvhardnessalphazweodotsalphasmallerbeta}
  \label{theorem:alpha 0 dots 0 alpha < beta ccdv}
  Let $f=(\alpha,0,\dots,0,-\beta)$ be a generator with $1\leq\alpha<\beta$. Then $f$-CCDV is NP-complete.
\end{restatable}

\begin{proof}(Sketch)
 From~\cite{HemaspaandraHemaspaandraSchnoor-CCAV-AAAI-2014}, we know that $\dual f$-CCAV is NP-complete, Proposition~\ref{prop:ccdvstar poly equiv ccav}, then implies that $f$-\ccdvstar is NP-complete as well. We  show that $f$-\ccdvstar reduces to $f$-CCDV. Given an instance of $f$-\ccdvstar, we convert it into an equivalent instance of $f$-CCDV by replacing the undeletable votes $R$ with votes that
 \begin{enumerate}
  \item\label{enum label:setup votes need same relative points} result in the same relative points as the votes in $R$, and 
  \item are not deleted in a successful CCDV action.
 \end{enumerate}

 We denote a vote $c_1>c_2>\dots>c_{n-1}>c_n$ simply as $c_1>c_n$ (the rest is irrelevant). With some light preprocessing, we can assume that no deletable vote has $p$ in the first or last position, this allows us to compute the number $\score p$ of points that $p$ will have after the delete action. Similarly, we can assume that $\score p=N_p\alpha$ for some $N_p\ge2$. 
 
 Satisfying point~\ref{enum label:setup votes need same relative points} above boils down to add, for an arbitrary candidate $c\neq p$, votes that let $c$ gain $\alpha$ points against $p$, and which will not be deleted. This is done as follows: We add dummy candidates $d_1,\dots,d_\ell$ (for a suitably chosen number $\ell$) and a single vote $c>d_1$, letting $c$ gain $\alpha$ points relative to $p$. To  ensure that the vote $c>d_1$ cannot be removed, we add votes setting up the scores as follows:

 \begin{itemize}
   \item Each $d_i$ for $1\leq i\leq\ell-1$ ties with $p$,
   \item the only way to make $d_i$ lose points (relative to $p$) is to remove votes $d_i>d_{i+1}$, which then lets $d_{i+1}$ gain points (relative to $p$).
 \end{itemize}
 
 Hence removing the vote $c>d_1$, which lets $d_1$ gain $\beta$ points relative to $p$ requires the controller to remove votes of the form $d_1>d_2$, which each lets $d_2$ gain $\beta$ points. This process continues for $d_i$ with $i\ge 2$. Thus, removing $c>d_1$ triggers a ``chain'' of additional removals---more than the budget allows. The numbers of votes needed to setup grows exponentially in the number of steps. However, since the controller can only remove a polynomial number of votes, we only require logarithmically many steps, yielding a polynomial construction.
 
 Constructing the actual set of votes that results in the above scores and satisfies the two points above is nontrivial, the construction is in fact the main technical difficulty in the proof.
\end{proof}

The following result is shown similarly, the difference is that instead of logarithmically many steps of an exponentially growing construction, here we apply a simpler linear process.

\begin{restatable}{restatableTheorem}{theoremalphazerodotszerobetasmalledalphaccdv}\label{theorem:alpha 0 dots 0 beta < alpha ccdv}
 Let $f=(\alpha,0,\dots,0,-\beta)$ be a generator with $\alpha>\beta\ge1$. Then $f$-CCDV is NP-complete.
\end{restatable}

\subsubsection{Reductions by inspection of the CCAV reduction}\label{sect:ccdv hardness by inspecting ccav reduction}

Similarly to the preceding Section~\ref{sect:ccdv hardness:generic ccdvstar reductions}, the results in this section are proved by a reduction from $f$-\ccdvstar to $f$-CCDV. However, while the reductions above were ``generic'' (reducing from an arbitrary \ccdvstar-instance), we now start with instances of $f$-\ccdvstar produced by the hardness proof of $f$-\ccdvstar. Therefore, we do not need to construct ``setup votes'' that implement any possible given set of scores, but only need to achieve exactly the points used in the hardness proof of $\dual f$-CCAV in~\cite{HemaspaandraHemaspaandraSchnoor-CCAV-AAAI-2014}.

In the following theorem, this is a significant advantage, as here the ``setup votes'' grant more points to the preferred candidate than to the remaining candidates. Therefore, it is easy to construct these votes in such a way that the controller has no incentive to delete them. 

The proof of the theorem also illustrates the relationship between hardness results for CCAV (as obtained in~\cite{HemaspaandraHemaspaandraSchnoor-CCAV-AAAI-2014}), and the hardness results for CCDV and bribery we obtain in the current paper: The proof of Theorem~\ref{theorem:0 dots 0 -gamma -beta -alpha CCDV} below uses the reduction of the corresponding hardness result for CCAV in~\cite{HemaspaandraHemaspaandraSchnoor-CCAV-AAAI-2014} as a starting point, but is technically more involved. We will later re-use parts of the following construction to obtain the corresponding hardness result for bribery as well, in the later Theorem~\ref{theorem:0 dots 0 -gamma -beta -alpha bribery}.

\begin{restatable}{restatableTheorem}{theoremzerodotszerominusgammaminusbetaminusalphaccdv}\label{theorem:0 dots 0 -gamma -beta -alpha CCDV}
 Let $f=(\alpha_3,\dots,\alpha_3,\alpha_4,\alpha_5,\alpha_6)$ with $\alpha_3>\alpha_4>\alpha_6$. Then $f$-CCDV is NP-complete.
\end{restatable}

\begin{proof}
 For the proof, we equivalently write $f$ as $f=(0,\dots,0,-\gamma,-\beta,-\alpha)$ with $0<\gamma<\alpha$.
 Then, $\dual f=(\alpha,\beta,\gamma,0,\dots,0)$. From~\cite{HemaspaandraHemaspaandraSchnoor-CCAV-AAAI-2014}, it follows that $\dual f$-CCAV is NP-complete, and hence, due to Proposition~\ref{prop:ccdvstar poly equiv ccav}, it suffices to show that $f$-\ccdvstar reduces to $f$-CCDV. 
 
 Therefore, let an $f$-\ccdvstar instance with undeletable votes $R$, deletable votes $D$, preferred candidate $p$, and budget $\ell$ be given. From the proof of Proposition~\ref{prop:ccdvstar poly equiv ccav}, we can assume that this instance is obtained from the hardness proof of $\dual f$-CCAV as follows:
 
 \begin{itemize}
  \item the votes in $D$ are exactly the votes available for addition in the CCAV instance, with the order of candidates reversed,
  \item the relative points gained by the candidates from the votes in $R\cup D$ are the same as the points of the candidates in the CCAV instance (before the addition of votes by the controller).
 \end{itemize}

The hardness proof of $\dual f$-CCAV uses a reduction from 3DM.
3DM is the following problem: Given a multiset
$M\subseteq X\times Y\times Z$ with $X$, $Y$ and $Z$ pairwise disjoint sets of equal size such that each $s\in X\cup Y\cup Z$ appears in exactly $3$ tuples of $M$, decide whether there is a set $C\subseteq M$ with $\card C=\card X$ such that each $s\in X\cup Y\cup Z$ appears in some tuple of $C$ (we also say that $C$ \emph{covers} $s$).
From the problem definition, it follows that $\card M=3\card X$.
The condition that each $s$ appears in exactly $3$ tuples is not standard;
we prove in the appendix
remains NP-complete.

Let $M\subseteq X\times Y\times Z$ be an instance of 3DM. Following the notation used in~\cite{HemaspaandraHemaspaandraSchnoor-CCAV-AAAI-2014}, we set $n=\card M$ and $k=\card X$. Since $\card M=3\card X$ in every 3DM instance, it follows that $n=3k$. The hardness proof of $\dual f$-CCAV, translated to the CCDV setting (i.e., we present the votes as in the CCDV instance---as reversals of votes from the CCAV instance), constructs the following situation:
 
 \begin{itemize}
  \item the candidate set is $\set p\cup X\cup Y\cup Z\cup\set{S_i,S_i'\suchthat S_i\in M}$,
  \item for each $S_i=(x,y,z)\in M$, the following votes are available for deletion (we only list the candidates gaining non-zero points from the vote):
  \begin{itemize}
    \item $\dots > S_i > p > x$
    \item $\dots > S_i > p > y$
    \item $\dots > S_i' > p > z$
    \item $\dots > S_i' > p > S_i$
  \end{itemize}
  \item the relative scores resulting from the registered voters of the CCAV instance (i.e., the undeletable voters of the \ccdvstar instance) are as follows:
  \begin{itemize}
    \item $\scorefinal p=\alpha+2\gamma$,
    \item $\scorefinal c=(n+2k)\beta+2\gamma$ for each $c\in X\cup Y\cup Z$,
    \item $\scorefinal{S_i}=(n+2k)\beta+\min(\alpha,2\gamma)$,
    \item $\scorefinal{S_i'}=(n+2k)\beta+\alpha+\gamma$.
  \end{itemize}
 \end{itemize}
 
 From the above votes introduced for the 3DM-elements, the candidates gain the following initial points---recall that each $c\in X\cup Y\cup Z$ appears in exactly $3$ triples from $M$, and, since $\card M=n=3k$, and there are $4$ votes introduced for every element in $M$, there are exactly $12k$ deletable votes introduced above:
 
 \begin{itemize}
   \item $\scorethreedm{p}=-12k\beta$, since $p$ gains $-\beta$ points in each of the $12k$ votes,
   \item $\scorethreedm{S_i}=-2\gamma-\alpha$, since $S_i$ gains $0$ points in all of the votes introduced for other elements $S_j\neq S_i\in M$, and gains $-\gamma$ points in $2$ of the votes above, and $-\alpha$ points in one of the vote,
   \item $\scorethreedm{S_i'}=-2\gamma$ analogously,
   \item $\scorethreedm{c}=-3\alpha$ for each $c\in X\cup Y\cup Z$, since each $c$ appears in exactly $3$ triples $S_i$ from $M$.
 \end{itemize}
 
 For each candidate $x$, the undeletable votes from the \ccdvstar instance thus let her gain exactly $\scorefinal x-\scorethreedm x$ points (modulo an offset added to the points of all candidates, since the CCAV reduction relies on an implementation lemma
 that only fixes the relative points of each candidate). The scores that the undeletable votes implement are therefore as follows (recall that $n=3k$):
 
 \medskip
 
 \begin{tabular}{ll}
  candidate $x$        & $\scorefinal x-\scorethreedm x$ \\ \hline
  $p$                  & $12k\beta+\alpha+2\gamma$ \\
  $c\in X\cup Y\cup Z$ & $5k\beta+3\alpha+2\gamma$ \\
  $S_i$                & $5k\beta+\alpha+2\gamma+\min(\alpha,2\gamma)$ \\
  $S_i'$               & $5k\beta+\alpha+3\gamma$
 \end{tabular}

 \medskip
 
 Since we only need to implement the relative scores among the candidates, it is enough to consider the points the candidates have to gain/lose relative to $p$. These are as follows (clearly, $p$ does not gain or lose any points relative to herself):
 
 \medskip

 \begin{tabular}{lll}
  candidate $x$        & points $x$ needs to lose (result) \\ \hline
  $c\in X\cup Y\cup Z$ & $7k\beta-2\alpha$ \\
  $S_i$                & $7k\beta-\min(\alpha,2\gamma)$ \\
  $S_i'$               & $7k\beta-\gamma$
 \end{tabular}
 
 \medskip
 
 To achieve this, we first introduce three dummy candidates $d_1$, $d_2$, and $d_3$ (by placing them in the $0$-point segment of all present votes), and add sufficiently many votes voting all relevant candidates ahead of the dummy candidates. This lets the dummies lose $\alpha$, $\beta$, or $\gamma$ points relatively to the other candidates; we do this often enough to ensure that the dummy candidates cannot win the election. Using these dummies, we can easily let a relevant candidate $c\neq p$ lose $\delta\in\set{\alpha,\beta,\gamma}$ points, by using a vote placing $c$ in the position worth $-\delta$ points, filling the other two positions out of the last three with dummy candidates, and voting the remaining candidates (including $p$) in the positions gaining $0$ points. Such a vote will never be deleted by the controller, since it has $p$ in the first position.
 
 If $1=\delta\in\set{\alpha,\beta,\gamma}$, then the required amount of points can easily be achieved by repeatedly losing $\delta=1$ points as described above. Hence assume that, in particular, $\beta\ge2$. For the candidate $c$, we proceed as follows:
 
 \begin{itemize}
  \item We add $\beta-2$ many votes letting $c$ lose $\alpha$ point each.
  \item After this, $c$, still needs to lose $(7k-\alpha)\beta$ points, which we can achieve by using $7k-\alpha$ many votes letting $c$ lose $\beta$ points as described above (we can without loss of generality assume that $7k\ge\alpha$).
 \end{itemize}
 
 For candidates $S_i$ and $S_i'$, we proceed analogously. 
 
 It hence follows that $p$ can be made a winner by removing at most $\ell$ votes in the \ccdvstar instance if and only if this is possible in the constructed $f$-CCDV instance.
\end{proof}

The following result uses a similar, but technically more involved argument.

\begin{restatable}{restatableTheorem}{theoremzerominusalphafiveminusalphaonedotsminusalphaoneccdvhardness}
 \label{theorem:0 minus alpha5 minus alpha1 dots minus alpha1 CCDV hardness}
 Let $f=(\alpha_1,\alpha_2,\alpha_3,\dots,\alpha_3)$ with $\alpha_1>\alpha_2>\alpha_3$. Then $f$-CCDV is NP-complete.
\end{restatable}

\subsubsection{Direct Proofs}\label{sect:ccdv hardness by direct 3DM reduction}

The remaining cases are shown with a direct proof (reducing from a variant of three dimensional matching, 3DM); they are in part similar to the hardness proofs of CCAV in~\cite{HemaspaandraHemaspaandraSchnoor-CCAV-AAAI-2014}.

\begin{restatable}{restatableTheorem}{theoremcollectionoffewcoefficientsCCDVhardnesscases}\label{theorem:collection of few coefficients CCDV hardness cases}
 $f$-CCDV is NP-complete in the following cases:
 \begin{enumerate}
  \item\label{few coefficients CCDV:ccdv hardness alpha1 alpha2 0 dots 0 minus alpha4 minus alpha5 minus alpha6} $f=(\alpha_1,\alpha_2,\alpha_3,\dots,\alpha_3,\alpha_4,\alpha_5,\alpha_6)$ with $\alpha_1>\alpha_3>\alpha_5$.
  \item\label{few coefficients CCDV:ccdv alpha1alpha2>alpha3>alpha6}$f=(\alpha_1,\alpha_2,\alpha_3,\dots,\alpha_3,\alpha_6)$ with $\alpha_1,\alpha_2>\alpha_3>\alpha_6$.
  \item\label{few coefficients CCDV:ccdv hardness alpha1 > alpha 3 > alpha 4} $f=(\alpha_1,\alpha_2,\alpha_3,\dots,\alpha_3,\alpha_4,\alpha_5,\alpha_6)$ with $\alpha_1>\alpha_3>\alpha_4$.
 \end{enumerate}
\end{restatable}

\subsection{CCDV hardness: ``Many'' coefficients}\label{sect:ccdv many coefficients}

In Section~\ref{sect:ccdv few coefficients general cases}, we have covered all generators of the form $f=(\alpha_1,\alpha_2,\alpha_3,\dots,\alpha_3,\alpha_4,\alpha_5,\alpha_6)$. Generators \emph{not} of this form must satisfy $\alpha^m_3>\alpha^m_{m-3}$ for some value $m$. We now prove that the CCDV problem is NP-hard for all generators satisfying this condition. The proof is similar to the corresponding result in~\cite{HemaspaandraHemaspaandraSchnoor-CCAV-AAAI-2014}. 

An important observation is that when $\alpha^m_3>\alpha^m_{m-3}$ holds for some $m$, then purity of our generators implies that the condition also holds for any $m'\ge m$. Clearly, if $m\ge6$ is a multiple of $3$ and $\alpha^m_3>\alpha^m_{m-3}$, then one of $\alpha^m_3>\alpha^m_{\frac23m}$ and $\alpha^m_{\frac23m}>\alpha^m_{m-3}$ must hold. In the first case, the generator---for this number $m$ of candidates---behaves in a similar way as $3$-approval in the sense that there are three positions in the vote that let the candidates gain more points than a ``large'' number of positions later in the sequence of coefficients. In the latter case, the generator (for this number $m$) behaves similarly to $4$-veto, as there are four ``bad'' positions in the votes.

Therefore, a generator satisfying $\alpha^m_3>\alpha^m_{m-3}$ behaves, for every $m'\ge m$, similarly to $3$-approval or to $4$-veto. However, the behavior can be different for different values for $m$. Therefore, our proof of NP-completeness for generators satisfying this condition uses an ``adaptive'' reduction from the problem 3DM, which, given a 3DM instance, constructs either a reduction exploiting the ``$3$-approval likeness'' or the ``$4$-veto likeness'' of the generator, depending on the size of the instance (which linearly corresponds to the number of candidates in the constructed election instance). This gives the following result:

\begin{restatable}{restatableTheorem}{theoremmanydifferentcoefficients}
\label{theorem:many different coefficients}
 Let $f$ be a polynomial-time pure generator such that $\alpha^m_3>\alpha^m_{m-3}$ for some $m$. Then $f$-CCAV is NP-hard.
\end{restatable}

\section{Bribery}\label{sect:bribery}

Bribery is closely related to both CCDV and manipulation:
Bribing $k$ voters can be viewed as deleting $k$ voters followed by
adding $k$ manipulation voters. We thus can use our results obtained
for CCDV and manipulation in Sections~\ref{sect:results:manipulation} and~\ref{sect:results:ccdv}
as a starting point to obtain a complexity classification of the bribery
problem.

However, it is not necessarily the
case that an optimal bribery consists of an optimal deletion followed by
an optimal manipulation (see Example~\ref{e:bribery}),
and so it is possible for bribery to be
hard while the manipulation and deletion of voters problems are easy.
However, we will show that for every pure scoring rule $f$,
$f$-Bribery is polynomial-time solvable if and only if $f$-CCDV is, though
the proofs for Bribery are more (and sometimes much more)
involved.

We also obtain an interesting relationship between the complexities of manipulation and bribery: From Theorem~\ref{theorem:man-p}, it follows that every ``few coefficients'' case leads to a polynomial-time solvable manipulation problem. From Theorem~\ref{theorem:ccdv dichotomy}, we know that, for CCDV, \emph{only} such cases can be solved in polynomial time (unless $\PTIME=\NP$). Therefore, we obtain the following corollary:

\begin{restatableCorollary}
  Let $f$ be a polynomial-time uniform pure generator. Then $f$-manipulation reduces to $f$-CCDV and $f$-manipulation reduces 
to $f$-bribery.
\end{restatableCorollary}

\subsection{Bribery Polynomial-Time Cases}\label{sect:polynomial time bribery}

Our first bribery algorithm (for a generator equivalent to $(2,1,\dots,1,0)$)
uses a reduction to network flow.

\begin{restatable}{restatableTheorem}{theorembriberyonezeroesminusoneptime}\label{theorem:100-1 bribery in ptime}
$(1,0,\dots,0,-1)$-bribery is in polynomial time.
\end{restatable}

\begin{proof}(Sketch)
There are three types of voters. Let $V_1$ be the set of voters that rank $p$
last, let $V_2$ be the set of voters that rank $p$ neither first nor last,
and let $V_3$ be the set of voters that rank $p$ first.
In this case, we can assume without loss of generality that we bribe
as many $V_1$ voters as possible, followed by as many $V_2$ voters
as possible.  We never have to bribe $V_3$ voters.
All bribed voters will put $p$ first, so we also know $p$'s score after
bribery.

The hardest case is the one where we bribe all $V_1$ voters and some 
$V_2$ voters.
We view bribery as deletion followed by manipulation.
Delete all $V_1$ voters. In $V_2$,
deleting a voter $a > \dots > b$ corresponds to transferring a point
from $a$ to $b$.   After deleting $k$ voters, the deleted voters will
be bribed to rank $p$ first and to rank some other candidate last.
After deleting $V_1$, for every candidate $c$,
$\score{c} = \scoresub{V_2\cup V_3}{c}$, i.e.,
the score of $c$ in $V_2 \cup V_3$.
For every $V_2$ voter $a > \dots > b$ that is deleted, transfer
one point from $a$ to $b$.  For every bribe $p > \dots > d$,
delete a point from $d$.  There are exactly $k$ bribes.
After bribery, $\score{p} = \card{V_3} + k$ and the score
of every other candidate needs to be at most
$\score{p} = \card{V_3} + k$.

It is not too hard to see that this problem can be translated in
min-cost network flow problem, in a similar,
though somewhat more complicated, way as in 
the CCAV algorithm for the same generator
in~\cite{HemaspaandraHemaspaandraSchnoor-CCAV-AAAI-2014}.
\end{proof}

We now state our second bribery result.

\begin{restatable}{restatableTheorem}{theorembriberyzeroesbetaalphainptime}\label{theorem:bribery zeroes -beta -alpha in ptime}
Let $\alpha \geq \beta \geq 0$. Then bribery for $(0, \ldots, 0, -\beta, -\alpha)$ is solvable in polynomial time. 
\end{restatable}

\begin{proof}(Sketch)
As in the proof of
Theorem~\ref{theorem:100-1 bribery in ptime},
we partition $V$ into $V_1$, $V_2$, and $V_3$.
$V_1$ consists of all voters in $V$ that rank $p$ last,
$V_2$ consists of all voters in $V$ that rank $p$ second-to-last,
and $V_3$ consists of the remaining voters.
In the proof of Theorem~\ref{theorem:100-1 bribery in ptime},
it was important that we never had to bribe $V_3$ voters.
This is not always the case here, as 
shown in the example
below. That also means that this case is very different from CCDV, since in 
CCDV we never have to delete $V_3$ voters.
It also shows that an optimal bribery is not always an optimal
deletion followed by an optimal manipulation.

\begin{example}
\label{e:bribery}
This example shows that we sometimes need to bribe $V_3$ voters.
We will use the scoring rule $(0, \ldots, 0, -1, -3)$.
Let $C = \{p, a, b, c, d, e, f\}$ and let $V$ consist of
the following votes:

\begin{tabular}{ccc}
$\cdots > p > a$ &
 $\cdots > e > f$ &
$\cdots > f > e$\\
$\cdots > p > b$ & 
 $\cdots > e > f$ &
$\cdots > f > e$\\
$\cdots > p > c$
\end{tabular}

Then $\score{p} = \score{a} = \score{b} = \score{c} = -3$, 
$\score{d} = 0$, and
$\score{e} = \score{f} = -8$.

We can make $p$ a winner by bribing one of the $V_3$ voters
to vote $p > \cdots > d$.  But
it is easy to see that we can not make $d$ a winner by bribing a $V_2$
voter, wlog, the voter voting $\cdots > p > a$, since
in the bribed election, the score of $p$ will be at most
$-2$, and so both $a$ and $d$ must be in the last position of
the bribed voter.
\end{example}

Though we may need to bribe $V_3$ voters, we can show that we
never need to bribe more than a constant number of $V_3$ voters.
This is crucial in obtaining a a polynomial-time bound on a
dynamic programming
approach similar to, but more complicated than, the one 
in the proof
for Theorem~\ref{theorem:man-p}.
\end{proof}

Together with Lin's results on the complexity of bribery for $k$-veto and $k$-approval election systems in~\cite{lin:thesis:elections},
the above results prove all polynomial-time bribery cases of our main result, Theorem~\ref{theorem:ccdv dichotomy}. In the remainder of Section~\ref{sect:bribery}, we therefore discuss our NP-hardness results for bribery.

\subsection{Bribery Hardness Approach}\label{sect:bribery hardness from ccdv hardness}

Our dichotomy results imply that each generator $f$ for which CCDV is \NP-hard also has an \NP-hard bribery problem. Proving this via a generic reduction from CCDV to bribery does not seem to be easy: Even though bribery can be seen as CCDV followed by manipulation, solving a bribery instance is not the same as first finding an optimal deletion of votes and then performing an optimal manipulation. Therefore, hardness of $f$-bribery does not easily follow from hardness of $f$-CCDV. We briefly discuss a proof strategy to obtain a hardness proof for $f$-bribery from a hardness proof of $f$-CCDV. 

A key difficulty in the construction of a bribery hardness proof is that the manipulation action of the controller allows her more freedom than her delete action: For the latter, the reduction controls the available votes, whereas the manipulation action can use arbitrary permutations of the candidates. However, we can always assume that the bribed voters will vote $p$ in the first position, and will place any ``dummy'' candidates in the positions following $p$ in their votes. This often allows us to compute the score of $p$ after the bribery action.

To limit the controller's freedom in the manipulation votes, we proceed as follows: We identify a sufficiently long subsequence of the coefficients that differ by only a ``small'' amount. (Such a sequence exists by definition for the generators treated in Section~\ref{sect:ccdv few coefficients general cases}, and can be found using a pigeon-hole argument for other generators.) We then set up the points such that the ``relevant'' candidates must be placed into this ``low-variation'' sequence of the vote. This is done using ``blocking'' candidates that must be placed in low-score positions, and dummy candidates occupying the high-score positions (except for the first position, in which the bribed voters always vote $p$). This ensures that moving a candidate inside the ``low-variation'' does not make a large difference, and allows the reduction to control the possible manipulation actions very tightly.

A second difficulty is that the controller is more powerful in bribery than in CCDV: In bribery, she can perform a manipulation action in addition to her delete action. Therefore, to make it ``hard'' for the controller to find an optimal bribery action, the scores in the election instance must be ``worse'' for $p$ than in the CCDV setting. This leads to setup votes that are more attractive to delete than in the CCDV case. Similarly as in CCDV, the main technical issue is to define setup votes that both obtain the required scores and still will not be deleted by the controller, however due to the reasons above this is even more difficult for bribery as for CCDV.

The ideas outlined above allow
us to prove the bribery hardness results of Theorem~\ref{theorem:ccdv dichotomy}.

As an example for our approach to bribery hardness, we prove the following theorem, which is the ``bribery version'' of Theorem~\ref{theorem:0 dots 0 -gamma -beta -alpha CCDV}. We therefore apply our recipe to the proof of Theorem~\ref{theorem:0 dots 0 -gamma -beta -alpha CCDV}, which is the CCDV hardness result for the same generator.

\begin{restatable}{restatableTheorem}{theoremzerodotszerominusgammaminusbetaminusalphabribery}\label{theorem:0 dots 0 -gamma -beta -alpha bribery}
 Let $f=(\alpha_3,\dots,\alpha_3,\alpha_4,\alpha_5,\alpha_6)$ with $\alpha_3>\alpha_4>\alpha_6$. Then $f$-bribery is NP-complete.
\end{restatable}

\begin{proof}
 We follow the above recipe to obtain a hardness proof for $f$-bribery. Applying the recipe requires to make some changes to the construction of the CCDV hardness proof. As in the proof of Theorem~\ref{theorem:0 dots 0 -gamma -beta -alpha CCDV}, we write $f$ as $f=(0,\dots,0,-\gamma,-\beta,-\alpha)$. In particular, since we can assume that all bribed voters will vote $p$ first, the score of $p$ will not change from manipulation votes (but may, of course, change from the deleted votes).
 
 The budget for the controller is, as in the CCAV proof from~\cite{HemaspaandraHemaspaandraSchnoor-CCAV-AAAI-2014}, $n+2k$, where $n=3k$, i.e., the budget is $5k$ (here, $k$ is again the number $\card X$ from the given 3DM instance). We make the following changes to the setup in the proof of Theorem~\ref{theorem:0 dots 0 -gamma -beta -alpha CCDV}:
 
 \begin{itemize}
  \item We introduce three new distinct candidates $b_\alpha$, $b_\beta$, and $b_\gamma$ (i.e., even if $\beta=\gamma$, then $b_\beta$ is still a different candidate from $b_\gamma$). These candidates will be placed in the three ``relevant'' positions of each manipulation vote in every successful bribery action.
  \item Recall that in the proof of the CCDV hardness result, i.e., Theorem~\ref{theorem:0 dots 0 -gamma -beta -alpha CCDV}, each element $(x,y,z)$ leads to four 3DM votes. We make the following addition: For each 3DM vote introduced in this way, we introduce $G$ additional votes obtained from the 3DM votes by swapping the candidates in the $-\gamma$ and $-\beta$-positions (for an appropriate value of $G$).
  More precisely: For each $(x,y,z)\in M$, we add the following votes:
  \begin{itemize}
    \item a single vote   $\dots > S_i > p > x$
    \item $G$ many votes  $\dots > p > S_i > x$
    \item a single vote   $\dots > S_i > p > y$
    \item $G$ many votes  $\dots > p > S_i > y$    
    \item a single vote   $\dots > S_i' > p > z$
    \item $G$ many votes  $\dots > p > S_i' > z$    
    \item a single vote   $\dots > S_i' > p > S_i$
    \item $G$ many votes  $\dots > p > S_i' > S_i$    
  \end{itemize}
  
  The purpose of these added votes (we will call them $G$-votes in the sequel) is to ensure that the candidates $b_\alpha$, $b_\beta$, and $b_\gamma$ gain points relatively to $p$, using votes that are less attractive than the actual 3DM votes from the construction of Theorem~\ref{theorem:0 dots 0 -gamma -beta -alpha CCDV}, such that the controller will delete the latter votes instead of the $G$-votes.
 \end{itemize}
 
 Note that the controller will clearly vote $p$ in the first position of all manipulation votes. Therefore, the score of $p$ will not change from the manipulation votes, but only from the CCDV-aspect of the bribery action. Hence, the score of $p$ behaves in exactly the same way as in the CCDV proof. In the bribery instance we construct, the final scores (i.e., from the 3DM votes including the $G$-votes plus the setup votes that we will introduce below) will be as follows (recall that $n=3k$ as above):
 
 \begin{itemize}
    \item $\scorefinal p=\alpha+2\gamma$,
    \item $\scorefinal c=5k\beta+2\gamma$ for each $c\in X\cup Y\cup Z$,
    \item $\scorefinal{S_i}=5k\beta+\min(\alpha,2\gamma)$,
    \item $\scorefinal{S_i'}=5k\beta+\alpha+\gamma$.
    \item $\scorefinal{b_\heartsuit}=\alpha+2\gamma+5k(\beta+\heartsuit)$ for $\heartsuit\in\set{\alpha,\beta,\gamma}$.
 \end{itemize}
 
 The score of $b_\heartsuit$ is such that $b_\heartsuit$ ties with $p$ if $5k$ 3DM votes are removed (letting $p$ gain $5k\beta$ points), and $b_\heartsuit$ is placed in the $-\heartsuit$-position of each of the $5k$ manipulation votes. Since the score of $b_\heartsuit$ cannot be decreased by deleting 3DM votes, this ensures that the $b_\heartsuit$ candidates must in fact be placed in the relevant positions of each manipulation vote, and therefore, with regard to the remainder of the candidates, the construction works as in the CCDV case.
 
 As in the proof of Theorem~\ref{theorem:0 dots 0 -gamma -beta -alpha CCDV},  we first compute the scores the candidates receive from the 3DM- and $G$-votes, for a candidate $x$, we call this value $\scorethreedm{x}$. 
 
 \begin{itemize}
  \item $\scorethreedm{p}=-12k(\beta+G\gamma)$, since, for each tuple in $M$, $p$ gains $-4\beta-4G\gamma$ points, and $\card M=3k$.
  \item $\scorethreedm{S_i}=-2\gamma-(G+1)\alpha-2G\beta$, 
  \item $\scorethreedm{S_i'}=-2\gamma-2G\beta$,
  \item $\scorethreedm{c}=-3\alpha(G+1)$ for each $c\in X\cup Y\cup Z$, since each $c$ appears in exactly $3$ tuples from $M$,
  \item $\scorethreedm{b_\heartsuit}=0$ for each $\heartsuit\in\set{\alpha,\beta,\gamma}$.
 \end{itemize}
 
 For each candidate $x$, with $\scoresetup{x}$ we denote the points that $x$ needs to receive from the setup votes in order to ensure that $\scorefinal x=\scorethreedm x+\scoresetup x$, i.e., $\scoresetup x=\scorefinal x-\scorethreedm x$. We get the following:
 
 \begin{itemize}
  \item $\scoresetup{p}=\alpha+2\gamma+12k(\beta+G\gamma)=(12kG+2)\gamma+12k\beta+\alpha$,
  \item $\scoresetup{S_i}=5k\beta+\min(\alpha,2\gamma)+2\gamma+(G+1)\alpha+2G\beta=\min(\alpha,2\gamma)+2\gamma+(5k+2G)\beta+(G+1)\alpha$,
  \item $\scoresetup{S_i'}=5k\beta+\alpha+\gamma+2\gamma+2G\beta=3\gamma+\beta(5k+2G)+\alpha$,
  \item $\scoresetup{c}=2\gamma+5k\beta+3\alpha(G+1)$ for each $c\in X\cup Y\cup Z$,
  \item $\scoresetup{b_\heartsuit}=2\gamma+\alpha+5k(\beta+\heartsuit)$.
 \end{itemize}
 
 Clearly, it again suffices to realize the relative scores among the candidates (and clearly, the absolute points of all candidates will be at most $0$, since we wrote $f$ as having no strictly positive coefficient). Hence, it suffices to construct setup votes that for each candidate $x$, let $x$ gain $\scoresetup{p}-\scoresetup{x}$ points \emph{less than the preferred candidate $p$}; this is the number of points that the $x$ must lose against $p$ from the setup votes. For each $x$, we get the following value (clearly for $p$ itself, the value is $0$):
 
 \begin{itemize}
  \item $\scorelosep{S_i}=(12kG+2)\gamma+12k\beta+\alpha-(\min(\alpha,2\gamma)+2\gamma+(5k+2G)\beta+(G+1)\alpha) = -\min(\alpha,2\gamma)+12kG\gamma+\beta(7k-2G)-G\alpha$,
  note that this value can be made arbitrarily large if $k$ is chosen sufficiently large, since $G$ is a constant chosen depending on $\alpha,\beta,\gamma$, but independent of the instance and therefore of $k$.
  \item $\scorelosep{S_i'}=(12kG+2)\gamma+12k\beta+\alpha-(3\gamma+\beta(5k+2G)+\alpha)
                          =(12kG+2)\gamma+12k\beta+\alpha-3\gamma-\beta(5k+2G)-\alpha
                          =\gamma(12kG-1)+\beta(7k-2G)$, again this value grows arbitrarily large in $k$.
  \item $\scorelosep{c}=(12kG+2)\gamma+12k\beta+\alpha - (2\gamma+5k\beta+3\alpha(G+1))
                       = (12kG+2)\gamma+12k\beta+\alpha - 2\gamma-5k\beta-3\alpha(G+1)
                       = 12kG\gamma+7\beta+\alpha(-3G-2)$,
                       again the value grows arbitrarily large in $k$.
  \item $\scorelosep{b_\heartsuit}=(12kG+2)\gamma+12k\beta+\alpha-(2\gamma+\alpha+5k(\beta+\heartsuit))
                                  =(12kG+2)\gamma+12k\beta+\alpha-2\gamma-\alpha-5k(\beta+\heartsuit)
                                  =12kG\gamma+(7k)\beta-5k\heartsuit
                                  =k(12G\gamma+7\beta-5\heartsuit)$.
                                  
                                  Note that, for a suitable choice of $G$, this value also grows arbitrarily for increasing $k$.
 \end{itemize}
 
 Hence, all candidates must lose points against $p$, and the number of points they must lose grows arbitrarily in $k$. Therefore, the points can be implemented using setup votes letting a candidate $x$ lose $\alpha$, $\beta$, or $\gamma$ points against all other relevant candidates. The controller will not delete these votes, since they have $p$ in one of the top positions.
 
 Note that, since $b_\alpha$, $b_\beta$ and $b_\gamma$ must lose $5k(\alpha+\beta+\gamma)$ points against $p$ even if $p$ gains $5k\beta$ points, it follows that these candidates must take all relevant positions in the manipulation votes, and, if $\beta>\gamma$, then $p$ can only win if $p$ indeed gains $5k\beta$ points from the deletions, i.e., only if only non-$G$ 3DM votes are deleted. Therefore, for the remainder of the proof, assume that $\gamma=\beta$.
 
 It remains to show that the controller will in fact delete only the non-$G$-3DM votes. Therefore, assume that there is a successful bribery action in which at least one of the $G$-votes is also deleted. If for every $G$-vote $v_G$ that is removed in the bribery action, the corresponding non-$G$ vote $v_{3DM}$ (obtained from the $G$-vote by swapping the $-\gamma$ and $-\beta$ positions) is not removed, then a bribery action deleting only non-$G$-3DM votes can be obtained by deleting $v_{3DM}$ instead of $v_G$ for every relevant vote (the effect for $p$ is at least as good when deleting $v_{3DM}$). Therefore, we can without loss of generality assume that there is a successful bribery action in which there is a non-$G$ 3DM vote $v_{3DM}$ such that both $v_{3DM}$ and the corresponding $G$-vote $v_G$ are deleted. We show that in this case, $p$ cannot win the election. To see this, first note that $p$ gains at most $5k\beta$ points from the delete actions. Therefore, after the bribery action, $p$'s points are at most (recall that we can assume $\beta=\gamma$)
 
 $$\scoresub{max}p\leq\alpha+2\gamma+(5k-1)\beta+\gamma=\alpha+(5k+2)\beta.$$
 
 We now make a case distinction depending on which type of $G$-vote is deleted. 
 
 \begin{itemize}
  \item First assume that $v_G$ has a candidate $c\in X\cup Y\cup Z$ in the last position. Then the same is true for $v_{3DM}$. Since both votes are deleted, $c$ gains at least $2\alpha$ points (and note that $c$ cannot lose points from deleting other votes). Therefore, $c$ has at least $5k\beta+2\gamma+2\alpha=(5k+2)\beta+2\alpha$ points, which is strictly more than $\scoresub{max}p$.
  \item Now assume that $v_G$ has a candidate $S_i$ in the last position. Analogously to the above, $S_i$ then gains at least $2\alpha$ points, and hence ends up with at least (recall that $\beta=\gamma$) $5k\beta+\min(\alpha,2\gamma) + 2\alpha \ge 5k\beta+2\gamma + 2\alpha = (5k+2)\beta+2\alpha$ points, which again is strictly more than $\scoresub{max}p$.
 \end{itemize}
 
 Therefore, in both cases $p$ does not win the election and we have a contradiction. Therefore, if the bribery instance is positive, then there exists a successful bribe in which only non-$G$ 3DM votes are deleted, as required.
\end{proof}

\section{Open Questions}\label{sect:open questions}

The main open question is to completely characterize the complexity
of $f$-manipulation, for all generators $f$. As discussed at the
end of Section~\ref{sect:results:manipulation},
we conjecture that this will not be a dichotomy theorem.
As a first step, we would like to prove that the
cases listed in Theorem~\ref{theorem:man-p} are exactly the polynomial-time 
cases (under some reasonable complexity-theoretic assumptions) or to
construct an explicit counterexample to this statement.

Other interesting avenues to pursue are 
going beyond NP-completeness, by looking at such issues as
fixed-parameter tractability (see, e.g., \cite{fal-nie:b:FPT}),
approximability (see e.g., \cite{DBLP:journals/jair/FaliszewskiHH15}),
and experimental results (see, e.g.,
\cite{wal:j:where-hard-veto} and \cite{rot-sch:c:empirical-control}).

\vspace*{5mm}

\begin{Huge}
 \noindent\textbf{Appendix}
\end{Huge}

\vspace*{5mm}

\begin{appendix}

In this appendix, we present complete proofs for all of our results, as well as some auxiliary results required for the proofs. The appendix is structured as follows:

\begin{itemize}
 \item In Section~\ref{sect:manipulation ptime proofs}, we prove our polynomial-time results for the manipulation problem. The remainder of the Appendix is then decidated to proving our dichotomy theorem.
 \item In Section~\ref{appendix:bribery hardness}, we state our individual bribery hardness results results.
 \item Section~\ref{section:dichotomy proof appendix} contains the proof of our dichotomy result for CCDV and bribery, based on the results stated in the main paper and in Section~\ref{appendix:bribery hardness}.
 \item Section~\ref{sect:auxiliary results} contains auxiliary results required for the invividual hardness proofs for CCDV bribery. These results follow in Sections~\ref{sect:ccdv proofs} and~\ref{sect:proofs bribery}, respectively.
\end{itemize}

\section{Proofs for Manipulation in Polynomial Time}\label{sect:manipulation ptime proofs}

In this section, we state and prove correctness of our polynomial-time algorithms for manipulation. We note that Theorem~\ref{theorem:man} is a direct consequence of Theorem~\ref{theorem:man-p}. However, the proof of the special case treated in Theorem~\ref{theorem:man} is simpler and still contains the main ideas required to the proof of the more general Theorem~\ref{theorem:man-p}. We therefore also present the proof of the simpler result. In addition, our proof of Theorem~\ref{theorem:bribery zeroes -beta -alpha in ptime} below uses the algorithm from Theorem~\ref{theorem:man}.

\subsection{Proof of Theorem~\ref{theorem:man}}

\begin{restatable}{restatableTheorem}{theoremman}\label{theorem:man}
Let $\alpha \geq \beta \geq 0$.  Then manipulation for
$(0, \ldots, 0, -\beta, -\alpha)$ is solvable in polynomial time.
\end{restatable}

\begin{proof}
Consider an instance of the manipulation problem, consisting of a set of candidates $C$, a preferred candidate $p$, the surplus $\surpl c$ for each $c\in C$ (i.e., the value $\score{c}-\score{p}$), and a number $k$ of available
manipulators.\footnote{Usually, the input to the manipulation problem does not contain the scores for each candidate but a set of voters already having voted, but it's clear how to get the scores from the voters in polynomial time. Note that this means we actually prove a stronger result, where the current scores of the candidates
don't have to be realizable and  may be given in binary.} 
Let $C - \{p\} = \{c_1, \ldots, c_m\}$.

It is obvious that all manipulators can without loss of generality be assumed to vote $p$ in the first place. The goal of the manipulators is to ensure that after manipulating, no candidate has a positive surplus.

Note that the obvious greedy approach of having a manipulator rank a
candidate with the largest surplus last won't work in all cases. For example,
if $\beta = 2$, $\alpha = 3$, the surplus of $c_1$ is 4,
the surplus of $c_2$ and $c_3$ is 3, and we have two manipulators,
the only successful manipulation is to have the manipulators vote
$\cdots > c_1 > c_2$ and $\cdots > c_1 > c_3$ and so we should not put
$c_1$ last.\footnote{This example is realizable with a couple of dummy
candidates.  Simply put (in the nonmanipulators) $p$ last once and next-to-last
twice, put $c_1$ next-to-last twice, put $c_2$ and $c_3$ last once, and
fill all other last two positions with dummies, in such a way that the dummies
do not have positive surplus.}

Note that if there exists a successful manipulation, then for
all $i$, $1 \leq i \leq m$, there exist
nonnegative integers $x_i$ (the number of times $c_i$ is ranked next to last 
by a manipulator) and $y_i$ (the number of times that $c_i$ is ranked last
by a manipulator) such that:
\begin{enumerate}
\item $x_i + y_i \leq k$,
\item $\sum_{1 \leq i \leq m} x_i = k$,
\item $\sum_{1 \leq i \leq m} y_i = k$, and
\item $\surpl{c_i} - \beta x_i - \alpha y_i \leq 0$.
\end{enumerate}

We define the following Boolean predicate $M$.
$M(k, k_\beta, k_\alpha, s_1, \ldots, s_\ell)$
is true if and only if for all $i$, $1 \leq i \leq \ell$,
there exist natural numbers $x_i$ and $y_i$ such that
\begin{enumerate}
\item $x_i + y_i \leq k$,
\item $\sum_{1 \leq i \leq \ell} x_i = k_\beta$,
\item $\sum_{1 \leq i \leq \ell} y_i = k_\alpha$, and
\item $s_i - \beta x_i - \alpha y_i \leq 0$.
\end{enumerate}

Note that if there is a successful manipulation, then
$M(k, k, k, \surpl{c_1}, \ldots , \surpl{c_m})$ is true.
We will now show that the converse is true as well:
If $M(k, k, k, \surpl{c_1}, \ldots , \surpl{c_m})$ is true then
there exists a successful manipulation.

We will prove this by induction on $k$.
If $k = 0$, then for all $i$, $\surpl{c_i} \leq 0$, so $p$ is a winner.

Now suppose that the claim holds for $k \geq 0$.  We will show that it also
holds for $k + 1$.  So, for all $i$, $1 \leq i \leq m$,
let $x_i, y_i$ be natural numbers such that:
\begin{enumerate}
\item for all $i$, $x_i + y_i \leq k+1$,
\item $\sum_{1 \leq i \leq m} x_i = k+1$,
\item $\sum_{1 \leq i \leq m} y_i = k+1$, and
\item $\surpl{c_i} - \beta x_i - \alpha y_i \leq 0$.
\end{enumerate}
Let $X = \{i \ | \ x_i + y_i = k+1\}$.  Note that $\card{X} \leq 2$.
Let $i, j$ be such that $i \neq j$, $x_i > 0$, $y_j > 0$,
and $X \subseteq \{i,j\}$.  Let one manipulator vote $\cdots > c_i > c_j$.
Subtract 1 from $x_i$ and $y_j$, subtract $\beta$ from
$\surpl{c_i}$, and subtract $\alpha$ from from $\surpl{c_j}$.
It follows from the induction hypothesis that
the remaining $k$ manipulators can vote to make $p$ a winner.

To conclude the proof of Theorem~\ref{theorem:man}, we will now 
show, by dynamic programming, that $M$ is computable
in polynomial time for unary $k, k_\beta, k_\alpha \geq 0$.  
This is easy:
\begin{enumerate}
\item
$M(k,k_\beta, k_\alpha)$ is true if and only if
$k_\beta = k_\alpha = 0$.
\item For $\ell \geq 1$,
$M(k,k_\beta, k_\alpha, s_1, \ldots, s_\ell)$ if and only if
there exist natural numbers $x_\ell$ and $y_\ell$ such that:
\begin{enumerate}
\item
$x_\ell + y_\ell \leq k$,
\item
$x_\ell \leq k_\beta$,
\item
$y_\ell \leq k_\alpha$,
\item
$s_\ell - \beta x_\ell - \alpha y_\ell \leq 0$, and
\item
$M(k,k_\beta - x_\ell, k_\alpha - y_\ell, s_1, \ldots, s_{\ell-1}).$
\end{enumerate}
\end{enumerate}
\end{proof}

\subsection{Proof of Theorem~\ref{theorem:man-p}}

\theoremmanp*

\begin{proof}
Consider an instance of the manipulation problem with $m+1$ candidates.
Let $r, m_1, \ldots, m_r, \alpha_2, \ldots, \alpha_r$ be positive integers such that $m_1 + \cdots + m_r = m$, $0 > -\alpha_2 > \cdots > -\alpha_r$, and
$f(0^{m+1})$ is equivalent to
$(0^{m_1},-\alpha_2^{m_2},\ldots, -\alpha_r^{m_r})$.

Let the instance of the manipulation problem consist of
a set of candidates $C$, a preferred candidate $p$, the surplus $\surpl c$ for each $c\in C$ (i.e., the value $\score{c}-\score{p}$), and a number $k$ of available manipulators.\footnote{Usually, the input to the manipulation problem does not contain the scores for each candidate but a set of voters already having voted, but it's clear how to get the scores from the voters in polynomial time. Note that this means we actually prove a stronger result, where both the current scores
don't have to be realizable and  may be given in binary.} 
Let $C - \{p\} = \{c_1, \ldots, c_m\}$.

It is obvious that all manipulators can without loss of generality be assumed to vote $p$ in the first place. The goal of the manipulators is to ensure that after manipulating, no candidate has a positive surplus.

Note that if there exists a successful manipulation, then for
all $i,j$, $1 \leq i \leq m$, $2 \leq j \leq r$, there exist
nonnegative integers $x_{i,j}$ (the number of times $c_i$ gets $-\alpha_j$
points from a manipulator) such that
\begin{enumerate}
\item $\sum_{2 \leq j \leq r} x_{i,j} \leq k$,
\item $\sum_{1 \leq i \leq m} x_{i,j} = m_j k$, and
\item $\surpl{c_i} - \sum_{j = 2}^r \alpha_j x_{i,j} \leq 0$.
\end{enumerate}

We define the following Boolean predicate $M$.
$M(k, \ell_2, \ldots, \ell_r, s_1, \ldots, s_\ell)$
is true if and only if for all
$i,j$, $1 \leq i \leq \ell$, $2 \leq j \leq r$, there exist
nonnegative integers $x_{i,j}$ such that
\begin{enumerate}
\item $\sum_{2 \leq j \leq r} x_{i,j} \leq k$,
\item $\sum_{1 \leq i \leq \ell} x_{i,j} = \ell_j$, and
\item $s_i - \sum_{j = 2}^r \alpha_j x_{i,j} \leq 0$.
\end{enumerate}

Note that if there is a successful manipulation, then
$M(k,  m_2 k, \ldots, m_r k,
\surpl{c_1}, \ldots , \surpl{c_m})$ is true.
We will now show that the converse is true as well:
If $M(k, m_2 k, \ldots, m_r k,
\surpl{c_1}, \ldots , \surpl{c_m})$ is true then
there exists a successful manipulation.

We will prove this by induction on $k$.
If $k = 0$, then for all $i$, $\surpl{c_i} \leq 0$, so $p$ is a winner.

Now suppose that the claim holds for $k \geq 0$.  We will show that it also
holds for $k + 1$.  So, for 
all $i,j$, $1 \leq i \leq m$, $2 \leq j \leq r$, let
$x_{i,j}$ be natural numbers such that
\begin{enumerate}
\item $\sum_{2 \leq j \leq r} x_{i,j} \leq k+1$,
\item $\sum_{1 \leq i \leq m} x_{i,j} = m_j (k+1)$, and
\item $\surpl{c_i} - \sum_{j = 2}^r \alpha_j x_{i,j} \leq 0$.
\end{enumerate}
For $2 \leq j \leq r$, let $X_j = \{c_i \ | \ x_{i,j} > 0\}$.
Consider the following sequence of sets:
$m_2$ copies of $X_2$ followed by $m_3$ copies of $X_3$ followed by $\dots$ 
followed by $m_r$ copies of $X_r$.  We will show that this sequence fulfills
the ``marriage condition,'' which then, by Hall's Theorem, implies that there
is a ``traversal,'' i.e., a sequence of distinct representatives of
this sequence of sets, which then gives us a vote for one of the
manipulators. For every candidate $c_i$ such that $c_i$ represents 
an occurrence of
set $X_j$ in the sequence,
we know that $x_{i,j} > 0$.  Subtract 1 from each such
$x_{i,j}$ 
and recompute the surpluses. It follows from the induction hypothesis that
the remaining $k$ manipulators can vote to make $p$ a winner.

Suppose for a contradiction that our sequence of sets does not
fulfill the marriage condition. Then there is a subcollection  ${\cal S}$
of $t$ sets such that the union of these $t$ sets, call it $S$,
contains fewer than $t$ elements. 

Note that $\sum_{c_i \in S}\sum_{2 \leq j \leq r} x_{i,j} < t(k+1)$, which
implies that
$\sum \{x_{i,j} \ | \ 1 \leq i \leq m, 2 \leq j \leq r, c_i \in S,
X_j \in {\cal S}\} < t(k+1)$.
But also note that, since $x_{i,j} = 0$ for $c_i \not \in X_j$,
$\sum \{x_{i,j} \ | \ 1 \leq i \leq m, 2 \leq j \leq r, c_i \in S,
X_j \in {\cal S}\} =
\sum \{x_{i,j} \ | \ 1 \leq i \leq m, 2 \leq j \leq r,
X_j \in {\cal S}\} = 
\sum \{m_j (k+1) \ | \ 2 \leq j \leq r, X_j \in {\cal S}\} = 
\sum \{m_j \ | \ 2 \leq j \leq r, X_j \in {\cal S}\} (k+1) \geq t(k+1)$,
which is a contradiction.

To conclude the proof of Theorem~\ref{theorem:man}, we will now 
show, by dynamic programming, that $M$ is computable
in polynomial time for unary
$k, \ell_2, \ldots, \ell_r \geq 0$ and $r$ bound by a fixed constant.
This is easy:
\begin{enumerate}
\item
$M(k, \ell_2, \ldots, \ell_r)$
is true if and only if
$\ell_2 = \cdots = \ell_r = 0$.
\item For $\ell \geq 1$,
$M(k, \ell_2, \ldots, \ell_r, s_1, \ldots, s_\ell)$ if and only if
there exist natural numbers $x_{\ell,j}$, $2 \leq j \leq r$,
such that:
\begin{enumerate}
\item $\sum_{2 \leq j \leq r} x_{\ell,j} \leq k$,
\item $x_{\ell,j} \leq \ell_j$, 
\item $s_\ell - \sum_{j = 2}^r \alpha_j x_{\ell,j} \leq 0$, and
\item
$M(k, \ell_2 - x_{\ell,2}, \ldots, \ell_r - x_{\ell, r},
s_1, \ldots, s_{\ell-1})$.
\end{enumerate}
\end{enumerate}
\end{proof}

\section{Results: Bribery Hardness}\label{appendix:bribery hardness}

In this section, we state the hardness results obtained for bribery. These results mirror the results for CCDV obtained in Sections~\ref{sect:ccdv few coefficients general cases} and~\ref{sect:ccdv many coefficients}. The following Section~\ref{section:dichotomy proof appendix} then shows how to combine the CCDV results from Section~\ref{sect:results:ccdv}, the polynomial-time bribery results from Section~\ref{sect:polynomial time bribery}, and the hardness results presented in Section~\ref{sect:bribery hardness from ccdv hardness} and the results here into a proof for our CCDV/bribery dichotomy result, Theorem~\ref{theorem:ccdv dichotomy}.

The proofs of the results in the current section can be found in Section~\ref{section:bribery hardness proofs}. 

\subsection{Bribery hardness for ``many coefficients''}

We now present the ``bribery analogues'' of the results obtained for CCDV in Section~\ref{sect:ccdv many coefficients}, namely, we prove that bribery is NP-hard if the generator uses ``many'' coefficients. In fact, in the same way as for CCDV, we obtain a result that is somewhat stronger, namely it covers all generators not of the form $(\alpha_1,\alpha_2,\alpha_3,\dots,\alpha_3,\alpha_4,\alpha_5,\alpha_6)$. Among others, this includes generalizations of $k$-veto for $k\ge4$ and $k$-approval for $k\ge3$.

The structure of the proof is slightly different to the corresponding situation with CCDV presented in Section~\ref{sect:ccdv many coefficients}, since a more complex case distinction is required. The reason is that we need ``blocking candidates'' in the front, making the proof of Theorem~\ref{theorem:approval generalization bribery} only work if $\alpha^m_4$ is ``large'' (this results in the need for an additional case, namely Theorem~\ref{theorem:bribery fixed coefficients additional case}), and in order to implement the setup votes, we sometimes rely on the generator providing $3$ different coefficients (for a suitably large number of candidates---note that due to the purity condition, the generator will then use at least $3$ different coefficients for all but a finite number of candidates).

We start with a generalization of $4$-approval. This corresponds to Theorem~\ref{theorem:approval generalization ccdv}, which generalizes $3$-approval for CCDV. As mentioned above, for the bribery case, generators that are ``close to'' $3$-approval in the sense that the ``bad'' positions start with $\alpha^m_4$, form a specific case, which is covered in Theorem~\ref{theorem:bribery fixed coefficients additional case}, and we rely on at least $3$ coefficients. The two-coefficient case is covered in Theorem~\ref{theorem:two coefficients case bribery}.

\begin{restatable}{restatableTheorem}{thmapprovalgeneralizationbribery}\label{theorem:approval generalization bribery}
 Let $f$ be an polynomial-time uniform $\mathbb Q$-generator with $f(m)=(\alpha^m_1,\dots,\alpha^m_m)$ for each $m$, such that $f$ uses at least three coefficients starting from some $m$.
 
 Then there is a polynomial-time computable function $g$ such that
 \begin{itemize}
  \item $g$ takes as input an instance $M$ of $F$-3DM for some arbitrary $F$ with $\card M=n$\footnote{in this case, $n$ is in fact the number of \emph{tuples} in $M$, not the size of a binary representation.} and produces an instance $I$ of $f$-bribery with $m=3n$ candidates,
  \item if $\alpha^m_4>\alpha^m_{2n}$, and $\card M\ge\card X^2+2\card X+2$, then: $M$ is a positive instance of 3DM if and only if $I$ is a positive instance of $f$-bribery.
 \end{itemize}
\end{restatable}

Similarly, we now generalize $3$-veto, i.e., state the bribery variant of Theorem~\ref{theorem:veto generalization}:

\begin{restatable}{restatableTheorem}{thmvetogeneralizationbribery}\label{theorem:veto generalization bribery}
 Let $f$ be an polynomial-time uniform $\mathbb Q$-generator with $f(m)=(\alpha^m_1,\dots,\alpha^m_m)$ for each $m$, such that $f$ uses at least three coefficients starting from some $m$.
 
 Then there is a polynomial-time computable function $g$ such that
 \begin{itemize}
  \item $g$ takes as input an instance $M$ of $F$-3DM for some arbitrary $F$ with $\card M=n$\footnote{in this case, $n$ is in fact the number of \emph{tuples} in $M$, not the size of a binary representation.} and produces an instance $I$ of $f$-bribery with $m=3n$ candidates,
  \item if $\alpha^m_{\frac23m}>\alpha^m_{m-4}$, and $\card M\ge\card X^3$, then: $M$ is a positive instance of 3DM if and only if $I$ is a positive instance of $f$-bribery.
 \end{itemize}
\end{restatable}

Note that this veto-like case is a bit simpler than the corresponding approval-related case above, since we do not need an additional position for the ``blocking candidate'' at the beginning of the vote. Unlike for the CCDV case, we need one other case with a ``fixed'' number of coefficients:

\begin{restatable}{restatableTheorem}{theorembriberyfixedcoefficientsadditionalcase}\label{theorem:bribery fixed coefficients additional case}
 Let $f$ be the generator $(\alpha_1,\alpha_2,\alpha_3,\alpha_4,\dots,\alpha_4,\alpha_5,\alpha_6,\alpha_7)$ with $\alpha_3>\alpha_4$. Then $f$-bribery is \NP-complete.
\end{restatable}

In the above Theorems~\ref{theorem:approval generalization bribery} and~\ref{theorem:veto generalization bribery}, we used the fact that the generator uses at least $3$ coefficients, which allowed us to easily implement the required setup votes. Hence we now need to treat the case that $f$ uses only $2$ distinct coefficients, without loss of generality these coefficients are then $0$ and $1$.

\begin{restatable}{restatableTheorem}{theoremtwocoefficientscasebribery}\label{theorem:two coefficients case bribery}
 Let $f$ be a generator such that $\alpha^m_3>\alpha^m_{m-4}$ for some $m$, and such that $f$ only uses $2$ coefficients. Then $f$-bribery is \NP-complete.
\end{restatable}

Note that Theorem~\ref{theorem:two coefficients case bribery} does not only cover generators of the form $k$-veto or $k$-approval, but also systems we might call $\frac1{10}$-approval or $\sqrt.$-veto

The results from this section now allow to prove that, in fact, if we have ``many'' different coefficients, i.e., a generator not of the form $(\alpha_1,\alpha_2,\alpha_3,\dots,\alpha_3,\alpha_4,\alpha_5,\alpha_6)$, then the bribery problem is \NP-complete.

\begin{restatable}{restatableCorollary}{corollarymanycoefficientsbriberyhardness}\label{corollary:many coefficients bribery hardness}
 Let $f$ be a polynomial-time uniform generator such that $\alpha^m_3>\alpha^m_{m-3}$ for some $m$. Then $f$-bribery is \NP-complete.
\end{restatable}

\subsection{Bribery hardness for ``few coefficients''}

The following result states that the hardness results for CCDV obtained in Theorems~\ref{theorem:alpha 0 dots 0 alpha < beta ccdv} and~\ref{theorem:alpha 0 dots 0 beta < alpha ccdv} as well as Corollary~\ref{corollary:alpha 0 dots 0 beta ccdv dichotomy} carry over to bribery:

\begin{restatable}{restatableCorollary}{corollaryalphazerodotszerobetabriberydichotomy}
 \label{corollary:alpha 0 dots 0 beta bribery dichotomy}
 Let $f=(\alpha_1,\alpha_2,\dots,\alpha_2,\alpha_3)$ be a generator with $\alpha_1>\alpha_2>\alpha_3$. If $f$ is equivalent to $(2,1,\dots,1,0)$, then $f$-bribery can be solved in polynomial time, otherwise, $f$-bribery is NP-complete.
\end{restatable}

We now state the bribery version of Theorem~\ref{theorem:0 minus alpha5 minus alpha1 dots minus alpha1 CCDV hardness}:

\begin{restatable}{restatableTheorem}{theoremzerominusalphafiveminusalphaonedotsminusalphaonebriberyhardness}
 \label{theorem:0 minus alpha5 minus alpha1 dots minus alpha1 bribery hardness}
 Let $f=(\alpha_1,\alpha_2,\alpha_3,\dots,\alpha_3)$ with $\alpha_1>\alpha_2>\alpha_3$. Then $f$-bribery is \NP-complete.
\end{restatable}

The next result is the bribery analog of Theorem~\ref{theorem:ccdv hardness alpha1 alpha2 0 dots 0 minus alpha4 minus alpha5 minus alpha6}:

\begin{restatable}{restatableTheorem}{theorembriberyhardnessalphaonealphatwozerodotszerominusalphafourminusalphafiveminusalphasix}
 \label{theorem:bribery hardness alpha1 alpha2 0 dots 0 minus alpha4 minus alpha5 minus alpha6}
 Let $f=(\alpha_1,\alpha_2,\alpha_3,\dots,\alpha_3,\alpha_4,\alpha_5,\alpha_6)$ with $\alpha_1>\alpha_3>\alpha_5$. Then $f$-bribery is NP-complete.
\end{restatable}

The following two results show that the hardness results in Theorems~\ref{theorem:ccdv alpha1alpha2>alpha3>alpha6} and~\ref{theorem:ccdv hardness alpha1 > alpha 3 > alpha 4} carry over to bribery as well:

\begin{restatable}{restatableTheorem}{theorembriberyalphaoneaphatwobiggerthanalphathreebiggerthanalphasix}
 \label{theorem:bribery alpha1alpha2>alpha3>alpha6}
 Let $f=(\alpha_1,\alpha_2,\alpha_3,\dots,\alpha_3,\alpha_6)$ with $\alpha_1,\alpha_2>\alpha_3>\alpha_6$. Then $f$-bribery is NP-complete.
\end{restatable}

\begin{restatable}{restatableTheorem}{theorembriberyhardnessalphaonebiggeralphathreebiggeralphafour}
 \label{theorem:bribery hardness alpha1 > alpha 3 > alpha 4}
 Let $f=(\alpha_1,\alpha_2,\alpha_3,\dots,\alpha_3,\alpha_4,\alpha_5,\alpha_6)$ with $\alpha_1>\alpha_3>\alpha_4$. Then $f$-bribery is NP-complete.
\end{restatable}

\section{Proof of CCDV and Bribery Dichotomy}\label{section:dichotomy proof appendix}

\theoremccdvandbriberydichotomy*

\begin{proof}
 The polynomial-time cases for CCDV follow directly from the discussion in Section~\ref{sect:ccav and ccdv}. For bribery, the result for the first three cases follows from Theorem~\ref{theorem:lin bribery results list}, the result for $f_4$ is Theorem~\ref{theorem:bribery zeroes -beta -alpha in ptime}, and the result for $f_5$ is Theorem~\ref{theorem:100-1 bribery in ptime}.
 
 Therefore, let $f$ be a pure generator not of this form.
 
 If there is some $m$ with $\alpha^m_3>\alpha^m_{m-3}$, then NP-hardness follows from Theorem~\ref{theorem:many different coefficients} and Corollary~\ref{corollary:many coefficients bribery hardness}.
 
 Therefore, it suffices to consider the case that $\alpha^m_3=\alpha^m_{m-3}$ for all $m\ge6$. Therefore, $f$ is of the form $f=(\alpha_1,\alpha_2,\alpha_3,\dots,\alpha_3,\alpha_4,\alpha_5,\alpha_6)$. 
 
 We first consider the case $\alpha_1=\alpha_3$, then $f=(\alpha_3,\dots,\alpha_3,\alpha_4,\alpha_5,\alpha_6)$. If $\alpha_3=\alpha_4$, then $f$ is equivalent to $(0,\dots,0,-\alpha,-\beta)$, i.e., a generator of the form $f_4$. Therefore, assume that $\alpha_3>\alpha_4$. If $\alpha_4=\alpha_6$, then $f$ is equivalent to $3$-veto (generator $f_1$). Therefore, assume that $\alpha_3>\alpha_4>\alpha_6$. In this case, NP-completeness follows from Theorems~\ref{theorem:0 dots 0 -gamma -beta -alpha CCDV} and~\ref{theorem:0 dots 0 -gamma -beta -alpha bribery}.
 
 Therefore, we can assume that $\alpha_1>\alpha_3$. If $\alpha_3>\alpha_4$, then NP-completeness follows from Theorems~\ref{theorem:collection of few coefficients CCDV hardness cases}.\ref{few coefficients CCDV:ccdv hardness alpha1 > alpha 3 > alpha 4} and~\ref{theorem:bribery hardness alpha1 > alpha 3 > alpha 4}. Therefore, assume, for the remainder of the proof, that $\alpha_1>\alpha_3=\alpha_4$. Therefore, $f$ is of the form $f=(\alpha_1,\alpha_2,\alpha_3,\dots,\alpha_3,\alpha_5,\alpha_6)\mathtext{ with }\alpha_1>\alpha_3$.
 
 If $\alpha_2=\alpha_5$, then the claim follows from Theorems~\ref{theorem:alpha 0 dots 0 alpha < beta ccdv} and~\ref{theorem:alpha 0 dots 0 beta < alpha ccdv} in the CCDV case, and from Corollary~\ref{corollary:alpha 0 dots 0 beta bribery dichotomy} for bribery. Therefore, we assume that $\alpha_2>\alpha_5$. We summarize: Using the results so far, we can assume that 
 
 $$f=(\alpha_1,\alpha_2,\alpha_3,\dots,\alpha_3,\alpha_5,\alpha_6)\mathtext{ with }\alpha_1>\alpha_3\mathtext{ and }\alpha_2>\alpha_5.$$
 
 Consider the case that $\alpha_3=\alpha_6$, then $f=(\alpha_1,\alpha_2,\alpha_3,\dots,\alpha_3)$. If $\alpha_1=\alpha_2$, then $f$ is $2$-approval (generator $f_3$), if $\alpha_2=\alpha_3$, then $f$ is $1$-approval (generator $f_2$). Therefore, we can assume that $\alpha_1>\alpha_2>\alpha_3$, in which case NP-completeness follows from Theorems~\ref{theorem:0 minus alpha5 minus alpha1 dots minus alpha1 CCDV hardness} and~\ref{theorem:0 minus alpha5 minus alpha1 dots minus alpha1 bribery hardness}.
 
 Therefore, we assume that $\alpha_3>\alpha_6$. We make a final case distinction.
 
 \begin{itemize}
  \item If $\alpha_5<\alpha_3$, then $f=(\alpha_1,\alpha_2,\alpha_3,\dots,\alpha_3,\alpha_5,\alpha_6)$ with $\alpha_1>\alpha_3>\alpha_5$. In this case, hardness follows from Theorems~\ref{theorem:collection of few coefficients CCDV hardness cases}.\ref{few coefficients CCDV:ccdv hardness alpha1 alpha2 0 dots 0 minus alpha4 minus alpha5 minus alpha6} and~\ref{theorem:bribery hardness alpha1 alpha2 0 dots 0 minus alpha4 minus alpha5 minus alpha6}.
  \item Finally, assume that $\alpha_3=\alpha_5$. In this case, $f=(\alpha_1,\alpha_2,\alpha_3,\dots,\alpha_3,\alpha_6)$, with $\alpha_1>\alpha_3$, $\alpha_2>\alpha_5=\alpha_3$, and $\alpha_3>\alpha_6$. In this case, NP-completeness follows from Theorems~\ref{theorem:collection of few coefficients CCDV hardness cases}.\ref{few coefficients CCDV:ccdv alpha1alpha2>alpha3>alpha6} and~\ref{theorem:bribery alpha1alpha2>alpha3>alpha6}, which then concludes the proof.
 \end{itemize}
\end{proof}

\section{Auxiliary Results}\label{sect:auxiliary results}

In this section we state and prove two auxiliary results required for our hardness proofs.

\subsection{Realizing Scores}

Throughout our hardness proofs, we will make use of the following result from~\cite{HemaspaandraHemaspaandraSchnoor-CCAV-AAAI-2014}, which shows that, given a sufficiently uniform generator, we can construct elections instances with any ``reasonable'' set of scores. This allows us to construct the appropriate situations in our hardness proofs for both CCDV and bribery. 

\begin{restatableLemma}\label{lemma:coefficients realization}
 Given a scoring vector $(\alpha_1,\dots,\alpha_m)$, and for each $c\in\set{1,\dots,m-1}$, numbers $a^c_1,\dots,a^c_m$ in signed unary, and a number $k$ in unary, we can compute, in polynomial time, votes such that the scores of the candidates 
 when evaluating these votes according to the scoring vector $(\alpha_1,\dots,\alpha_m)$
are as follows: There is some $o$ such that for each $c\in\set{1,\dots,m-1}$,
$\score{c}=o+\sum_{i=1}^ma^c_i
\alpha_i$, and $\score{c}>\score{m}+k
\alpha_1$.
\end{restatableLemma}

\subsection{3DM and its variants}

In this section, we formally define 3DM and prove that the restriction we use in our hardness proofs in fact remains NP-complete.

3DM is the following problem: Given a multiset (in most cases, $M$ will be a set, however some proofs use the multi-set version of the problem) $M\subseteq X\times Y\times Z$ with $X$, $Y$ and $Z$ pairwise disjoint sets of equal size such that each $s\in X\cup Y\cup Z$ appears in exactly $3$ tuples of $M$, decide whether there is a set $C\subseteq M$ with $\card C=\card X$ such that each $s\in X\cup Y\cup Z$ appears in some tuple of $C$ (we also say that $C$ \emph{covers} $s$). 
From the problem definition, it follows that $\card M=3\card X$. 

The condition that each $s$ appears in exactly $3$ tuples is not standard; hence we prove that this version of 3DM indeed remains NP-complete.

\begin{restatable}{restatableProposition}{propthreedmrestriction}\label{prop:3dm restriction}
 3DM is NP-complete.
\end{restatable}

\begin{proof}
In~\cite{gar-joh:b:int}, it is proved that the version of 3DM where every element may occur in \emph{at most} three triples of $M$ is NP-complete. This immediately implies that the version of 3DM where every element occurs in two or three triples of $M$ is NP-complete at well: If some element $c\in X\cup Y\cup Z$ appears in exactly one triple $(x,y,z)$, we know that the triple $(x,y,z)$ must be part of the cover. So we can delete $(x,y,z)$ and the elements $x$, $y$, and $z$, as well as all other triples that contain one of $x,y,z$. We continue this process until either some element $c$ occurs in no triples (in which case we have a negative instance) or all elements occur in two or three triples. 

We now prove the actual result. For this, consider an instance of 3DM in which every element occurs in two or three triples of $M$.  Note that the number of elements in $X$ that occur in two triples is the same as the number of elements in $Y$ that occur in two triples, which is the same as the number of elements in $Z$ that occur in two triples. Let $t$ denote this number. We view these elements as $t$ elements of $X\times Y\times Z$.  Let $(x,y,z)$ be such an element, i.e., $x\in X$, $y\in Y$, and $z\in Z$ such that $x$, $y$, and $z$ each occur in two triples in $M$.  We add three elements, $x'\in X$, $y'\in Y$, and $z'\in Z$, and add the following four triples to $M$:

\begin{itemize}
 \item $(x, y', z')$,
 \item $(x', y, z')$,
 \item $(x', y', z)$,
 \item $(x', y', z')$.
\end{itemize}

We do this for all $t$ triples and call the resulting set of triples $M’$.  Note that every element occurs in exactly three triples in $M’$. It is easy to see that any cover in $M'$ needs to contain $(x',y',z')$. Therefore, $(x, y', z')$, $(x’, y, z')$, and $(x', y', z)$ are not in the matching, which implies that a matching in $M'$ restricted to the original elements is a matching in $M$.
\end{proof}

In some reductions, the following variant of 3DM is useful: For a natural number $F\ge1$, $F$-3DM is defined analogously to 3DM, except that every $s\in X\cup Y\cup Z$ appears in exactly $3F$ tuples of $M$. The question is still whether there is a cover with size $\card X$. Clearly, $F$-3DM is still NP-complete; in fact the following slightly stronger claim holds:

\begin{restatableProposition}\label{prop:f3dm np complete}
 There is a polynomial-time algorithm which, given a 3DM-instance $I$ and a natural number $F\ge1$ in unary, produces an $F$-3DM instance $I'$ such that $I$ is a positive 3DM instance if and only if $I'$ is a positive $F$-3DM instance.
\end{restatableProposition}

\begin{proof}
 The proof follows from simply repeating every tuple in the set $M$ from the given instance $I$ exactly $F$ times. Clearly the question whether there is a cover of any size is invariant under this transformation.
\end{proof}

\section{Proofs of Results for CCDV}\label{sect:ccdv proofs}

\subsection{Relationships Between CCDV, \ccdvstar, and CCAV}\label{sect:relationship ccdv ccav}

We now relate the complexities of $f$-\ccdvstar and $f$-CCAV. The problems are very similar, but clearly, a reduction between them has to do more than simply reversing all votes, since the $f$-CCDV problem is only the same as the $\dual f$-CCAV problem in a situation with a fixed set of scores of all candidates, when the votes available for deleting/addition are obtained from each other by simply reversing the order. To set up this situation---that is, to ensure that the relative scores of all candidates are the same in both settings---the following result relies on an implementation Lemma, and therefore requires a uniformity condition on the generator $f$.

\propccdvstarpolyequivccav*

\begin{proof}
 We first show that $f$-CCAV reduces to $\dual f$-\ccdvstar. Hence, let an $f$-CCAV instance be given, with registered voters $R$ and unregistered voters $U$, favorite candidate $p$, and a number $k$ of votes that may be added. The votes in the $\dual f$-\ccdvstar instance are constructed as follows:
 
 \begin{itemize}
  \item the \emph{deletable} votes $D$ contain, for each vote $v\in U$, the vote $\dual v$ obtained from $v$ by reversing the order of candidates in $v$. (Recall that we regard the set of votes as multisets, i.e., each vote can appear more than once.)
  \item the \emph{not deletable} votes $R$ are setup-votes to ensure that the relative points of each candidate are the same as when counting only the votes in $R$ with regard to the original generator $f$. (These can be constructed using Lemma~\ref{lemma:coefficients realization}. Note that simply using the reversals of the votes in $R$ does not give the correct result, since we need the $R$-points with regard to $f$, but must use the generator $\dual f$ to achieve them.)
 \end{itemize}
 
 It is now obvious that adding (up to $k$) votes from $U$ has the same effect as deleting the corresponding deletable votes.
 
 The proof of the converse direction is very similar: Let an $f$-\ccdvstar instance be given, consisting of deletable votes $D$ and fixed votes $R$. Then deleting a vote $v$ in $D$ has the same effect as adding the vote $\dual v$. Hence, by using the implementation Lemma in the same way as in the proof of Proposition~\ref{prop:ccdvstar poly equiv ccav}, we can reduce $f$-\ccdvstar to $\dual f$-CCAV.
\end{proof}

\subsection{Polynomial-time results for CCDV}

A direct consequence of Proposition~\ref{prop:ccdvstar poly equiv ccav} (together with the obvious fact that $f$-CCDV always reduces to $f$-\ccdvstar) is that the polynomial-time results obtained in for CCAV carry over (using dualization). We therefore immediately get the following result:

\begin{restatable}{restatableTheorem}{theoremccdvfewcoefficientsptimecases}
 \label{theorem:ccdv few coefficients ptime cases}
 For the following generators, CCDV and \ccdvstar can be solved in polynomial time:
 \begin{enumerate}
  \item $f_1=(1,\ldots, 1, 0, 0, 0)$
  \item $f_2=(1, 0, \dots,0)$ ($1$-approval),
  \item $f_3=(1, 1, 0, \dots, 0)$ ($2$-approval),
  \item for some $\alpha\ge\beta \ge 0$, $f_4=(0, \ldots, 0, -\beta, -\alpha)$,
  \item $f_5=(2,1,\dots,1,0)$.
 \end{enumerate}
\end{restatable}

\begin{proof}
 All of these results follow directly from Proposition~\ref{prop:ccdvstar poly equiv ccav} and the results in~\cite{HemaspaandraHemaspaandraSchnoor-CCAV-AAAI-2014}. We use the fact that the problem $f$-CCAV is the same problem as $f'$-CCAV when $f'$ is obtained from $f$ by an affine transformation.
 \begin{enumerate}
  \item In this case, $\dual f$ is equivalent to $(1,1,1,0,\dots,0)$.
  \item In this case, $\dual f$ is equivalent to $(1,\dots,1,0)$
  \item In this case, $\dual f$ is equivalent to $(1,\dots,1,0,0)$
  \item In this case, $\dual f$ is equivalent to $(\alpha,\beta,0,\dots,0)$ for some $\alpha\ge\beta\ge0$.
  \item In this case, $\dual f$ is equivalent to $(2,1,\dots,1,0)$ (i.e., $f$ is self-dual.)
 \end{enumerate}
\end{proof}

\subsection{Proof of Results from Section~\ref{sect:ccdv many coefficients}}

In this section, we prove hardness of CCDV for all generators not of the form $(\alpha_1,\alpha_2,\alpha_3,\dots,\alpha_3,\alpha_4,\alpha_5,\alpha_6)$.

\subsubsection{Proof of Theorem~\ref{theorem:approval generalization ccdv}}

This result covers the ``approval-like'' behavior of generators $f$ satisfying $\alpha^m_3>\alpha^m_{m-3}$:

\begin{restatable}{restatableTheorem}{thmapprovalgeneralizationccdv}\label{theorem:approval generalization ccdv}
 Let $f$ be an polynomial-time uniform $\mathbb Q$-generator with $f(m)=(\alpha^m_1,\dots,\alpha^m_m)$ for each $m$. Then there is a polynomial-time computable function $g$ such that
 \begin{itemize}
  \item $g$ takes as input an instance $M$ of 3DM  with $\card M=n$\footnote{in this case, $n$ is in fact the number of \emph{tuples} in $M$, not the size of a binary representation.} and produces an instance $I$ of $f$-CCDV with $m=3n$ candidates,
  \item if $\alpha^m_3>\alpha^m_{2n}$, then: $M$ is a positive instance of 3DM if and only if $I$ is a positive instance of $f$-CCDV.
 \end{itemize}
\end{restatable}

\begin{proof}
 Let $M\subseteq X\cup Y\cup Z$ be a 3DM-instance. Let $n=\card M=3\card X$, and let $X=\set{c_1,\dots,c_{\card X}}$, $Y=\set{c_{\card X+1},\dots,c_{2\card X}}$, and $Z=\set{c_{2\card X+1},\dots,c_{3\card X}}$. For $c\in X\cup Y\cup Z$, let $k(c)$ denote the unique number $i$ with $c=c_i$, and let $r(c)$ be defined as $0$, $1$, or $2$, depending on whether $c\in X$, $c\in Y$, or $c\in Z$.
 
 Without loss of generality, assume $\alpha^m_m=0$ (subtract $\alpha^m_m$ from every coefficient otherwise). Clearly, it suffices to consider the case $\alpha^3_m>\alpha^m_{2n}$, we produce an arbitrary instance with the correct number of candidates otherwise. Hence, there is some maximal $t$ with $\alpha^m_{3+t}>\alpha^m_{2n}$. Since $f$ is polynomial-time uniform, $t$ can be computed in polynomial time. Clearly, $t<2n-3$. We construct an instance of $f$-CCDV with candidate set $X\cup Y\cup Z\cup\set{b_1,\dots,b_t,p}\cup D$ where $p$ is the preferred candidate, the $b_i$ and $d_i$ are additional \emph{blocking} and \emph{dummy} candidates, and $D$ is a set of $2n-3$ dummy candidates. Since $t<2n-3$, it follows that $\card D\ge3$. By construction, the total number of candidates is $m=3n$.
 
 We first consider the case $t\ge3$. For each $(x,y,z)\in M$, we add a vote as follows
 
 $$b_1>\dots>b_t>x>y>z>\RESTcandidates>S_{xyz}>p,$$
 
 where $S_{xyz}$ denotes the order $s_n>s_{n-1}>\dots s_2>s_1$, with candidates $x$, $y$, and $z$ replaced with dummy candidates from $D$ (recall that $\card D\ge3$) and $\RESTcandidates$ contains the remaining dummy candidates. We say that this vote \emph{covers} the candidates $x$, $y$, and $z$, and we will call the votes obtained from $M$ \emph{3DM-voters}, to distinguish them from the votes introduced below that serve to set up the necessary scores for our candidates. Using Lemma~\ref{lemma:coefficients realization}, we add setup voters ensuring that the relative scores of the candidates are as follows:
 
 \begin{itemize}
  \item $\score{p}=0$,
  \item $\score{b_i}=\card X\cdot\alpha^m_i$ for each $i\in\set{1,\dots,t}$,
  \item $\score{c}=\alpha^m_{t+r(c)}+(\card X-1)\alpha^m_{m-k(c)}$ for each $c\in X\cup Y\cup Z$,
  \item for each $d\in D$, $\score d$ is low enough to ensure that $d$ cannot win the election by deleting at most $\card X$ voters.
 \end{itemize}
 
 We first show that $M$ is a positive 3DM-instance if and only if $p$ can be made a winner of the election by deleting at most $\card X$ of the 3DM-voters. Below, we then argue that the scores can be set up in such a way that, in order to make $p$ win with at most $k$ deletions, the controller will always remove only 3DM-voters.
 
 First assume that $M$ is positive, i.e., there is a cover $C\subseteq M$ with $\card C=\card X$. We remove exactly the votes corresponding to the tuples in $C$, and show that $p$ wins the resulting election. Removing these voters changes the scores of the candidates as follows:
 
 \begin{itemize}
  \item $p$ is in the last position of all the removed votes; since $\alpha^m_m=0$, the score of $p$ remains $0$,
  \item each $b_i$ is in position $i$ in each of the $\card X$ removed votes and hence loses $\card X\cdot\alpha^m_i$ points, therefore $b_i$ ties with $p$,
  \item each $c\in X\cup Y\cup Z$ loses $\alpha^m_{t+r(c)}$ points from deleting the vote covering $c$, as well as $(\card X-1)\cdot\alpha^m_{m-k(c)}$ points from the votes not covering $c$. Therefore, the final score of $c$ is $0$ and $c$ also ties with $p$.
 \end{itemize}
 
 Hence, after removing the votes corresponding to the cover, $p$ is a winner of the election as required.
 
 For the converse, assume that $p$ wins the election after deleting at most $\card X$ 3DM-votes. Deleting fewer than $\card X$ of the 3DM-votes does not suffice, since each $b_i$ must lose at least $\card X\cdot\alpha^m_i$ points (note that $\alpha^m_i\ge\alpha^m_{3+t}>\alpha^m_{2n}\ge 0$), and loses exactly $\alpha^m_i$ points from each removed 3DM-vote. Therefore, exactly $\card X$ 3DM-votes are removed. We prove that these votes correspond to a cover. 
 
 Assume that this is not the case, then there is some $c\in X\cup Y\cup Z$ such that no vote covering $c$ is removed. Then $c$ appears in position $m-k(c)$ in each of the $\card X$ removed votes, and therefore loses $\card X\cdot\alpha^m_{m-k(c)}$ points. Therefore, the final score of $c$ is 
 
 $$\alpha^m_{t+r(c)}+(\card X-1)\alpha^m_{m-k(c)}-\card X\alpha^m_{m-k(c)}=\alpha^m_{t+r(c)}-\alpha^m_{m-k(c)}\ge\alpha^m_{t+3}-\alpha^m_{2n}>0=\score p.$$
 
 Hence, $c$ beats $p$ in the election, a contradiction.
 
 To show that if $p$ can be made a winner of the election, then $M$ is a positive 3DM instance, it remains to show that the setup voters can be chosen such that, in order to make $p$ win with deleting at most $k$ votes, the controller can remove only 3DM-votes.
 
 By construction, the candidates in $B=\set{b_1,\dots,b_t}$ must lose as many points as they gain in $\card X$ many 3DM-votes. Therefore, it suffices to construct the setup votes such that each of these votes gives fewer points to $B$ than each 3DM-vote (which give the maximal possible amount of points to $B$, as the candidates from $B$ are voted in the top $t$ many spots).
 
 Recall that in the proof of Lemma~\ref{lemma:coefficients realization}, all setup votes introduced are obtained from an arbitrary vote $\overrightarrow v_{init}$ by cycling the vote $\overrightarrow v_{init}$ and swapping the position of two candidates in $\overrightarrow v_{init}$. We use the following initial vote $\overrightarrow v_{init}$ in the construction from Lemma~\ref{lemma:coefficients realization}:
 
 $$b_1>d_1>d_2>p>d_3>s_1>b_2>s_2>s_3>s_4>s_5>b_3>\dots>b_t>\RESTcandidates,$$
 
 where $\RESTcandidates$ contains the remainder of the candidates from $X\cup Y\cup Z\cup D$. Recall that $t\ge3$, and we can without loss of generality assume that $n\ge6$. Clearly, by cycling the vote $v_{init}$ and swapping the position of $2$ candidates, no vote is obtained that has all candidates from $B$ among the top $t+3$ positions, as required to give the maximal number of points to $B$ (since $\alpha^m_{t+3}>\alpha^m_{t+4}$), and have the candidate $p$ \emph{not} occur in the first $t+3$ positions (since $p$ does not receive any points from the 3DM-votes).
  
 This concludes the proof for the case $t\ge3$.

 We now consider the case $t\leq 2$. By choice of $t$, we know that $\alpha^m_{t+3}>\alpha^m_{t+4}=\alpha^m_{2n}=\alpha^m_{2n}$ (since $m=3n$). Since $m=3n$ and $t\leq 2$, we know that $\card D=m-t-n-1=2n-t-1\ge 2n-3$, since $m=3n$ and $t\leq 2$. We can without loss of generality assume that $n\ge 3$, and thus $\card D\ge n$. For each tuple $(x,y,z)$ in the 3DM-instance, we produce a vote

 $$x>y>z>b_1>\dots>b_t>S\setminus\set{x,y,z}>D>p,$$

 where $S\setminus\set{x,y,z}$ contains these candidates in some arbitrary order. We again call these voters \emph{3DM-voters}. Using Lemma~\ref{lemma:coefficients realization}, we introduce additional setup votes ensuring that the relative scores of the relevant candidates are as follows: 

\begin{itemize}
 \item $\score{p}=0$,
 \item for each relevant $i$, we have $\score{b_i}=\card X\cdot\alpha^m_{3+i}$,
 \item for each $c\in X\cup Y\cup Z$, we have $\score c=\alpha^m_{r(c)}+(\card X-1)\cdot\alpha^m_{t+4}$,
 \item for each $d\in D$, $d$ cannot win the election by deleting at most $\card X$ voters.
\end{itemize}

In each of the 3DM-voters, the score of a candidate $c\in X\cup Y\cup Z$ that does not appear in the first three positions is $\alpha^m_{t+4}$, since $\alpha^m_{t+4}=\alpha^m_{2n}$, and there are at least $n$ dummy candidates in $D$.

By the same reasoning as in the case $t\ge3$, one can easily see that the given 3DM-instance is positive if and only if $p$ can be made a winner by removing at most $\card X$ 3DM-votes. It remains to show how to construct the setup votes such that, when the controller removes at most $\card X$ many votes, she can only make $p$ win the election when only 3DM-voters are removed.

With $R$, we denote the set of relevant candidates that $p$ has to defeat, i.e., $R=X\cup Y\cup Z\cup\set{b_i\suchthat 1\leq i\leq t}$. Similar as in the case $t\ge3$, removing 3DM-votes removes the maximal number of points from $R$, and from the setup of the scores it is clear that it is necessary to remove as many points from $R$ as possible with removing $\card X$ votes (namely, $\card X\cdot((\sum_{i=1}^{3+t}\alpha^m_i)+(n-3)\cdot\alpha^m_{4+t})$). Hence it suffices to construct the setup votes such that in each of these, the candidates $R$ have fewer points than $(\sum_{i=1}^{3+t}\alpha^m_i)+(n-3)\cdot\alpha^m_{4+t}$.

Similarly to the above case, this can be achieved by choosing the initial setup vote $\overrightarrow v_{init}$, as $$s_1>\Box>\Box>s_2>\Box>\Box>s_3>\Box>\Box>\dots>s_n>\Box>\Box,$$
where $\Box$ is a placeholder for an arbitrary candidate not from $X\cup Y\cup Z$. (Recall that $m=3n$, and $n=\card{X\cup Y\cup Z}$, hence there are sufficiently many candidates not from $S$ filling the $\Box$-positions.) Again, the setup votes are obtained from $\overrightarrow v_{init}$ by rotating an arbitrary number of positions and swapping at most two candidates. Clearly, every vote obtained like this has at least one candidate from $R$ in a position $i$ with $i>5\ge t+3$, hence each of these votes gives fewer points to the candidates from $R$ than the votes introduced for the tuples of 3DM above. As argued above, this concludes the proof.
\end{proof}

\subsubsection{Proof of Theorem~\ref{theorem:veto generalization}}

The next result generalizes $4$-veto. Note that for $3$-veto, CCDV and bribery can be solved in polynomial time due to Theorem~\ref{theorem:ccdv few coefficients ptime cases} and Theorem~\ref{theorem:lin bribery results list}, respectively.

\begin{restatable}{restatableTheorem}{thmvetogeneralization}\label{theorem:veto generalization}
 Let $f$ be an polynomial-time uniform $\mathbb Q$-generator with $f(m)=(\alpha^m_1,\dots,\alpha^m_m)$ for each $m$. Then there is a polynomial-time computable function $g$ such that
 \begin{itemize}
  \item $g$ takes as input an instance $M$ of 3DM  with $\card M=n$\footnote{again, $n$ is in fact the number of \emph{tuples}} and produces an instance $I$ of $f$-CCDV with $m=3n$ candidates,
  \item if $\alpha^m_{2n}>\alpha^m_{m-3}$, then: $M$ is a positive instance of 3DM if and only if $I$ is a positive instance of $f$-CCDV.
 \end{itemize}
\end{restatable}

\begin{proof}
 As in the proof of Theorem~\ref{theorem:approval generalization ccdv}, we assume that $\alpha^m_m=0$. Let $M\subseteq X\times Y\times Z$ be a 3DM-instance with $\card M=3\card X=n$; we can without loss of generality assume that $n\ge4$. We use the values $r(c)$ and $k(c)$ as defined in the previous proof. We construct an $f$-CCDV instance with candidate set $X\cup Y\cup Z\set{b_1,\dots,b_n}\cup\set p\cup D$, where $D$ is a set of dummy candidates with $\card D=n-1$. By construction, this instance has $3m$ candidates. We again only consider the case that $\alpha^m_{2n}>\alpha^m_{m-3}$. Since $n\ge4$, we know that $\card D\ge 3$. For each $(x,y,z)\in M$, we introduce a voter voting as follows:
 
 $$b_1>b_2>\dots>b_n>S_{xyz}>\RESTcandidates>z>y>x>p,$$
 
 where $S_{xyz}$ is the sequence $s_1>\dots>s_n$ with each $x$, $y$, and $z$ replaced by a candidate from $D$ (recall that $\card D\ge 3$), and $\RESTcandidates$ contains the remaining candidates from $D$ (in an arbitrary order). We again call these votes \emph{3DM-votes}.
 
 Using Lemma~\ref{lemma:coefficients realization}, we introduce additional votes ensuring that the scores are as follows:
 
 \begin{itemize}
  \item for each $i\in\set{1,\dots,n}$, we have $\score{b_i}=\card X\cdot \alpha^m_i$,
  \item for each $c\in X\cup Y\cup Z$, we have $\score c=(\card X-1)\alpha^m_{n+k(c)}+\alpha^m_{m-r(c)-1}$,
  \item $\score p=0$,
  \item for each candidate $d\in D$, the score is so low that $d$ cannot win the election when at most $\card X$ votes are removed.
 \end{itemize}
 
 We show that in the resulting election, $p$ can be made a winner by removing at most $\card X$ votes if and only if $M$ is a positive instance of 3DM. As in the proof of Theorem~\ref{theorem:approval generalization ccdv}, we first show this claim where we only consider 3DM-votes, and then argue how the setup votes can be constructed in such a way that the controller will always remove 3DM-votes only.
 
 First assume that $M$ is a positive 3DM-instance, then is a cover $C\subseteq M$ with exactly $\card X$ elements. We delete exactly the votes corresponding to $C$. Then:
 
 \begin{itemize}
  \item each $b_i$ loses exactly $\card X\cdot\alpha^m_i$ points, therefore having $0$ points afterwards,
  \item each $c\in X\cup Y\cup Z$ loses exactly $\alpha^m_{m-r(c)-1}$ points from the one vote corresponding to a tuple covering $c$, and $(\card X-1)\alpha^m_{n+k(c)}$ points from the remaining $\card X-1$ votes, hence $c$ ends up with $0$ points,
  \item the score of $p$ is not affected, since $\alpha^m_m=0$, hence $p$ still has $0$ points after the votes are removed.
 \end{itemize}

 Therefore, all relevant candidates tie and $p$ is indeed a winner of the election as claimed.
 
 For the converse, assume that $p$ can be made a winner with removing at most $\card X$ of the 3DM-votes. Since $b_1$ must lose $\ell\cdot\alpha^m_1$ many points, at least $\card X$ votes must be removed, hence exactly $\card X$ votes are removed. Let $C$ be the subset of $M$ corresponding to the removed votes, we claim that $C$ is a cover. Assume indirectly that $C$ is not a cover, then, since $\card C=\card X$, there is some $c\in X\cup Y\cup Z$ appearing in at least two tuples of $C$. Then, $c$ loses at most
 
 $$(\card X-2)\cdot\alpha^m_{n+k(c)}+2\alpha^m_{m-r(c)-1}$$ many points, hence $c$ has at least
 
 $
 \begin{array}{lll}
        (\card X-1)\alpha^m_{n+k(c)}+\alpha^m_{m-r(c)-1}-(\card X-2)\cdot\alpha^m_{n+k(c)}-2\alpha^m_{m-r(c)-1} 
 & =   &\alpha^m_{n+k(c)}-\alpha^m_{m-r(c)-1} \\
 & \ge &\alpha^m_{2n}   -\alpha^m_{m-3} \\
 & > & 0,
 \end{array}$
 
 which is a contradiction. Hence $C$ is a cover as required.
 
 It remains to show how to construct the setup votes such that the deletion of any set of $\card X$ many votes which does not only include 3DM-votes fails to ensure that $p$ wins the election. For this, it suffices to construct the initial vote $\overrightarrow v_{init}$ in the following way:
 
 $$\Box>\Box>b_1>\Box>\Box>b_2>\Box>\Box>\dots >b_{n-1}>\Box>\Box>b_n$$
 
 (recall that $m=3\cdot n$, hence there are enough candidates to fill the $\Box$-positions). Since all votes in the setup votes are constructed from $\overrightarrow v_{init}$ by swapping two candidates or cycling, each setup vote has a candidate $b_i$ in one of the last four votes. Since $\alpha^m_{m-3}<\alpha^m_{2n}\leq\alpha^m_{n}$, this implies that in a non-3DM-vote, the sum of the scores of $B=\set{b_1,\dots,b_n}$ is lower than  in a 3DM-vote (where each $b_i$ receives at least $\alpha^m_n$ points). From the above it follows that in order for $p$ to win, the set $B$ must lose exactly $\ell$ times the number of points they lose when deleting a single 3DM-vote. Hence in order to make $p$ win with deleting at most $\ell$ votes, only 3DM-votes can be deleted. This completes the proof.
\end{proof}

\subsubsection{Proof of Theorem~\ref{theorem:many different coefficients}}

We now combine the above two results to show that all ``many coefficients''-cases of CCDV are NP-complete:

\theoremmanydifferentcoefficients*

\begin{proof}
 Clearly, if the condition is true for some $m_0$, then it remains true for each $m\ge m_0$. We use a reduction from 3DM. Let $M$ be an instance of 3DM, let $n=\card{M}$, let $m=3n$. Without loss of generality, assume $n\ge4$, and $3n\ge m_0$. Since $\alpha^m_3>\alpha^m_{m-3}$, we know that $\alpha^m_3>\alpha^m_{2n}$ or $\alpha^m_{2n}>\alpha^m_{m-3}$ must hold, and we can determine, in polynomial time, which of these cases holds, since $f$ is polynomial-time uniform. Let $g$ be the reduction from Theorem~\ref{theorem:approval generalization ccdv} in the first case, and the one from Theorem~\ref{theorem:veto generalization} in the second case. In both cases, $f(M)$ is a positive instance of $f$-CCDV if and only if $M$ is a positive 3DM-instance. Since $g$ is polynomial-time computable, this completes the proof.
 In both cases, $g(M)$ is 
\end{proof}

\subsection{Proofs of Results from Section~\ref{sect:ccdv hardness:generic ccdvstar reductions}}

In this section, we present hardness results for CCDV that are obtained by a direct reduction from \ccdvstar. In these cases, the difficulty in the reduction is the construction of an appropriate set of ``setup votes'' that mirror the ``undeletable'' voters from the \ccdvstar-instance. Since in CCDV, no voter is immune from deletion, we need to setup these votes such that they are ``unattractive'' to delete. More precisely, in the two following proofs deleting one of the ``setup votes'' will immediately imply that the intended candidate $p$ does not win the election, since deleting such a vote implies a ``chain reaction'' of further required deletions that exceeds the budget available to the controller. In the proof of Theorem~\ref{theorem:alpha 0 dots 0 alpha < beta ccdv}, we use an exponential construction (but only logarithmically many steps of it in order to be able to perform the reduction in polynomial time), whereas the proof of Theorem~\ref{theorem:alpha 0 dots 0 beta < alpha ccdv} relies on a simpler linear construction.

\subsubsection{Proof of Theorem~\ref{theorem:alpha 0 dots 0 alpha < beta ccdv}}

\theoremccdvhardnessalphazweodotsalphasmallerbeta*

\begin{proof}
 The dual generator to $f$ is $\dual f=(\beta,0,\dots,0,-\alpha)$. Due to \cite{HemaspaandraHemaspaandraSchnoor-CCAV-AAAI-2014}, $\dual f$-CCAV is NP-complete, hence due to Proposition~\ref{prop:ccdvstar poly equiv ccav}, $f$-\ccdvstar is NP-complete as well. Therefore, it suffices to prove that $f$-\ccdvstar reduces to $f$-CCDV.
 
 Let $I'$ be an instance of $f$-\ccdvstar, consisting of a set $D$ of deletable votes, a set $R$ of votes that cannot be deleted, a preferred candidate $p$, and a budget $k$ indicating the number of votes that the controller can delete. In the following, we will denote a vote $c_1>c_2>\dots>c_{n-1}>c_n$ simply as $c_1>c_n$, since the remaining candidates all receive $0$ points from this vote and thus are not relevant. Without loss of generality, we assume the following:
 
 \begin{itemize}
  \item There is no deletable vote of the form $p>c$ for any candidate $c$: Clearly, the controller will never delete such a vote, hence we can move all these votes to the set $R$ of undeletable votes without changing whether the instance is positive.
  \item There is no deletable vote of the form $c>p$ for any candidate $c$: Since $\beta>\alpha$, removing such a vote (letting $p$ gain at least $\beta$ points against every candidate) is more profitable for the controller than removing any vote of the form $(c_1>c_2)$ (letting $p$ gain $\alpha$ points against a single candidate). We therefore can simply remove all of these votes and decrease the budget accordingly. (In case that there are more votes of this kind than the budget allows, the controller will only remove such votes, and we can use the obvious greedy strategy to decide the instance $I'$, producing a fixed positive/negative instance $I$ depending on whether $I'$ is positive/negative.)
 \end{itemize}
 
 As discussed earlier, the task in a reduction from $f$-\ccdvstar to $f$-CCDV is to convert the undeletable votes from $I'$ into votes that give the same relative points to all relevant candidates, and which cannot be removed by the controller when trying to make $p$ win the election with deleting at most $k$ votes.
 
 Without loss of generality, we assume that $\mathop{gcd}(\alpha,\beta)=1$, and hence there are natural numbers $A,B$ with $1=A\cdot\alpha-B\cdot\beta$.\footnote{From numbers $C$ and $D$ with $1=C\cdot\beta-D\cdot\alpha$, we obtain $A$ and $B$ as required as $A=t\beta-D$ and $B=t\alpha-C$ for sufficiently large $t$ such that these numbers are positive.} Therefore, it suffices to show how to construct votes that the controller will not remove, and which add $\alpha$ points, respectively remove $\beta$ points from the candidates (relative to $p$). Adding $1$ point to a candidate $c\neq p$ then can be implemented by using $A$ groups of votes that each add $\alpha$ points to $c$, and then $B$ groups of votes each removing $\beta$ points. We can assume that indeed all candidates must gain points relatively to $p$ by adding votes of the form $p>c$ for arbitrary candidates $c$ to ensure that $p$ has a sufficient headstart over the remaining candidates. Clearly, such votes will never be deleted by the controller. We do this at least twice to ensure that $\score p\ge2\alpha$.
 
 To reduce the score of a candidate $c$ by $\beta$, we simply add a vote $x>c$ for a fresh dummy candidate $x$, who appears in the $0$-point segment of all other votes. Then, clearly, $p$ beats $x$, and this vote will not be deleted by the controlled. Therefore, it remains to show how to construct a set of votes that add $\alpha$ points to a candidate $c\neq p$ and which the controller will not delete. Since $p$ gets $0$ points in all of the deletable votes (and in all votes we introduce below that replace the votes in $R$), we can compute the score $\score{p}$ that $p$ will have in the final election (this score is unaffected by the controller's delete actions).
 
 Since $p$ never appears in the last position of any vote, $p$ gains either $\alpha$ or $0$ points from any present vote. Hence, $\score p=N_p\alpha$ for some natural number $N_p\ge2$. In the following, let $B=\left\lceil\frac\beta\alpha\right\rceil$. Since $\alpha<\beta$, we know that $B\ge2$. To let a candidate $c\neq p$ gain $\alpha$ points relatively to $p$, we proceed as follows: We add dummy candidates $d_1,\dots,d_\ell$ with $\ell=\left\lceil\log_{\frac\beta\alpha}(k)\right\rceil+2$ by placing them in the positions awarding $0$ points of every existing vote. These dummy candidates are only used for the process of adding $\alpha$ points to $c$ once, further additions in the sequel (even to the same candidate $c$) use a new set of dummy candidates. We now add a single vote $c>d_1$, which lets $c$ gain $\alpha$ points relative to $p$. It remains to add votes ensuring that the vote $c>d_1$ cannot be removed by the controller. These votes will set up the scores of the dummy candidates as follows:

 \begin{itemize}
   \item Each $d_i$ for $1\leq i\leq\ell-1$ ties with $p$,
   \item the only way to make $d_i$ lose points (relative to $p$) is to remove votes $d_i>d_{i+1}$, which then lets $d_{i+1}$ gain points (relative to $p$).
 \end{itemize}
 
 Hence removing the vote $c>d_1$, which lets $d_1$ gain $\beta$ points relative to $p$ requires the controller to remove votes of the form $d_1>d_2$, which in turn lets $d_2$ gain $\beta$ points for each removal, this process continues for $d_i$ with $i\ge 2$. In this way, removing $c>c_1$ triggers a ``chain reaction'' of additional necessary removals. We will set up the votes in such a way that this process forces the controller to remove more votes than her budget allows. This implies that, in fact, she cannot remove the vote $c>d_1$, as required. The numbers of votes we need to add in each step, and consequently the numbers of points that the candidates gain, will essentially grow exponentially in $\frac\beta\alpha$. Since the controller can only remove a polynomial number of votes, we only require logarithmically many steps, yielding a construction that can be performed in polynomial time as required. Specifically, we use the following votes:
 
 \begin{description}
  \item[induction start: score of $d_1$] After adding the vote $c>d_1$, $d_1$ currently has $-\beta$ points.
     \begin{itemize}
      \item We add $(\alpha-1)$ many votes of the form $x>d_1$ for a fresh (as above---the controller will never delete these votes as $p$ strictly beats $x$) dummy candidate $x$, which further decrease the score of $d_1$ to  $-\alpha\beta$ (this ensures that the score of $d_1$ is a multiple of $\alpha$, recall that $\score p=N_p\alpha$). 
      \item We add $N_1:=N_p+\beta$ votes of the form $d_1>d_2$, each of these votes lets $d_1$ gain $\alpha$ points relative to $p$, hence after this $\score{d_1}=-\alpha\beta+(N_p+\beta)\alpha=N_p\alpha=\score p$.
     \end{itemize}
  \item[induction step: score of $d_{i+1}$] After adding the $N_i$ votes $d_i>d_{i+1}$, $d_{i+1}$ currently has $-N_i\beta$ points.
     \begin{itemize}
      \item Add $m_{i+1}:=(-N_i)\mod\alpha$ many votes $x>d_{i+1}$ for a new dummy candidate $x$ as above, then $d_{i+1}$'s score is $-(N_i+m_{i+1})\beta$. By construction, this is a multiple of $\alpha$, namely $-\left\lceil\frac{N_i}{\alpha}\right\rceil\alpha\beta=-M_i\alpha$ (with $M_i=\left\lceil\frac{N_i}\alpha\right\rceil\beta$). Since $m_{i+1}\in\set{0,\dots,\alpha-1}$, the number of votes added in this step is bounded by the constant $\alpha$.
      \item Add $N_{i+1}:=N_p+M_i$ votes of the form $d_{i+1}>d_{i+2}$, then $\score{d_{i+1}}=N_p\alpha=\score p$.
     \end{itemize}
 \end{description}
 
 We claim that the number of votes added is polynomial in the input. For this, we first prove inductively that $N_i\leq N_p+i\beta N_pB^{i-1}$ for all relevant $i$.
 
 \begin{description}
  \item[induction start] Since $N_p\ge1$, we have that $N_1=N_p+\beta\leq N_p+\beta N_pB^0$, hence for $i=1$ the claim holds.
  \item[induction step] Due to the above, $N_{i+1}=N_p+M_i=N_p+\left\lceil\frac{N_i}\alpha\right\rceil\beta$. We have
  
  \begin{tabular}{lclr}
   $N_{i+1}$ & $=$    & $N_p+\left\lceil\frac{N_i}\alpha\right\rceil\beta$ \\
             & $\leq$ & $N_p+\left(\frac{N_i}\alpha+1\right)\beta$ \\
             & $=$    & $N_p+\beta+N_i\frac\beta\alpha$ \\
             & $\leq$ & $N_p+\beta+N_iB$ \\
             & $\leq$ & $N_p+(N_i+\beta)B$ & ($B\ge1$) \\
             & $\leq$ & $N_p+\left(N_p+i\beta N_pB^{i-1}+\beta\right)B$ & (induction) \\
             & $\leq$ & $N_p+\left(N_pB^{i-1}+i\beta N_pB^{i-1}+\beta B^{i-1}\right)B$ & ($B^{i-1}\ge1$) \\
             & $=$    & $N_p+(N_p+i\beta N_p+\beta)B^i$ \\
             & $\leq$ & $N_p+(N_p\beta+i\beta N_p)B^i$ & ($N_p\beta\ge N_p+\beta$, as $N_p,\beta\ge2$) \\
             & $=$    & $N_p+(i+1)N_p\beta B^i,$
  \end{tabular}
  
  as required. This completes the induction.
 \end{description}
 
 We now use the above bound on $N_i$ to show that in fact, only polynomially many votes are added.
Recall that $\ell=\left\lceil\log_{\frac\beta\alpha}(k)\right\rceil+2$. Therefore, for each relevant $i$, we have that
 
 \begin{tabular}{lclr}
  $N_i$ & $\leq$ & $N_p+i\beta N_pB^{i-1}$ \\
        & $\leq$ & $N_p+\ell\beta N_pB^\ell$ & ($i\leq\ell$) \\
        & $\leq$ & $N_p+\ell\beta N_pB^{\log_{\frac\beta\alpha}(k)}B^3$ & (definition of $\ell$) \\
        & $=$    & $N_p+\ell\beta N_p k^{\log_{\frac\beta\alpha}(B)}B^3.$ & ($a^{\log_b(x)}=x^{\log_b(a)}$ for all $a,b,x$)
 \end{tabular}
 
 To see that $a^{\log_b(x)}=x^{\log_b(a)}$ for all $a,b,x$, recall that $\log_a(x)=\frac{\log_b(x)}{\log_b(a)}$. From this we get $\log_b(x)=\log_a(x)\cdot\log_b(a)=\log_a(x^{\log_b(a)})$, and by building the power to the base $a$, this finally implies $x^{\log_b(a)}=a^{\log_b(x)}$ as required.

 Since $N_p$, $\ell$, and $k$ are polynomial in the instance and $\alpha$, $\beta$, and $B$ are constant, this shows that the number of votes required is in fact polynomial in the size of $I'$.
 
 We now show that the above construction in fact enforces that in a successful control operation, only votes corresponding to deletable votes from the instance $I'$ are removed. First note that, since all dummy candidates tie with $p$, and there is no removable vote in which $p$ gains any points, the controller will not remove any vote of the form $d_i>d_{i+1}$, unless she also removes a vote of the form $c>d_1$ as introduced above. Hence it suffices to show that when the controller removes a vote $c>d_1$, she cannot make $p$ win the election by removing at most $k-1$ additional votes.
 
 We prove inductively that if the controller removes the vote $c>d_1$, then she has to remove at least $\left(\frac\beta\alpha\right)^i$ votes of the form $d_i>d_{i+1}$.
 
 \begin{description}
  \item[induction start.] If the controller removes the vote $c>d_1$, then candidate $d_1$ (who previously tied with $p$) gains $\beta$ points against $p$. The only way to make $d_1$ lose these points again (relative to $p$) is to remove votes that vote $d_1$ ahead of $p$. The only votes of this form are the votes $d_1>d_2$, removing such a vote lets $d_1$ lose $\alpha$ points against $p$. Therefore, at least $\frac\beta\alpha$ such votes must be removed.
  
  \item[induction step.] Assume that inductively, $\left(\frac\beta\alpha\right)^i$ votes of the form $d_i>d_{i+1}$ are removed. Each such removal lets $d_{i+1}$ gain $\beta$ points against $p$, hence $d_{i+1}$ gains at least $\left(\frac\beta\alpha\right)^i\beta=\frac{\beta^{i+1}}{\alpha^i}$ points. Since $d_{i+1}$ initially ties with $p$, the controller must remove votes to let $d_{i+1}$ lose this number of points. Analogously to the case $i=1$, the only votes allowing this are votes of the form $d_{i+1}>d_{i+2}$, each removal of one of these votes lets $d_{i+1}$ lose $\alpha$ points relative to $p$. Therefore, at least $\frac1\alpha\frac{\beta^{i+1}}{\alpha^i}=\left(\frac\beta\alpha\right)^{i+1}$ must be removed, as claimed.
 \end{description}
 
 In particular, for $i=\ell-1$, we show that the controller must remove at least $k$ votes of the form $d_{i}>d_{i+1}$, which she cannot do, as her budget is $k$, and she already removed the vote $c>d_1$. To see that at least $k$ such votes must be removed, recall that, due to the above, at least $\left(\frac\beta\alpha\right)^i$ of these votes must be removed. With $i=\ell-1$, it follows that
 
 $$\left(\frac\beta\alpha\right)^{\ell-1}=\left(\frac\beta\alpha\right)^{\left\lceil\log_{\frac\beta\alpha}(k)+1\right\rceil}\ge\left(\frac\beta\alpha\right)^{\log_{\frac\beta\alpha}(k)+1}\ge k,$$
 
 as required. This concludes the proof.
\end{proof}

\subsubsection{Proof of Theorem~\ref{theorem:alpha 0 dots 0 beta < alpha ccdv}}

\theoremalphazerodotszerobetasmalledalphaccdv* 

\begin{proof}
 We proceed analogously to the case $\alpha<\beta$ treated in Theorem~\ref{theorem:alpha 0 dots 0 alpha < beta ccdv}. With the exact same argument as in that proof, it suffices to show how to add a group of votes that add $\alpha$ many points to a candidate $c\neq p$. We can again assume that $p$ gains $0$ points in the deletable votes from the original instance $I'$, since the hardness proof for $\dual f$-CCAV of \cite{HemaspaandraHemaspaandraSchnoor-CCAV-AAAI-2014} ensures this. We can therefore also assume, just as above, that $\score p=N_p\alpha$ for some $N_p\in\mathbb N$, and this score does not change with the controller's actions.
 
 To let a candidate $c\neq p$ gain $\alpha$ points relative to $p$, we again add dummy candidates and a single vote $c>d_1$. Again, if the controller removes this setup-vote, $d_1$ gains $\beta$ points (relative to $p$). Since $\alpha>\beta$, the construction is much simpler than in the proof of Theorem~\ref{theorem:alpha 0 dots 0 alpha < beta ccdv}: Removing a single vote $d_1>d_2$ suffices to make $d_1$ lose the points gained by removing the vote $c>d_1$ is removed. Analogously, removing a single vote $d_{i+1}>d_{i+2}$ undoes the $\beta$ points gained by $d_{i+1}$ when a vote $d_i>d_{i+1}$ is removed. Instead of the exponential/logarithmic process of the proof of Theorem~\ref{theorem:alpha 0 dots 0 alpha < beta ccdv}, we thus can use a simple linear chain: We add $k+1$ dummy candidates, and can easily ensure that each of them has scores as required. In the proof of Theorem~\ref{theorem:alpha 0 dots 0 alpha < beta ccdv}, this lead to a number of votes approximately $\left(\frac\beta\alpha\right)^i$ in step $i$. Since in the current proof, $\frac\beta\alpha$ is smaller than $1$, we do not run into a similar exponentially-growing process this time.
 
 To be more precise, we add dummy candidates $d_1,\dots,d_{k+1}$, and votes as follows:
 
 \begin{itemize}
  \item a single vote $c>d_1$,
  \item for each $i\in\set{1,\dots,k}$, $N_i$ many votes $d_i>d_{i+1}$ (the number $N_i\ge1$ will be defined below).
 \end{itemize}

 We set up the scores of $d_i$ such that $d_i$ currently does not beat $p$, but does so after gaining $\beta$ points, i.e., 
 
 $$N_p\alpha-\beta<\score{d_i}\leq N_p\alpha.$$
 
 Hence, the score of $d_i$ must lie in an interval of length $\beta$ (recall that $\beta<\alpha$). This can be achieved by adding votes as follows:
 
 \begin{description}
  \item[induction start: score of $d_1$.] After adding the vote $c>d_1$, the score of $d_1$ is $-\beta$. To ensure that the score is at least $N_p\alpha-\beta+1$, we add $N_1:=N_p+1$ votes of the form $d_1>d_2$; after this, the score of $d_1$ is $N_p\alpha+\alpha-\beta$. We then add (constantly many) votes of the form $x>d_1$, each removing $\beta$ points from $d_1$, to move the score of $d_1$ into the required interval of length $\beta$.
  \item[induction step: score of $d_{i+1}$.] After adding the $N_i$ many votes $d_i>d_{i+1}$, the score of $d_{i+1}$ is $-N_i\beta$. To ensure that the score is at least $N_p\alpha-\beta+1$, we add $N_{i+1}=\left\lceil N_i\frac\beta\alpha\right\rceil+N_p$ votes of the form $d_{i+1}>d_{i+2}$. We then add constantly many votes of the form $x>d_{i+1}$ to move the score of $d_{i+1}$ into the interval of length $\beta$. 
 \end{description}
 
 It is easy to see that the number of added votes is polynomial: $N_1$ is clearly polynomial, and $N_{i+1}$ is essentially obtained from $N_i$ by multiplying with $\frac\beta\alpha<1$ and adding $N_p$.
 
 It remains to show, analogously to the proof of Theorem~\ref{theorem:alpha 0 dots 0 alpha < beta ccdv}, that the controller cannot remove the vote $c>d_1$ and still make $p$ win with at most $k-1$ further deletions. By construction, removing the vote $c>d_1$ lets $d_1$ beat $p$ in the election, and to reduce the score of $d_1$ relative to $p$, a vote $d_1>d_2$ must be removed. Inductively, when a vote $d_i>d_{i+1}$ is removed, then $d_{i+1}$ gains $\beta$ points relative to $p$ and thus beats $p$, hence a vote $d_{i+1}>d_{i+2}$ must be removed. Therefore, a sequence of $k$ many votes must be removed in addition to the vote $c>d_1$, and hence the controller cannot remove the vote $c>d_1$.
\end{proof}

\subsubsection{Corollary for $f=(\alpha,0,\dots,0,\beta)$}

As a corollary of the above two results, we obtain the following characterization of the complexity of CCDV for all generators of the form $(\alpha_1,\alpha_2,\dots,\alpha_2,\alpha_3)$:

\begin{restatable}{restatableCorollary}{corollaryalphazerodotszerobetaccdvdichotomy}
 \label{corollary:alpha 0 dots 0 beta ccdv dichotomy}
 Let $f=(\alpha_1,\alpha_2,\dots,\alpha_2,\alpha_3)$ be a generator with $\alpha_1>\alpha_2>\alpha_3$. If $f$ is equivalent to $(2,1,\dots,1,0)$, then $f$-CCDV can be solved in polynomial time, otherwise, $f$-CCDV is NP-complete.
\end{restatable}

\begin{proof}
 Clearly, $f$ is equivalent to a generator of the form $(\alpha,0,\dots,0,-\beta)$ with $\alpha,\beta>0$. The polynomial-time case follows from Theorem~\ref{theorem:ccdv few coefficients ptime cases}. If $f$ is not equivalent to $(2,1,\dots,1,0)$, then $f$ also is not equivalent to $(1,0,\dots,0,-1)$. Therefore, $\alpha\neq\beta$, and NP-hardness follows from Theorem~\ref{theorem:alpha 0 dots 0 alpha < beta ccdv} or Theorem~\ref{theorem:alpha 0 dots 0 beta < alpha ccdv}.
\end{proof}

\subsection{Proofs of Results from Section~\ref{sect:ccdv hardness by inspecting ccav reduction}}

We now present the proofs of the results stated in Section~\ref{sect:ccdv hardness by inspecting ccav reduction}, these are hardness proofs for CCDV obtained essentially by a reduction from those \ccdvstar-cases that arise in the hardness proofs of the corresponding dual generators from~\cite{HemaspaandraHemaspaandraSchnoor-CCAV-AAAI-2014}. (The proof of Theorem~\ref{theorem:0 dots 0 -gamma -beta -alpha CCDV} is contained completely in the main paper.)

\subsubsection{Proof of Theorem~\ref{theorem:0 minus alpha5 minus alpha1 dots minus alpha1 CCDV hardness}}

\theoremzerominusalphafiveminusalphaonedotsminusalphaoneccdvhardness*

\begin{proof}
 We can equivalently write $f$ as $f=(0,-\alpha_5,-\alpha_1,\dots,-\alpha_1)$ with $0<\alpha_5<\alpha_1$. Similarly as in the proof of Theorem~\ref{theorem:0 dots 0 -gamma -beta -alpha CCDV} above, we consider the hardness proof of $\dual f$-CCAV from~\cite{HemaspaandraHemaspaandraSchnoor-CCAV-AAAI-2014}, where $\dual f=(\alpha_1,\dots,\alpha_1,\alpha_5,0)$. Their proof gives us, by applying Proposition~\ref{prop:ccdvstar poly equiv ccav}, a hardness proof of $f$-\ccdvstar with a reduction from 3DM. We consider the votes introduced and points required by this reduction. Following the proof in~\cite{HemaspaandraHemaspaandraSchnoor-CCAV-AAAI-2014}, and reversing the votes according to~\ref{prop:ccdvstar poly equiv ccav}, we obtain the following reduction. The set of candidates is the same as in the above proof of Theorem~\ref{theorem:0 dots 0 -gamma -beta -alpha CCDV}. For each $S_i=(x,y,z)\in M$, the following four votes are introduced (we only state the first two candidates, since the remaining ones all get the same number of points from the vote):
 
 \begin{itemize}
  \item $S_i>S_i'>\dots$,
  \item $x>S_i>\dots$,
  \item $y>S_i>\dots$,
  \item $z>S_i'>\dots$.
 \end{itemize}

 The scores of the candidates before the controller's action are as follows:
 
 \begin{itemize}
  \item $\score p=0$,
  \item $\score c=\alpha_1$ for each $c\in X\cup Y\cup Z$,
  \item $\score{S_i}=\min(\alpha_1,2(\alpha_1-\alpha_5))$ for each $S_i\in M$,
  \item $\score{S_i'}=\alpha_1-\alpha_5$ for each $S_i\in M$.
 \end{itemize}
 
 We obtain these scores by adding votes of the form $x>d>\dots$ and $d>x>\dots$ for a relevant candidate $x$ and a dummy candidate $d$, where we use a new dummy candidate $d$ each time, who gets $-\alpha_1$ points from all remaining votes. Clearly, using enough ``setup''-votes, we can ensure that the dummy candidates do not win the election.
 
 The budget awarded to the controller by the reduction from~\cite{HemaspaandraHemaspaandraSchnoor-CCAV-AAAI-2014} is $n+2k=5k$. We now show that, in order to make $p$ win the election when the scores are as constructed above and with at most $5k$ removals of voters, the controller will only remove that have candidates from $X\cup\ Y\cup Z\cup\set{S_i, S_i'\suchthat S_i\in M}$ in the first two positions. This then in particular shows that none of the above-introduced setup votes can be deleted by the controller. For this, we fix a set $V$ of votes with $\card V\leq 5k$ such that $p$ wins the election after the votes in $V$ are removed. Then, in particular, after removing the votes in $V$, $p$ is not beaten by any candidate in $X\cup Y\cup Z$. We define the following sets:
 
 \begin{itemize}
  \item Let $C_1$ contain all candidates $c\in X\cup Y$ such that $V$ contains at least one vote having $c$ in the first position,
  \item let $C_2$ contain the remaining candidates from $X\cup Y$, i.e., $C_2=(X\cup Y)\setminus C_1$.
  \item Similarly, let $D_1$ contain all $c\in Z$ such that $V$ contains at least one vote having $c$ in the first position,
  \item let $D_2$ contain the remaining candidates from $Z$, i.e., $D_2=Z\setminus D_1$.
  \item Similarly, let $E_1$ contain all $S_i\in M$ such that $V$ contains at least one vote having $S_i$ in the first position,
  \item let $E_2$ contain the remaining sets $S_i$, i.e., $E_2=M\setminus E_1$.
 \end{itemize}
 
 By definition and since $X\cap Y=\emptyset$, we obtain
 
 \begin{itemize}
  \item $\card{C_1}+\card{C_2}=\card X+\card Y=2k$,
  \item $\card{D_1}+\card{D_2}=\card Z=k$,
  \item $\card{E_1}+\card{E_2}=\card M=3k$.
 \end{itemize}
 
 In each of the $5k$ votes in $V$, only at most the first two candidates gain points relative to $p$, hence there are only $10k$ positions in $V$ which ensure that after removing, the candidates lose points relative to $p$. Since each $S_i'$ must lose points relative to $p$ (as $\alpha_1>\alpha_5$), there are at most $7k$ positions in $V$ available for candidates in $C_1\cup C_2\cup D_1\cup D_2\cup E_1\cup E_2$.
 
 Each candidate in $X\cup Y\cup Z$ must lose $\alpha_1>\alpha_1-\alpha_5$ points relative to $p$; each candidate $S_i$ must lose $\min(\alpha_1,2(\alpha_1-\alpha_5))>\alpha_1-\alpha_5$ points relative to $p$. Therefore, each of these candidates that does not appear in the first position of a vote in $V$ must appear in the second position of at least two votes from $V$. Since there are only $7k$ positions available for these candidates, it follows that
 
 $$\card{C_1}+2\card{C_2}+\card{D_1}+2\card{D_2}+\card{E_1}+2\card{E_2}\leq 7k.$$
 
 With $\card{C_2}=2k-\card{C_1}$, $\card{D_2}=k-\card{D_1}$ and $\card{E_2}=3k-\card{E_1}$, we obtain
 
 \medskip
 
 \begin{tabular}{llll}
 &                   $\card{C_1}+2(2k-\card{C_1})+\card{D_1}+2(k-\card{D_1})+\card{E_1}+2(3k-\card{E_1})$ & $\leq$ & $7k$ \\
 $\Leftrightarrow$ & $4k-\card{C_1}+2k-\card{D_1}+6k-\card{E_1}$                                          & $\leq$ & $7k$ \\
 $\Leftrightarrow$ & $12k-(\card{C_1}+\card{D_1}+\card{E_1})$                                             & $\leq$ & $7k$ \\
 $\Leftrightarrow$ & $\card{C_1}+\card{D_1}+\card{E_1}$                                                   & $\geq$ & $5k$.
 \end{tabular}
 
 \medskip
 
 Clearly, we also have $\card{C_1}+\card{D_1}+\card{E_1}\leq 5k$, since in the $5k$ votes from $V$, only $5k$ first positions are available. Therefore, we obtain
 
 $$\card{C_1}+\card{D_1}+\card{E_1}=5k.$$
 
 Let $e_2^*$ denote the number of votes in $V$ that have a candidate of the form $S_i$ in the second position. Since every candidate from $E_2$ appears in no first position of any vote in $V$, but must lose more than $\alpha_1-\alpha_5$ points relatively to $p$, every such candidate must appear in the second position of at least two votes in $V$. Therefore, $e_2^*\ge2\card{E_2}$. We show that, in fact, the two values are equal.
 
 Therefore, assume indirectly that $\card{E_2}<\frac12e_2^*$. We consider the number of relevant positions in votes in $V$ that the candidates in $X\cup Y\cup Z\cup\set{S_i\suchthat S_i\in M}$ require. Since each candidate in $E_1$ requires one first position, and additionally, $e_2^*$ second positions are required by the candidates of the form $S_i$, and each of the $3k$ candidates $S_i'$ required at least one position, and we assumed that $e_2^*>2\card{E_2}$, and we know from above that $\card{C_1}+\card{D_1}+\card{E_1}=5k$, the number of positions required is at least
 
 \begin{tabular}{ll}
      & $\card{C_1}+2\card{C_2}+\card{D_1}+2\card{D_2}+\card{E_1}+e_2^*+3k$ \\
  $>$ & $\card{C_1}+2(2k-\card{C_1})+\card{D_1}+2(k-\card{D_1})+\card{E_1}+2\card{E_2}+3k$ \\
  $=$ & $\card{C_1}+2(2k-\card{C_1})+\card{D_1}+2(k-\card{D_1})+\card{E_1}+2(3k-\card{E_1})+3k$ \\  
  $=$ & $4k-\card{C_1}+2k-\card{D_1}+6k-\card{E_1}+3k$ \\
  $=$ & $15k-(\card{C_1}+\card{D_1}+\card{E_1})$ \\
  $=$ & $15k-5k$ \\
  $=$ & $10k$.
 \end{tabular}
 
 Since there are only $10k$ positions available, we have a contradiction. Therefore, it follows that $\card{E_2}=\frac12e_2^*$ as claimed.
 
 We now consider all votes in $V$ having a candidate $c\in X\cup Y$ in the first place. Since such a $c$ only needs to lose $\alpha_1$ points against $p$, and there is a vote available (introduced by the reduction above) that has a second relevant candidate (namely $S_i$ with $c\in S_i$) in the second position, and these are the only votes introduced whose deletion removes points from both $c$ and another relevant candidate, we can without loss of generality assume that each vote in $V$ having $c\in X\cup Y$ in the first position has a vote $S_i$ in the second position. In particular, this implies that $\card{C_1}\leq e_2^*$, therefore $$\frac12\card{C_1}\leq\frac12e_2^*=\card{E_2}.$$ 

 From $\card{C_1}+\card{D_1}+\card{E_1}=5k$ and $\card{D_1}\leq\card Z=k$, we obtain $\card{C_1}+\card{E_1}\ge 4k$, with $\card{E_1}=3k-\card{E_2}$, this implies $\card{C_1}-\card{E_2}\ge k$. Since we know $\card{C_1}-\card{E_2}\leq\frac12\card{C_1}$ from the above, it follows that $k\leq\card{C_1}-\card{E_2}\leq\frac12\card{C_1}$, and therefore $\card{C_1}\ge 2k$. Since $\card{C_1}\leq\card{X\cup Y}=2k$, this implies $$\card{C_1}=2k\mathtext{ and }\card{E_2}\ge\frac12\card{C_1}=k.$$
 
 From $\card{C_1}+\card{D_1}+\card{E_1}=5k$, we thus get $\card{D_1}+\card{E_1}=3k$, thus $\card{D_1}+(3k-\card{E_2})=3k$, i.e., $\card{D_1}=\card{E_2}$. Since $\card{D_1}\leq\card{Z}=k$, we get $k\leq\card{E_2}=\card{D_1}\leq k$, i.e.,
 
 $$\card{D_1}=\card{E_2}=k.$$
 
 Therefore, we know that $\card{C_1}=2k$, $\card{C_2}=2k-\card{C_1}=0$, $\card{D_1}=k$, $\card{D_2}=k-\card{D_1}=0$, $\card{E_1}=3k-\card{E_2}=2k$, and $\card{E_2}=k$. Therefore, these candidates together use
 
 $\card{C_1}+2\card{C_2}+\card{D_1}+2\card{D_2}+\card{E_1}+2\card{E_2}=2k+2\cdot 0+k+2\cdot 0+2k+2k=7k$ relevant positions in $V$. Since the candidates of the form $S_i'$ each use at least one relevant position of a vote in $V$, this means that all $10k$ relevant positions in $V$ are used by relevant candidates. In particular, the controller cannot remove any vote that has a dummy candidate in a relevant position, and thus does not remove any of the setup votes. This concludes the proof.
\end{proof}

\subsection{Proofs of Results from Section~\ref{sect:ccdv hardness by direct 3DM reduction}}

In this section we prove Theorem~\ref{theorem:collection of few coefficients CCDV hardness cases}, the proof is split up into the three distinct types of generators covered by the theorem.

\begin{restatable}{restatableTheorem}{theoremccdvhardnessalphaonealphatwozerodotszerominusalphafourminusalphafiveminusalphasix}
 \label{theorem:ccdv hardness alpha1 alpha2 0 dots 0 minus alpha4 minus alpha5 minus alpha6}
 Let $f=(\alpha_1,\alpha_2,\alpha_3,\dots,\alpha_3,\alpha_4,\alpha_5,\alpha_6)$ with $\alpha_1>\alpha_3>\alpha_5$. Then $f$-CCDV is NP-complete.
\end{restatable}

\begin{proof}
 For the proof, we normalize to $\alpha_3=0$ and write $f$ as $f=(\alpha_1,\alpha_2,0,\dots,0,-\alpha_4,-\alpha_5,-\alpha_6)$ with  $\alpha_1>0>-\alpha_5\ge-\alpha_6$ and $\alpha_2\ge0>-\alpha_5$.

 We again reduce from 3DM. Let $M\subseteq X\times Y\times Z$ be a be 3DM-instance. We construct an instance of $f$-CCDV as follows:
 
 \begin{itemize}
  \item The candidate set is $X\cup Y\cup Z\cup\set p\cup D$, where $p\notin X\cup Y\cup Z$ is the preferred candidate, and $D$ is a set of dummy candidates,
  \item For each $(x,y,z)\in M$, we add a vote $$(z>d_1>\RESTcandidates>d_2>y>x),$$ where $d_1,d_2\in D$, and $\RESTcandidates$ contains the remaining candidates in an arbitrary order (all of these candidates obtain $0$ points from this vote). We call these votes \emph{3DM-votes}.
  \item We introduce additional setup-votes (see below for details) such that the relative points of the candidates gained from the 3DM-votes and the setup votes are as follows:
  \begin{itemize}
   \item $\score p=0$
   \item $0<\score z\leq\alpha_1$ for each $z\in Z$ (note that $\alpha_1>0$)
   \item $-2\alpha_5<\score y\leq-\alpha_5$ for each $y\in Y$ (note that $\alpha_5>0$),
   \item $-2\alpha_6<\score x\leq-\alpha_6$ for each $x\in X$ (note that $\alpha_6\ge\alpha_5>0$),
   \item no dummy candidate $d\in D$ can win the election when at most $\card X$ votes are deleted.
  \end{itemize}
 \end{itemize}
 
 We first show that $M$ is positive if and only if $p$ can be made a winner of the election with deleting at most $\card X$ of the 3DM-votes. In the following, we identify elements of $M$ and the corresponding 3DM-votes.
 
 First assume that $M$ is a positive instance, i.e., there is some $C\subseteq M$ with $\card C=\card X$ such that each $c\in X\cup Y\cup Z$ appears in exactly one tuple of $C$. We show that $p$ wins the election when exactly the votes in $C$ are removed. Since, by construction, the dummy candidates cannot win the election, it suffices to show that no candidate $c\in X\cup Y\cup Z$ beats the preferred candidate $p$. Hence let $c$ be such a candidate. Since $C$ is a cover, exactly one 3DM-vote in which $c$ gets a non-zero amount of points is removed. Depending on whether $c\in Z$, $c\in Y$, or $c\in X$, $c$ gains $\alpha_1$, $-\alpha_5$, or $-\alpha_6$ points from this vote. We make a case distinction.
 
 \begin{itemize}
  \item If $c\in Z$, then $c$ loses $\alpha_1$ points. Since $\score c\leq\alpha_1$ initially, the final score of $c$ is at most $0$, and hence $c$ does not beat $p$.
  \item If $c\in Y$, then $c$ loses $-\alpha_5$ points, i.e., gains $\alpha_5$ points. Since $\score c\leq-\alpha_5$ initially, $c$ also does not beat $p$.
  \item If $c\in X$, then analogously, $c$ gains $\alpha_6$ points; since $\score c\leq\alpha_6$, $c$ does not beat $p$.
 \end{itemize}
 
 Therefore, $p$ wins the election as required. 
 
 For the converse, assume that there is a set $C$ of at most $\card X$ 3DM-votes such that removing $C$ makes $p$ win the election. Since initially, each candidate $z\in Z$ beats $p$, each such $z$ must lose points relatively to $p$. The only way for $z$ to lose points relative to $p$ by removing 3DM-votes is to remove a 3DM-vote that has $z$ in the first position. Therefore, $C$ contains, for each $z\in Z$, a tuple covering $z$. To show that $C$ also covers each $x\in X$ and each $y\in Y$, it suffices, due to cardinality reasons, to show that no such candidate appears in two tuples (or votes) from $C$.
 
 First assume that some $y\in Y$ appears in two votes from $C$. Then $y$ gains $2\alpha_5$ points. Since the score of $y$ is initially more than $-2\alpha_5$, this implies that the final score of $y$ is more than $0$, hence $y$ beats $p$, a contradiction. Analogously, if $x\in X$ appears in two votes from $C$, then $x$ gains $2\alpha_6$ points, beating $p$, we again have a contradiction.

 This proves that $M$ is a positive instance if and only if $p$ can be made a winner by deleting at most $\card X$ of the 3DM-votes. It remains to show how we can add votes that set up the relative points of the candidates as required above, and which will not be deleted by the controller. We first compute the points that the candidates obtain from the 3DM-votes; we denote these points with $\scoresub{3DM}c$ for a candidate $c$.
 
 \begin{itemize}
  \item $p$ gets $0$ points from each 3DM-vote, hence $\scoresub{3DM}p=0$,
  \item each $c\in X\cup Y\cup Z$ appears in exactly $3$ tuples from $C$, and gains $0$ points in all other 3DM-votes. Therefore:
  \begin{itemize}
    \item For each $z\in Z$, $\scoresub{3DM}z=3\alpha_1$,
    \item for each $y\in Y$, $\scoresub{3DM}y=-3\alpha_5$,
    \item for each $x\in X$, $\scoresub{3DM}x=-3\alpha_6$.
  \end{itemize}
 \end{itemize}
 
 Therefore, relative to $p$, each $z\in Z$ must lose $2\alpha_1$ points. We achieve this by adding two votes of the form 
 
 $$p>d_1>\RESTcandidates>d_2>d_3>d_4,$$
 
 where $d_1,d_2,d_3,d_4$ are dummy candidates from $D$, and $\RESTcandidates$ contains all remaining candidates. This lets $p$ gain $2\alpha_1$ points relatively to all candidates in $X\cup Y\cup Z$, and hence each $z\in Z$ beats $p$ by exactly $\alpha_1$ points, as required. However, these votes let $p$ also gain points against each candidate in $X\cup Y$. To get the required relative scores for these candidates, we proceed as follows:
 
 \begin{itemize}
  \item After adding the votes above, $p$ beats each $y\in Y$ by $3\alpha_5+2\alpha_1$ points. To ensure that $-2\alpha_5<\score y\leq-\alpha_5$ (relative to $p$), we add the following votes:
  \begin{itemize}
    \item votes of the form $y>d_1>\RESTcandidates>d_2>d_3>d_4$ (with $\RESTcandidates$ and dummies as previously) let $y$ gain $\alpha_1$ points relative to $p$. 
    \item votes of the form $d_1>d_2>\RESTcandidates>d_3>y>d_4$ let $y$ lose $\alpha_5$ points relative to $p$.
  \end{itemize}
  Combining these steps allows to adjust the relative score of $y$ to lie in the required interval of length $\alpha_5$.
  \item For $x\in X$, we proceed analogously by adding votes letting $x$ gain $\alpha_1$ points and votes removing $\alpha_6$ points, which allow the relative score of $x$ to lie in the required interval of length $\alpha_6$.
 \end{itemize}
 
 To ensure that the dummy candidates cannot win the election, we use the following construction: For each occurrence of a dummy candidate in one of the above votes, we use a fresh dummy candidate (who gets $0$ points in all other votes). Then each dummy candidate has at most $\alpha_1$ (absolute) points; the preferred candidate $p$ has $2\alpha_1$ (absolute) points. Therefore, no dummy candidate can beat $p$.
 
 It remains to show that if the CCDV instance is positive, i.e., if $p$ can be made a winner of the election with at most $\card X$ deletions, then $p$ can in fact be made a winner with deleting at most $\card X$ 3DM-votes. To see this, note that each candidate $z\in Z$ must lose points relative to $p$. This is only possible by removing votes in which $z$ is votes ahead of $p$. By construction, this is the case only for the 3DM-votes. Since in each 3DM-vote, only a single candidate from $Z$ is voted ahead of $p$, it is necessary to remove at least $\card X$ 3DM-votes in order to make $p$ win the election. Therefore, if $p$ is made a winner of the election by removing at most $\card X$ votes, then in fact exactly $\card X$ votes are removed, and each of them is a 3DM-vote. This concludes the proof.
\end{proof}

\begin{restatable}{restatableTheorem}{theoremccdvalphaoneaphatwobiggerthanalphathreebiggerthanalphasix}
 \label{theorem:ccdv alpha1alpha2>alpha3>alpha6}
 Let $f=(\alpha_1,\alpha_2,\alpha_3,\dots,\alpha_3,\alpha_6)$ with $\alpha_1,\alpha_2>\alpha_3>\alpha_6$. Then $f$-CCDV is NP-complete.
\end{restatable}

\begin{proof}
 We can write $f$ more simply as $f=(\alpha,\beta,0,\dots,0,-\gamma)$ with $0\notin\set{\alpha,\beta,\gamma}$. We proceed very similarly to the proof of Theorem~\ref{theorem:ccdv hardness alpha1 alpha2 0 dots 0 minus alpha4 minus alpha5 minus alpha6}; we again reduce from 3DM, with an instance given as $M\subseteq X\times Y\times Z$ with $\card{M}=3k$, $\card X=\card Y=\card Z=k$ be given, where each $c\in X\cup Y\cup Z$ appears in exactly $3$ tuples of $M$, and the sets $X$, $Y$, and $Z$ are pairwise disjoint. We construct an instance of $f$-CCDV as follows:
 
 For each $(x,y,z)\in M$, we add a vote $x>y>\RESTcandidates>z$, where $\RESTcandidates$ includes all remaining candidates. We use setup votes to ensure that the relative scores are as follows:
 
 \begin{itemize}
  \item $\score p=0$,
  \item $\score x=\alpha$ for each $x\in X$,
  \item $\score y=\beta$ for each $y\in Y$,
  \item $\score z=-\gamma$ for each $z\in Z$.
 \end{itemize}
 
 We show that the 3DM instance is positive if and only if $p$ can be made a winner of the election with deleting at most $k$ of the 3DM-votes. We again identify the elements of $M$ and the votes obtained from them.
 
 First assume that $C\subseteq M$ is a cover with size $k$. Then removing all votes in $C$ lets each candidate in $X$ lost $\alpha$ points, each candidate in $Y$ loses $\beta$ points, and each candidate in $Z$ gains $\gamma$ points, hence all candidates tie and $p$ wins the election.
 
 Therefore, it suffices to show that we can in fact construct setup votes achieving the required relative points that will not be deleted by the controller. This is simpler than the corresponding proof in Theorem~\ref{theorem:ccdv hardness alpha1 alpha2 0 dots 0 minus alpha4 minus alpha5 minus alpha6}: Since the controller must ensure that each $x$ and each $y$ loses points relative to $p$, it suffices to achieve the required points with setup votes that each vote at most one candidate from $X\cup Y$ ahead of $p$. Since $2k$ of these candidates (that is all of them) need to lose points against $p$, and removing a 3DM-vote results in $2$ of them losing points against $p$, and removing each setup votes allow at most one candidate from $X\cup Y$ to lose points against $p$, the controller can only remove votes that vote two candidates from $X\cup Y$ ahead of $p$, i.e., the setup votes.
 
 Now, it is easy to adjust the scores of every relevant candidate $c\in X\cup Y\cup Z$ with votes voting $c$ in the first, second, or last spot, all other relevant candidates in the $0$-point sequence, and dummy candidates in the two remaining positions.
 
 To ensure that the dummy candidates cannot win the election, we use essentially the same idea as in the proof of Theorem~\ref{theorem:ccdv hardness alpha1 alpha2 0 dots 0 minus alpha4 minus alpha5 minus alpha6}: We use a fresh dummy candidate for each position where one is needed, and then increase the points of each relevant candidate $c\in X\cup Y\cup Z\cup\set p$ with a vote voting $c$ first, dummy candidates in the second and last positions, and all remaining candidates in the $0$-points sequence. These votes do not change the relative points of the relevant candidates, and ensure that the dummies cannot win the election. Due to the same reasoning as above, these votes cannot be removed by the controller, hence this concludes the proof.
\end{proof}

\begin{restatable}{restatableTheorem}{theoremccdvhardnessalphaonebiggeralphathreebiggeralphafour}
 \label{theorem:ccdv hardness alpha1 > alpha 3 > alpha 4}
 Let $f=(\alpha_1,\alpha_2,\alpha_3,\dots,\alpha_3,\alpha_4,\alpha_5,\alpha_6)$ with $\alpha_1>\alpha_3>\alpha_4$. Then $f$-CCDV is NP-complete.
\end{restatable}

\begin{proof}
 We reduce from 3DM. Hence let $M\subseteq X\times Y\times Z$ be a 3DM-instance. Due to cardinality reasons, $M$ is positive if and only if there are $\card X$ many elements from $M$ which pairwise differ in all components. From $M$, we construct an $f$-CCDV instance as follows: The set of candidates is $X\cup Y\cup Z\cup\set p\cup\set B\cup D$, where $B$ is a new candidate, and $D$ is a set of dummy candidates. For each element $(x,y,z)$ in $M$, we introduce a vote
 
 $$B > d > \RESTcandidates > x > y > z,$$
 
 with $d\in D$, and where $\RESTcandidates$ contains the remaining candidates in an arbitrary order (all of these candidates receive $\alpha_3$ points from this vote). Again, we call these votes \emph{3DM-votes}.
 For a candidate $c\in X\cup Y\cup Z$, we define $r(c)=4,5,6$ if $c\in X,Y,Z$, respectively. We now set up the relative points of the candidates as follows:
 
 \begin{itemize}
  \item $\scorefinal p=0$,
  \item $\scorefinal B=\card X(\alpha_1-\alpha_3)$,
  \item $-2(\alpha_3-\alpha_{r(c)})<\scorefinal c\leq-(\alpha_3-\alpha_{r(c)})$ (note that $r(c)\ge4$ and hence $\alpha_3>\alpha_4\ge\alpha_{r(c)}$).
 \end{itemize}
 
 (recall that these points are relative to the score of $p$, and hence can be negative). We show below how votes can be added to the election instance such that 
 
 \begin{itemize}
  \item the resulting scores relative to $p$ are as indicated above,
  \item to make $p$ a winner with deleting at most $\card X$ votes, the controller can only delete votes introduced from the elements in the 3DM instance (we will call these 3DM-votes in the sequel).
 \end{itemize}
 
 We claim that the 3DM-instance is positive if and only if $p$ can be made a winner of the election with deleting at most $\card X$ of the 3DM-votes.
 
 First assume that the instance is positive, and let $C\subseteq M$ be a cover with $\card C=\card X$. Then deleting exactly the votes that correspond to the cover changes the scores of the candidates (relative to $p$) as follows:
 
 \begin{itemize}
  \item in each of these votes, $B$ gets $\alpha_1$ points and $p$ gets $\alpha_3$ points. Hence, relative to $p$, $B$ loses $\card X(\alpha_1-\alpha_3)$ points, and hence ties with $p$.
  \item for each candidate $c\in X\cup Y\cup Z$, there are $\card X-1$ votes in $C$ in which $c$ gets $\alpha_3$ points, and there is a single vote in which $c$ gets $\alpha_{r(c)}$ points, while $p$ gets $\alpha_3$ points in all of these votes. Hence, relative to $p$, each $c$ gains $\alpha_3-\alpha_{r(c)}$ points. Since $c$'s initial score is at least $\alpha_3-\alpha_{r(c)}$ below $p$, this means that $c$ does not beat $p$.
 \end{itemize}
 
 Now assume that $p$ can be made a winner of the election by removing at most $\card X$ of the 3DM-voters. Since $B$ must lose $\card X(\alpha_1-\alpha_3)$ points, and removing each vote lets $B$ lose $\alpha_1-\alpha_3$ points against $p$ (and $\alpha_1-\alpha_3>0$), it follows that exactly $\card X$ voters must be removed. Let $C\subseteq M$ be the set corresponding to the removed voters. We show that $C$ is a cover. Assume indirectly that this is not the case, then, since $\card C=\card X=\card Y=\card Z$, there is some element $c\in X\cup Y\cup Z$ appearing in at least two of the tuples in $C$. Then with deleting the votes corresponding to $C$, $c$ gains at least $2(\alpha_3-\alpha_{r(c)})$ points against $p$. Since $c$ initially has more than $-2(\alpha_3-\alpha_{r(c)})$ points, this means that $c$ beats $p$ in the final election, a contradiction. Therefore, $C$ is indeed a cover as claimed.
 
 It remains to show how to add votes to the above-introduced 3DM-votes such that the required scores are achieved, and such that any successful control action removing at most $\card X$ many votes will only delete 3DM-votes. We will achieve the latter property by only adding votes that give the candidate $p$ at least as many points as the candidate $B$. Since in order to at least tie with $B$, the preferred candidate $p$ must gain at least $\card X(\alpha_1-\alpha_3)$ points against $p$, and gains exactly $\alpha_1-\alpha_3$ points when a vote as above is deleted, this ensures that in order to ensure that $B$ does not beat $p$, only 3DM-votes can be deleted. 
 
 Hence it remains to show how to implement the above scores (relative to $p$) be adding to the above 3DM-votes only votes in which $p$ gets at least as many points as $B$. For this, we first compute the points (relative to $p$) that each candidate gets from the above 3DM-votes. 
 
 \begin{itemize}
  \item In each of the $3\card X$ votes from the 3DM-instance, $B$ gains $\alpha_1-\alpha_3$ points relative to $p$. Hence $B$'s initial relative score is $$\scorethreedm{B}=3\card X(\alpha_1-\alpha_3).$$
  \item Since each $c\in X\cup Y\cup Z$ appears in exactly $3$ tuples from $M$, each such candidate $c$ loses $\alpha_3-\alpha_{r(c)}$ points relative to $p$ in $3$ of the 3DM-votes, and ties with $p$ in the remaining ones. Therefore, $c$'s initial relative score is $$\scorethreedm{c}=-3(\alpha_3-\alpha_{r(c)}).$$
 \end{itemize}
 
 Therefore, we need to adjust the scores as follows:
 
 \begin{description}
  \item[adjusting the relative score between $B$ and $p$.] Relatively to $p$, the candidate $B$ must lose 
  
  $$\scorethreedm B-\scorefinal B=3\card X(\alpha_1-\alpha_3)-\card X(\alpha_1-\alpha_3)=2\card X(\alpha_1-\alpha_3)$$ 
  
  points. We achieve this by adding $2\card X$ votes of the form
  
  $$p>d>\RESTcandidates>d>d>d,$$
  
  where $d$ stands for (different) dummy candidates from $D$, and $\RESTcandidates$ contains other candidates in an arbitrary order. Each of these votes lets $p$ gain $\alpha_1-\alpha_3$ points against $B$; hence the $2\card X$ votes have the required effect. Note that these votes also add $2\card X(\alpha_1-\alpha_3)$ points to $p$ relative to each candidate $c\in X\cup Y\cup Z$.
  \item[adjusting the relative score between $c$ and $p$, for $c\in X\cup Y\cup Z$.] Recall that each of these $c$ lost $2\card X(\alpha_1-\alpha_3)$ points relative to $p$ by the above adjustment. We adjust the points of $c$ in two steps:
  
  \begin{enumerate}
   \item We add votes to ensure that, relative to $p$, $c$ has at least $-2(\alpha_3-\alpha_{r(c)})+1$ points. Since $c$ lost $2\card X(\alpha_1-\alpha_3)$ against $p$ in the above step, this means $c$ has to gain at least
   
   $$\begin{array}{lcl}
   & & -2(\alpha_3-\alpha_{r(c)})+1+2\card X(\alpha_1-\alpha_3)-\scorethreedm{c} \\ & = &
   -2(\alpha_3-\alpha_{r(c)})+1+2\card X(\alpha_1-\alpha_3)-(-3(\alpha_3-\alpha_{r(c)})) \\ & = &
   -2\alpha_3+2\alpha_{r(c)}+1+2\card X(\alpha_1-\alpha_3)+3\alpha_3-3\alpha_{r(c)} \\ & = & 
   2\card X(\alpha_1-\alpha_3)+\alpha_3-\alpha_{r(c)}+1
   \end{array}$$
   
   points, we denote this number with $\delta_c$. To let $c$ gain at least $\delta_c$ points, we add $\lceil\frac{\delta_c}{\alpha_1-\alpha_3}\rceil$ (recall that $\alpha_1-\alpha_3>0$) votes of the form $$c>d>\RESTcandidates>d>d>d,$$ where again the $d$ stand for (different) dummy candidates, and $\RESTcandidates$ contains all remaining candidates.
   \item after the above step, the score of $c$ may be larger than allowed---recall that we need the score of $c$ to lie in the interval of length $\alpha_3-\alpha_{r(c)}$ between $-2(\alpha_3-\alpha_{r(c)})$ and $-(\alpha_3-\alpha_{r(c)})$. To move the score into this interval, we repeatedly add votes that have all relevant candidates except $c$ in a position gaining $\alpha_3$ points, $c$ in the position gaining $\alpha_{r(c)}$ points, and four dummy candidates in the remaining positions (such a vote lets $c$ gain $\alpha_3-\alpha_{r(c)}$ points against $p$). This ensures that $c$'s score is in the required interval.
  \end{enumerate}
  
  By construction, $p$ gets at least as many points as $B$ in the above votes, and the required points are achieved. It remains to show that the dummy candidates are indeed irrelevant, i.e., never win the election when at most $\card X$ votes are removed. To achieve this, we set up the above votes as follows: We use a fresh dummy candidate for each position where a dummy candidate appears in the above votes, and position the dummy candidate in the block receiving $\alpha_3$ points in the remaining votes. Since $p$ gets at least $\alpha_3$ points in every vote, this means that for each dummy candidate $d$, there is at most one vote in which she can gain points against $p$, and in this vote, she gains $\alpha_1-\alpha_3$ points. However, in the $2\card X$ votes introduced above to adjust the relative score of $p$ and $B$, the candidate $p$ gains $2\card X(\alpha_1-\alpha_3)$ points against the dummy candidates. Hence, when at most $\card X$ votes are removed, $p$ still has a headstart of at last $(\card X-1)(\alpha_1-\alpha_3)$ against each dummy candidate, and therefore strictly beats each dummy candidate. This concludes the proof.
 \end{description}
\end{proof}

Now, Theorem~\ref{theorem:collection of few coefficients CCDV hardness cases} easily follows from the above results:

\theoremcollectionoffewcoefficientsCCDVhardnesscases*

\begin{proof}
 This follows from the above results, since point~\ref{few coefficients CCDV:ccdv hardness alpha1 alpha2 0 dots 0 minus alpha4 minus alpha5 minus alpha6} is exactly the generator covered in Theorem~\ref{theorem:ccdv hardness alpha1 alpha2 0 dots 0 minus alpha4 minus alpha5 minus alpha6}, point~\ref{few coefficients CCDV:ccdv alpha1alpha2>alpha3>alpha6} is Theorem~\ref{theorem:ccdv alpha1alpha2>alpha3>alpha6}, and point~\ref{few coefficients CCDV:ccdv hardness alpha1 > alpha 3 > alpha 4} is Theorem~\ref{theorem:ccdv hardness alpha1 > alpha 3 > alpha 4}
\end{proof}

\section{Proofs of Results for Bribery}\label{sect:proofs bribery}

\subsection{Proofs for Polynomial-Time Bribery Results}

In this section, we prove our polynomial-time bribery results. The following result for approval and veto-like scoring systems were obtained by Lin~\shortcite{lin:thesis:elections}:

\begin{restatable}{restatableTheorem}{theoremlinbriberyresultslist}\label{theorem:lin bribery results list}
 \begin{itemize}
  \item Bribery for $k$-approval is in \PTIME\ if $k\leq 2$ and NP-hard otherwise,
  \item Bribery for $k$-veto is in \PTIME\ if $k\leq 3$ and NP-hard otherwise.
 \end{itemize}
\end{restatable}

For the remaining polynomial-time bribery cases, note that, while corresponding (with duality) to the generators which give rise to a CCAV-problem in polynomial time, we cannot use arguments analogous to the results presented in Section~\ref{sect:relationship ccdv ccav}: Due to the additional manipulation step, the bribery problem is conceptually more complex than CCAV or CCDV. Therefore, algorithms for bribery tend to be more complicated than their CCAV or CCDV counterparts.

\subsubsection{Proof of Theorem~\ref{theorem:100-1 bribery in ptime}}

We now prove our first polynomial-time algorithm for bribery:

\theorembriberyonezeroesminusoneptime*

\begin{proof}
Let $C$ be the set of candidates, let $V$ be the set of voters,
let $p$ be the preferred candidate, and let $k$ be the number of voters we
can bribe. Without loss of generality, assume we bribe exactly $k$
voters and that all bribed voters rank $p$ first. 

There are three types of voters. Let $V_1$ be the set of voters that rank $p$
last, let $V_2$ be the set of voters that rank $p$ neither first nor last,
and let $V_3$ be the set of voters that rank $p$ first.  Note that bribing
a voter in $V_1$ to vote $p > \dots > d$ is at least as good as bribing
a voter in $V_2$ to vote $p > \dots > d$, which is at least as good as bribing
a voter in $V_3$ to vote $p > \dots > d$.
Thus we will assume that we
bribe as many $V_1$ voters as possible, followed by as many $V_2$ voters
as possible, followed by $V_3$ voters. Since we assume that all $k$
bribed voters put $p$ first, we also know $p$'s score after
bribery.

We consider the following three cases.

\begin{enumerate}
\item $k \leq \card{V_1}$.  We bribe $k$ voters from $V_1$.  
In this case, we view bribery as deleting $k$ voters followed by
a manipulation of $k$ voters.  Greedily delete $a > \dots > p$ for
highest scoring $a$, and update the scores.  Repeat until $k$
voters are deleted. Then add $p > \dots > a$ for highest scoring $a$.
Update the scores until $k$ voters have been added (or use the
manipulation algorithm for this case).

\item $k \geq \card{V_1} + \card{V_2}$.

In this case bribery will make $p$ will be a winner, since after bribery
$p$ will be ranked first by every voter.

\item $\card{V_1} < k < \card{V_1} + \card{V_2}$.

We bribe all voters in $V_1$ and $k - \card{V_1}$ of the $V_2$ voters.

We again view bribery as deletion followed by manipulation.
Delete all $V_1$ voters.  In $V_2$,
deleting a voter $a > \dots > b$ corresponds to transferring a point
from $b$ to $a$.   After deleting $k$ voters, the deleted voters will
be bribed to rank $p$ first and to rank some other candidate last.
After deleting $V_1$, $\score{c} = \scoresub{V_2\cup V_3}{c}$.
For every $V_2$ voter $a > \dots > b$ that is deleted, transfer
one point from $a$ to $b$.  For every bribe $p > \dots > d$, 
delete a point from $d$.  There are exactly $k$ bribes.  
After bribery, $\score{p} = \card{V_3} + k$ and the score
of every other candidate needs to be at most
$\score{p} = \card{V_3} + k$.  

All this immediately translates into the following min-cost network flow
problem.

\begin{enumerate}
\item We have a source $s$ and a sink $t$.
\item We have a node $c$ for every candidate $c \neq p$.
\item For every $c \in C - \{p\}$, there is an edge from
$s$ to $c$ with capacity $\scoresub{V_2\cup V_3}{c}$ and cost 0.
We will be looking for a flow that saturates all of these 
edges.  This ensures that candidates start with right score.
\item For every $c \in C - \{p\}$, there is an edge from
$c$ to $t$ with capacity $\card{V_3} + k$ and cost 0.
This ensures that after bribery, the score of every other candidate
will be at most the score of $p$.
\item For every $a,b \in C - \{p\}$ there is an edge from $a$ to $b$ with
cost 1
and capacity the number of voters that vote $a > \dots > b$.
These are the only edges with a cost.  The min cost will correspond
to the number of $V_2$ voters that we bribe.  So, we need to min
cost to be $k - \card{V_1}$.
\item We have now handled the CCDV part of the problem.  All that is
left to do is to handle the manipulation part.  So, we are adding $k$
vetoes to candidates other than $p$.  Add a node $v$ to handle the 
vetoes.  There is an edge from $v$ to $t$ of capacity $k$ and cost 0 so that
there will be at most $k$ vetoes.
And for every candidate $c \in C - \{p\}$, we
add an edge from $c$ to $v$ of capacity $k$ and cost 0.
\end{enumerate}
It is easy to see that there is a successful bribery of and only if there
is a network flow with value $\sum_{c \in C - \{p\}}\scoresub{V_2\cup V_3}{c}$
and min cost at most $k - \card{V_1}$.
\end{enumerate}
\end{proof}

\subsubsection{Proof of Theorem~\ref{theorem:bribery zeroes -beta -alpha in ptime}}

The following polynomial-time bribery proof uses the manipulation algorithm from the proof of Theorem~\ref{theorem:man}.

\theorembriberyzeroesbetaalphainptime*

\begin{proof}
Let $C$ be the set of candidates, let $V$ be the set of voters,
let $p$ be the preferred candidate and let $k$ be the number of voters we
can bribe.

As in the proof of Theorem~\ref{theorem:100-1 bribery in ptime}, we
partition $V$ into $V_1$, $V_2$, and $V_3$.  
$V_1$ consists of all voters in $V$ that rank $p$ last,
$V_2$ consists of all voters in $V$ that rank $p$ second-to-last,
and $V_3$ consists of the remaining voters.
In the proof of Theorem~\ref{theorem:100-1 bribery in ptime}, it was important
that bribing a $V_1$ voter is always at least as good as 
bribing a $V_2$ voter, which is always at least as good
as bribing a $V_3$ voter.  This is not always the case here, 
as we will see in Example~\ref{example:bribery}:

\begin{example}
\label{example:bribery}
This example shows that we sometimes need to bribe $V_3$ voters.
It also shows that it is sometimes better to bribe a $V_3$ voter
than it is to bribe a $V_2$ voter.
It also shows that an optimal bribery is not always an optimal
deletion followed by an optimal manipulation, since an optimal
deletion would never delete a voter from $V_3$.
We will use the scoring rule $(0, \ldots, 0, -1, -3)$.
Let $C = \{p, a, b, c, d, e, f\}$ and let $V$ consist of the following
voters: 
\begin{itemize}
\item one ($V_2$) voter voting $\cdots > p > a$,
\item one ($V_2$) voter voting $\cdots > p > b$, 
\item one ($V_2$) voter voting $\cdots > p > c$, 
\item two ($V_3$) voters voting $\cdots > e > f$,  and
\item two ($V_3$) voters voting $\cdots > f > e$.
\end{itemize}
The scores of the candidates are as follows.
\begin{itemize}
\item $\score{d} = 0$, 
\item $\score{p} = \score{a} = \score{b} = \score{c} = -3$, and
\item $\score{e} = \score{f} = -8$.
\end{itemize}
We can make $p$ a winner by bribing one of the $V_3$ voters
to vote $p > \cdots > d$.  But
it is easy to see that we can not make $d$ a winner by bribing a $V_2$
voter, wlog, the voter voting $\cdots > p > a$, since
in the bribed election, the score of $p$ will be at most
$-2$, and so both $a$ and $d$ must be in the last position of
the bribed voter.

\medskip
And the following example shows that it is sometimes better to bribe a
$V_2$ voter than it is to bribe a $V_1$ voter.

We will use the scoring rule $(0, \ldots, 0, -2, -3)$.
Let $C = \{p, a, b, c, d\}$ and let $V$ consist of the following
voters: 
\begin{itemize}
\item one ($V_1$) voter voting $\cdots > a > p$,
\item one ($V_1$) voter voting $\cdots > b > p$, 
\item one ($V_1$) voter voting $\cdots > c > p$, 
\item one ($V_2$) voter voting $\cdots > p > d$, 
\item one ($V_3$) voter voting $\cdots > b > d$, 
\item one ($V_3$) voter voting $\cdots > d > b$, 
\item one ($V_3$) voter voting $\cdots > c > d$, 
\item one ($V_3$) voter voting $\cdots > d > c$, 
\item two ($V_3$) voters voting $\cdots > a > d$, and 
\item one ($V_3$) voter voting $\cdots > d > a$.
\end{itemize}
The scores of the candidates are as follows.
\begin{itemize}
\item $\score{p} = -11$,
\item $\score{a} = -9$, 
\item $\score{b} = \score{c} = -7$, and
\item $\score{d} = -17$.
\end{itemize}
We can make $p$ a winner by bribing the $V_2$ voter to vote
$p > \cdots > b > c$.  But if we bribe one of the $V_1$ voters,
$p$'s score will be at most -10, and so the bribed voter needs to 
put $a$, $b$, and $c$ in the two last positions, which won't fit.
\end{example}

Though we may need to bribe $V_3$ voters, we will show that the number
of such voters is limited.

\begin{claim}
\label{claim:few-v3}
If there is a successful bribery, then there is a successful bribery
where we bribe at most $X$ voters from $V_3$.
\end{claim}

\begin{proof}
If $k \geq \card{V_1} + \card{V_2}$, we can make $p$ a winner by
bribing all voters from $V_1 \cup V_2$ to put $p$ first.

So, let $k < \card{V_1} + \card{V_2}$.
Consider a successful bribery that bribes a minimum number, $\ell$,
of $V_3$ voters and suppose for a contradiction that $\ell > X$.
If we bribe a voter from $V_3$, then for every $c \neq p$,
$\surpl{c}$ decreases by at most $\alpha$.  
If we bribe a voter $\cdots > a > p$ to vote  $p > \cdots > a$, then
for every $c \neq p$,
$\surpl{c}$ decreases by at least $\alpha$.
So, it is never better to bribe a voter from $V_3$ than it
is to bribe a voter from $V_1$.  It follows that all voters from
$V_1$ are bribed and so there are at least $\ell + 1$ unbribed $V_2$ voters.

Let $r = \lceil \alpha/\beta \rceil$.
If there exists a set $C' \subseteq C - \{p\}$ of $r$ candidates
and a set $V' \subseteq V_2$ of $r$ unbribed voters such that
for every $c \in C'$ there is a voter in $V'$ voting $\cdots > p > c$, then
deleting these $r$ unbribed $V_2$ voters
will not increase the surplus of any candidate
(since $p$'s score goes up by $r \beta \geq \alpha$), while deleting any
$r$ voters from $V_3$ will not decrease the surplus of any candidate.
It follows that bribing the $r$ unbribed $V_2$ voters instead
of $r$
$V_3$ voters will give a successful bribery, which contradicts
the assumption of minimality.

Let $s \geq r$ be the number of unbribed $V_2$ voters.
By the argument above, there is a candidate $a$ such that at least 
$s' =  \lceil \frac{s}{r-1} \rceil$ of these voters vote $\cdots p > a$.
Note that if we delete such a voter,
the surplus of all candidates other than $a$ does not increase.
In addition, note 
that if we delete a $V_3$ voter, the surplus
of none of the candidates decreases.
After deletion of one of these unbribed $V_2$ voters that 
vote $\cdots > p > a$, 
the score of $p$ will be $-(s - 1)\beta$
and the score of $a$ will be at most $-(s' - 1) \alpha$.
We need $-(s - 1)\beta \geq -(s' - 1) \alpha$.
It is easy to see that this is true if we choose
\[s \geq \frac{(\alpha/\beta - 1) (r - 1)}{\alpha/\beta - r + 1}.\]

Here is the derivation:
\begin{eqnarray*}
s & \geq &\frac{(\alpha/\beta - 1) (r - 1)}{\alpha/\beta - r + 1} \Rightarrow\\
(\alpha/\beta - r + 1)s & \geq & (\alpha/\beta - 1) (r - 1) \Rightarrow \\
\frac{(\alpha/\beta - r + 1)s} {r-1} & \geq & (\alpha/\beta - 1) \Rightarrow \\
\frac{(\alpha/\beta)s}{r-1} - \frac{(r-1)s}{r-1} & \geq & (\alpha/\beta - 1) \Rightarrow \\
(\frac{s}{r-1} - 1) \frac{\alpha}{\beta} & \geq & (s - 1) \Rightarrow \\
(s' - 1) \alpha & \geq & (s - 1)\beta.
\end{eqnarray*}

It follows that bribing one $V_2$ voter voting
$\cdots p > a$ instead of bribing any $V_3$ voter
is also a successful bribery,
which contradicts the minimality
of the number of $V_3$ voters that are bribed.
\end{proof}

We will now adapt the dynamic programming approach from 
the proof of Theorem~\ref{theorem:man} to show that bribery
is in polynomial time.

Consider an instance of the bribery problem.
Let $C$ be the set of candidates, let $V$ be the set of voters,
let $p$ be the preferred candidate and let $k$ be the number of voters we
can bribe.  Let $C - \{p\} = \{c_1, \ldots, c_m\}$.  

Note that by Claim~\ref{claim:few-v3}, 
there exists a successful bribery if and only if
there exists a set $V_3' \subseteq V_3$
and nonnegative integers $k_1$ and $k_2$ such that
$k_1 + k_2 + \card{V'_3} \leq k$ and $\card{V_3} \leq X$ and
there exists a successful bribery that bribes 
$k_1$ $V_1$ voters, $k_2$ $V_2$ voters, and all voters in $V_3'$.
We assume that every bribed voter puts $p$ first.
Without loss of generality, we assume that $k \leq \card{V}$ and
that $k = k_1 + k_2 + \card{V'_3}$.  

Let $\surpl c = \score{c}-\score{p}$.
Let $V_{1,i}$ be the set of $V_1$ voters that rank $c_i$ next to last
and let $V_{2,i}$ be the set of $V_2$ voters that rank $c_i$ last.

Note that if there exists a successful bribery 
that bribes $k_1$ $V_1$ voters,
$k_2$ $V_2$ voters, and all voters in $V'_3$, then
for all $i$, $1 \leq i \leq m$, there exist nonnegative integers
$x_i$ (the number of times $c_i$ is ranked next to last by a bribed voter),
$y_i$ (the number of times that $c_i$ is ranked last by a bribed voter),
$z_i$ (the number of bribed voters in $V_{1,i}$), and
$w_i$ (the number of bribed voters in $V_{2,i}$) such that:
\begin{enumerate}
\item $x_i + y_i \leq k$,
\item $z_i \leq \card{V_{1,i}}$,
\item $w_i \leq \card{V_{2,i}}$,
\item $\sum_{1 \leq i \leq m} x_i = k$,
\item $\sum_{1 \leq i \leq m} y_i = k$, 
\item $\sum_{1 \leq i \leq m} z_i = k_1$,
\item $\sum_{1 \leq i \leq m} w_i = k_2$, and
\item $\surpl{c_i} - \surplsub{V'_3}{c_i} - \beta x_i - \alpha y_i + \beta z_i + \alpha w_i \leq 0$.
\end{enumerate}

We define the following Boolean predicate $M$.
$M(k, k_\beta, k_\alpha, k'_1, k'_2, s_1, \ldots, s_\ell)$
is true if and only if for all $i$, $1 \leq i \leq \ell$,
there exist natural numbers $x_i, y_i, z_i$, and $w_i$ such that
\begin{enumerate}
\item $x_i + y_i \leq k$,
\item $z_i \leq \card{V_{1,i}}$,
\item $w_i \leq \card{V_{2,i}}$,
\item $\sum_{1 \leq i \leq \ell} x_i = k_\beta$,
\item $\sum_{1 \leq i \leq \ell} y_i = k_\alpha$, 
\item $\sum_{1 \leq i \leq \ell} z_i = k'_1$,
\item $\sum_{1 \leq i \leq \ell} w_i = k'_2$,
\item $s_i - \beta x_i - \alpha y_i + \beta z_i + \alpha w_i \leq 0$.
\end{enumerate}

Note that if there is a successful bribery that bribes
$k_1$ $V_1$ voters, $k_2$ $V_2$ voters, and all voters in $V'_3$, then
$M(k, k, k, k_1, k_2,
\surpl{c_1} - \surplsub{V'_3}{c_1}, \ldots , \surpl{c_m} - \surplsub{V'_3}{c_m})$ is true.
We will now show that the converse is true as well:
For $k = \card{V'_3} + k_1 + k_2$,
if $M(k, k, k, k_1, k_2,
\surpl{c_1} - \surplsub{V'_3}{c_1}, \ldots , \surpl{c_m} - \surplsub{V'_3}{c_m})$ is true then
there exists a successful bribery that bribes
$k_1$ $V_1$ voters, $k_2$ $V_2$ voters, and all voters in $V'_3$.

If $k = \card{V'_3}$ then $z_i = w_i = 0$ and the correctness of the claim
follows from the proof of Theorem~\ref{theorem:man}.

Now suppose that the claim holds for $k \geq \card{V'_3}$. 
We will show that it also
holds for $k + 1$.  So, let $k + 1 = k'_1 + k'_2 + \card{V'_3}$ and
for all $i$, $1 \leq i \leq m$,
let $x_i, y_i, z_i, w_i$ be natural numbers such that:
\begin{enumerate}
\item $x_i + y_i \leq k + 1$,
\item $z_i \leq \card{V_{1,i}}$,
\item $w_i \leq \card{V_{2,i}}$,
\item $\sum_{1 \leq i \leq m} x_i = k + 1$,
\item $\sum_{1 \leq i \leq m} y_i = k + 1$, 
\item $\sum_{1 \leq i \leq m} z_i = k'_1$,
\item $\sum_{1 \leq i \leq m} w_i = k'_2$, and
\item $\surpl{c_i} - \beta x_i - \alpha y_i + \beta z_i + \alpha w_i \leq 0$.
\end{enumerate}

If $k_1' > 0$, let
$r$ be such that $z_r > 0$ and we will bribe a $V_{1,r}$ voter.
Let $k_1 = k'_1 - 1 $ and let $k_2 = k'_2$.
Otherwise, $k'_2 > 0$ and we let
$r$ be such that $w_r > 0$ and we will bribe a $V_{2,r}$ voter.
Let $k_1 = k'_1$ and let $k_2 = k'_2 - 1$.
Let $X = \{i \ | \ x_i + y_i = k+1\}$.  Note that $\card{X} \leq 2$.
Let $i, j$ be such that $i \neq j$, $x_i > 0$, $y_j > 0$,
and $X \subseteq \{i,j\}$.  Let the bribed voter vote $\cdots > c_i > c_j$.
Subtract 1 from  $x_i$ and $y_j$ and $z_r$ (or $w_r$)
and recompute the surpluses. It follows from the induction hypothesis that
we can bribe
$k_1$ $V_1$ voters, $k_2$ $V_2$ voters, and all voters in $V'_3$
to make $p$ a winner.

To conclude the proof of Theorem~\ref{theorem:bribery zeroes -beta -alpha in ptime}, we will now 
show, by dynamic programming, that $M$ is computable
in polynomial time for unary $k, k_\beta, k_\alpha, k_1', k_2' \geq 0$.  
This is easy:
\begin{enumerate}
\item
$M(k,k_\beta, k_\alpha, k_1', k_2')$ is true if and only if
$k_\beta = k_\alpha = k_1' = k_2' =  0$.
\item For $\ell \geq 1$,
$M(k,k_\beta, k_\alpha, k_1', k_2', s_1, \ldots, s_\ell)$ if and only if
there exist natural numbers $x_\ell, y_\ell, z_\ell$, and $w_\ell$ such that:
\begin{enumerate}
\item
$x_\ell + y_\ell \leq k$,
\item
$x_\ell \leq k_\beta$,
\item
$y_\ell \leq k_\alpha$,
\item
$z_\ell \leq k_1'$,
\item
$w_\ell \leq k_2'$,
\item
$s_\ell - \beta x_\ell - \alpha y_\ell + \beta z_\ell + \alpha w_\ell
\leq 0$, and
\item
$M(k,k_\beta - x_\ell, k_\alpha - y_\ell, k_1' - z_\ell, k_2' - w_\ell,
s_1, \ldots, s_{\ell-1}).$
\end{enumerate}
\end{enumerate}
\end{proof}

\subsection{Proofs for Bribery Hardness Results}\label{section:bribery hardness proofs}

This section contains our hardness proofs for bribery, except for Theorem\ref{theorem:0 dots 0 -gamma -beta -alpha bribery}, which is proven in the main paper.

\subsubsection{Proof of Corollary~\ref{corollary:alpha 0 dots 0 beta bribery dichotomy}}

\corollaryalphazerodotszerobetabriberydichotomy*

\begin{proof}
 This follows in a similar way as Corollary~\ref{corollary:alpha 0 dots 0 beta ccdv dichotomy} for CCDV: The polynomial-time result for the generator $(2,1,\dots,1,0)$ follows from Theorem~\ref{theorem:100-1 bribery in ptime}, the hardness results follow from modifications of the proofs of Theorems~\ref{theorem:alpha 0 dots 0 alpha < beta ccdv} and~\ref{theorem:alpha 0 dots 0 beta < alpha ccdv}, which give hardness results for the corresponding CCDV-cases, as follows:
 
 We use essentially the same reduction, except that we construct an instance $I_b$ of $f$-bribery instead of the $f$-CCDV instance $I_c$ constructed in the above proofs. $I_b$ uses the same budget $k$ for the controller as the instance $I_c$. In the following, we will see bribery as deletion of voters followed by manipulation (keeping in mind that bribery is not necessarily \emph{optimal} CCDV followed by \emph{optimal} manipulation). In particular, we will refer to the bribed votes (i.e., the votes the bribed voters cast after the bribery) as \emph{manipulation votes}.
 
 The differences between the CCDV and the bribery setting and the resulting differences between $I_c$ and $I_b$ are as follows:
 
 \begin{itemize}
  \item Without loss of generality, all manipulation votes place the preferred candidate $p$ in the first position. Therefore, $p$'s final score is $k\alpha$ higher than in the CCDV setting, but still is a multiple of $\alpha$, and the value $N_p$ can be computed by the reduction.
  \item The manipulation votes vote $k$ (not necessarily distinct) candidates in the position giving $-\beta$ points. We say that the votes \emph{veto} these candidates. In order to handle this additional strategic freedom of the controller, we proceed as follows:
  \begin{itemize}
    \item We add additional candidates $b_1,\dots,b_k$, who each need to lose $\beta$ points against $p$ in order to not beat $p$.
    \item The scores of $b_1,\dots,b_k$ are set using the same ``setup-vote'' strategy as used in the proofs of Theorems~\ref{theorem:alpha 0 dots 0 alpha < beta ccdv} and~\ref{theorem:alpha 0 dots 0 beta < alpha ccdv} above.
    \item We show that, in order to make $p$ win the election, the controller must use at least one veto for each of the candidates $b_1,\dots,b_k$. To see this, recall that each of these candidates must lose $\beta$ points against $p$. If no veto is used for some $b_i$, then $b_i$ needs to lose these points via deleting a vote of the form $b_i>d_1$. As shown in the proofs of the CCDV results, this causes a ``chain reaction'' of additional removals which requires more removals than allowed by the budget $k$, we have a contradiction. Therefore, at least one veto must be used to ensure that $b_i$ does not beat $p$.
    \item Since at least one veto must be used for each $b_i$, we can without loss of generality assume that the controller uses one veto for each $b_i$ and no deletions, since this suffices to ensure that $b_i$ does not beat $p$.
    \item Therefore, we know that the manipulation votes vote $p>b_1$, $p>b_2$, \dots, $p>b_k$, and the remainder of each proof is identical to the respective CCDV case---with a value $N_p$ increased by $k$ in comparison with the CCDV case, as explained above.
  \end{itemize}
 \end{itemize}
 
 This completes the proof.
\end{proof}

\subsubsection{Proof of Theorem~\ref{theorem:0 minus alpha5 minus alpha1 dots minus alpha1 bribery hardness}}

\theoremzerominusalphafiveminusalphaonedotsminusalphaonebriberyhardness*

\begin{proof}
 In this case, the proof is an easy addition modification to the proof of Theorem~\ref{theorem:0 minus alpha5 minus alpha1 dots minus alpha1 CCDV hardness}. Without loss of generality, we assume that $\alpha_3=0$. Clearly, all manipulation votes will be of the form 
 
 $$p>x>\RESTcandidates$$
 
 for the preferred candidate $p$ and a dummy candidate $x$ (recall that the proof of Theorem~\ref{theorem:0 minus alpha5 minus alpha1 dots minus alpha1 CCDV hardness} introduces dummy candidates that can never win the election, we reuse these candidates here). Since the controller's budget is $5k$, $p$ will gain $5k\alpha_1$ points against every relevant candidate from the manipulation votes. It therefore suffices to let all other relevant candidates gain $5k\alpha_1$ additional points using the setup mechanism described in the proof of Theorem~\ref{theorem:0 minus alpha5 minus alpha1 dots minus alpha1 CCDV hardness}.
\end{proof}

\subsubsection{Proof of Theorem~\ref{theorem:bribery hardness alpha1 alpha2 0 dots 0 minus alpha4 minus alpha5 minus alpha6}}

\theorembriberyhardnessalphaonealphatwozerodotszerominusalphafourminusalphafiveminusalphasix*

\begin{proof}
 We follow the recipe from Section~\ref{sect:bribery hardness from ccdv hardness} to obtain the result from the hardness result for CCDV proved in Theorem~\ref{theorem:ccdv hardness alpha1 alpha2 0 dots 0 minus alpha4 minus alpha5 minus alpha6}. As in the proof of that theorem, we write $f$ as $f=(\alpha_1,\alpha_2,0,\dots,0,-\alpha_4,-\alpha_5,-\alpha_6)$ with $\alpha_1>0>-\alpha_5\ge-\alpha_6$ and $\alpha_2\ge0>-\alpha_5$. Applying the recipe requires the following:
 
 \begin{itemize}
  \item The reduction from 3DM uses the same transformation as above, i.e., for a triple $(x,y,z)$, we generate a 3DM vote $z>d_1>\RESTcandidates>y>x$.
  \item Clearly, all manipulators will vote $p>d>\RESTcandidates>r_1>r_2$, where $d$ is a dummy candidate and $r_1$ and $r_2$ are relevant candidates.
  \item Therefore, $p$ gains $\card X\alpha_1$ points (where $X$ is the set from the 3DM instance, recall that $\card X$ is the controller's budget in the reduction from Theorem~\ref{theorem:ccdv hardness alpha1 alpha2 0 dots 0 minus alpha4 minus alpha5 minus alpha6}, it will also be her budget in the current proof), and the relevant candidates lose $\card X(\alpha_5+\alpha_6)$ points.
  \item We therefore increase the points of each candidate in $X\cup Y\cup Z$ by $\card X\alpha_1$, compared to the reduction from Theorem~\ref{theorem:ccdv hardness alpha1 alpha2 0 dots 0 minus alpha4 minus alpha5 minus alpha6}.
  \item Additionally, we introduce two candidates $B_5$ and $B_6$, such that $\scorefinal{B_i}=\card X(\alpha_1+\alpha_i)$. The intention is that these candidates tie with $p$ if they are voted in the last two positions in all manipulator votes.
 \end{itemize}

 This is realized as follows:
 
 \begin{itemize}
  \item The precise required points of all candidates (as always, relative to $p$) are as follows:
  \begin{itemize}
    \item $\scorefinal{z}=(\card X+1)\alpha_1$, so each $z$ must lose exactly $\alpha_1$ points to tie with $p$ (recall that $p$ gains $\card X\alpha_1$ points from the manipulation votes),
    \item $\scorefinal{y}=\card X\alpha_1-\alpha_5$, so each $y$ may gain exactly $\alpha_6$ points from the deletion of one tuple in which $x$ is voted in the second to last position,
    \item $\scorefinal{x}=\card X\alpha_1-\alpha_6$, so each $x$ may gain exactly $\alpha_6$ points from the deletion of one tuple in which $x$ is voted last,
    \item $\scorefinal{B_i}=\card X(\alpha_1+\alpha_i)$ for $i\in\set{5,6}$.
  \end{itemize}
 \end{itemize}
 
  The proof of the CCDV result (Theorem~\ref{theorem:ccdv hardness alpha1 alpha2 0 dots 0 minus alpha4 minus alpha5 minus alpha6}) relies on the fact that each $z$ can only lose points relative to $p$ by removing 3DM votes. In order to keep this feature in our current bribery setting, we need to ensure that $B_5$ and $B_6$ are voted in the last two positions of every manipulation vote, so that no candidate $z$ can lose points by being voted in one of the last two positions in a manipulation vote. Therefore, there are two main issues to handle:
  
  \begin{enumerate}
   \item we need to ensure that each $z\in Z$ gain sufficiently many points from the 3DM votes, namely $\card X\alpha_1$ points more than in the CCDV reduction,
   \item we must ensure that $B_i$ gains enough points from votes that the controller cannot delete.
  \end{enumerate}
  
  For the first issue, we use the following idea: Using Proposition~\ref{prop:f3dm np complete}, we first transform the given 3DM-instance into an $F$-3DM instance, with a suitably chosen $F$ (see below). Recall that in $F$-3DM, the size of the desired cover is still $\card X$, which is the budget of the controller in Theorem~\ref{theorem:ccdv hardness alpha1 alpha2 0 dots 0 minus alpha4 minus alpha5 minus alpha6}, and which will also be her budget in the current proof. The only relevant difference is that each $c\in X\cup Y\cup Z$ now appears in exactly $3F$ tuples from $M$. This results in the following scores from the 3DM votes:
  
  \begin{itemize}
   \item $\scorethreedm{z}=3F\alpha_1$,
   \item $\scorethreedm{y}=-3F\alpha_5$,
   \item $\scorethreedm{z}=-3F\alpha_6$.
  \end{itemize}
  
  By choosing $F$ sufficiently large enough, and adding votes $p>d_1>\RESTcandidates>d_2>d_3$ for dummy candidates $d_1,d_2,d_3$, we can ensure that the score of each $z$ is exactly the desired number $(\card X+1)\alpha_1$. (Recall that using Proposition~\ref{prop:f3dm np complete}, we can transform a 3DM-instance into an $F$-3DM instance when $F$ is given in unary.) As in the proof of Theorem~\ref{theorem:ccdv hardness alpha1 alpha2 0 dots 0 minus alpha4 minus alpha5 minus alpha6}, we can adjust the points of candidates in $X$ and $Y$ using votes that have $z$ and $p$ both in the $0$-point segment of the votes, and which therefore will not be removed by the controller.
  
  For the second issue, we add points to $B_i$ with votes as follows:
  
  \begin{itemize}
   \item A single vote $B_i>d>\RESTcandidates>d>d_1$, where $d$ stands for arbitrary dummy candidates never used again,
   \item $d_1>d>\RESTcandidates>d>d_2$, with $d$ as above, of these votes we add enough to ensure that $d_1$ does not beat $p$, but cannot gain $\alpha_6$ points without beating $p$. (If adding these votes lets $d_1$ gain too many points, we use additional votes having dummy candidates in the $\alpha_1,\alpha_2,$ and $\alpha_5$ positions, and with $d_2$ in the $-\alpha_6$ position, clearly these votes will not be removed by the controller.)
  \end{itemize}
  
  We now argue that the only way for the controller to remove points from $B_i$ is by voting them in the last two positions of every manipulation vote, which then implies that the candidates from $Z$ must lose their points by removals of 3DM votes as in the proof of Theorem~\ref{theorem:ccdv hardness alpha1 alpha2 0 dots 0 minus alpha4 minus alpha5 minus alpha6}. For this, we denote the set $Z\cup\set{B_5,B_6}$ with $R$. These are the candidates that need to lose points against $p$. Note that none of the votes we introduce have any candidate from $R$ in the $\alpha_2$ position. Therefore, by bribing $\card X$ voters, the candidates in the set $R$ can only lose points (as always, relative to $p$) as follows:
  
  \begin{itemize}
   \item each of the $\card X$ many removed votes lets the set $R$ lose at most $\alpha_1$ points, hence from the removal, these candidates (combined) lose at most $\card X\alpha_1$ points.
   \item each of the $\card X$ many manipulation votes lets the set $R$ lose at most $\card R\alpha_1+\alpha_5+\alpha_6$ points, since $p$ gains $\alpha_1$ points against each of the $\card R$ candidates in $R$, and additionally the candidates in the last two positions lose $\alpha_5$ and $\alpha_6$ points, respectively. Therefore, the manipulation votes let $R$ lose at most $\card X(\card R\alpha_1+\alpha_5+\alpha_6)$ points against $p$.
  \end{itemize}
  
  Altogether, the bribery action therefore lets $R$ lose $\card X\alpha_1+\card X(\card R\alpha_1+\alpha_5+\alpha_6)=
  \card X((\card R+1)\alpha_1+\alpha_5+\alpha_6)$ points. Since $\card R=\card X+2$ (as $\card Z=\card X$), this value is identical to 
  $\card X((\card X+3)\alpha_1+\alpha_5+\alpha_6)$.
  
  Initially, each $z\in Z$ must lose $(\card X+1)\alpha_1$ points, and $B_i$ must lose $\card X(\alpha_1+\alpha_i)$ points. Since there are $\card X$ many candidates $z$, this means that the group $R$ must lose $\card X(\card X+1)\alpha_1+\card X(\alpha_1+\alpha_5)+\card X(\alpha_1+\alpha_6)=\card X((\card X+3)\alpha_1+\alpha_5+\alpha_6)$ points.
  
  Since the two values are equal, every possible loss of a point in the bribery action must be used for the set $R$. Now assume that $B_i$ loses a point by deleting a vote (as opposed to the intention, i.e., losing points only by manipulation). Then, a vote of the form $B_i>d>\RESTcandidates>d>d_1$ is deleted. However, this means that the candidate $d_1$ gains too many points, and hence $d_1$ must lose points using a delete or manipulation action. This is a contradiction, since we just showed that only candidates in $R$ may lose points (relative to $p$) using the bribery action. We therefore know that, in fact, only 3DM votes are removed. This concludes the proof.
\end{proof}

\subsubsection{Proof of Theorem~\ref{theorem:bribery alpha1alpha2>alpha3>alpha6}}

\theorembriberyalphaoneaphatwobiggerthanalphathreebiggerthanalphasix*

\begin{proof}
 This proof is obtained from the proof of Theorem~\ref{theorem:ccdv alpha1alpha2>alpha3>alpha6} similarly to the way the proof of Theorem~\ref{theorem:bribery hardness alpha1 alpha2 0 dots 0 minus alpha4 minus alpha5 minus alpha6} is obtained from that of Theorem~\ref{theorem:ccdv hardness alpha1 alpha2 0 dots 0 minus alpha4 minus alpha5 minus alpha6}. 
 
 We write the generator $f$ as $f=(\alpha,\beta,0,\dots,0,-\gamma)$ with $0\notin\set{\alpha,\beta,\gamma}$. As in the CCDV hardness proof, for each triple $(x,y,z)$ we add a vote $x>y>\RESTcandidates>z$.
 
 As in the proof of~\ref{theorem:bribery hardness alpha1 alpha2 0 dots 0 minus alpha4 minus alpha5 minus alpha6}, following the recipe from Section~\ref{sect:bribery hardness from ccdv hardness}, we again need to increase the scores of each $c\in X\cup Y\cup Z$ by $\card X\alpha$, and we introduce a candidate $B_\gamma$ who will be voted last in every manipulation vote.
 
 As in the proof of Theorem~\ref{theorem:ccdv alpha1alpha2>alpha3>alpha6}, setting up the scores is easier, since the 3DM-votes are more attractive to delete, as deleting them ``hurts'' two candidates instead of just one as in the proofs of Theorems~\ref{theorem:bribery hardness alpha1 alpha2 0 dots 0 minus alpha4 minus alpha5 minus alpha6} and~\ref{theorem:ccdv hardness alpha1 alpha2 0 dots 0 minus alpha4 minus alpha5 minus alpha6}:
 
 \begin{itemize}
  \item We can simply use setup votes that only have one relevant candidate in one of the non-zero positions, and fill the other two with dummy candidates,
  \item This allows us to increase the points of all candidates in $X\cup Y\cup Z$ by $\card X\alpha$ and add a candidate $B_\gamma$ with $\card X(\alpha+\gamma)$ points,
  \item then, by the same capacity argument as in Theorem~\ref{theorem:ccdv alpha1alpha2>alpha3>alpha6}, only votes with two relevant candidates (i.e., from $X\cup Y\cup\set B$) are deleted, hence the candidate $B_\gamma$ is voted last in every manipulation vote. 
 \end{itemize}
 
 Therefore, the reduction works in essentially the same way as in Theorem~\ref{theorem:ccdv alpha1alpha2>alpha3>alpha6}.
\end{proof}

\subsubsection{Proof of Theorem~\ref{theorem:bribery hardness alpha1 > alpha 3 > alpha 4}}

\theorembriberyhardnessalphaonebiggeralphathreebiggeralphafour*

\begin{proof}
 The proof is an easy application of the recipe from Section~\ref{sect:bribery hardness from ccdv hardness} to the proof of Theorem~\ref{theorem:bribery hardness alpha1 > alpha 3 > alpha 4}:
 
 \begin{itemize}
  \item Clearly, $p$ will be voted first in every manipulation vote. Since the construction will enforce that only 3DM votes will be deleted (each of which lets $p$ get $\alpha_3$ points), this means that $p$ will gain $\card X(\alpha_1-\alpha_3)$ points from a bribery action.
  \item As a consequence, we need to increase the points of each candidate by $\card X(\alpha_1-\alpha_3)$.
  \item Additionally, we introduce candidates $B_4$ and $B_5$, where for $i\in\set{4,5}$, $B_i$ must lose $\card X(\alpha_3-\alpha_i)$ points from the manipulation votes.
  \item We additionally increase the points of $B$ by $\card X(\alpha_3-\alpha_6)$. This implies that, in a successful bribery:
  \begin{enumerate}
   \item Only vote giving $\alpha_1$ points to $B$ first can be deleted, and 
   \item $B$ must be voted in a position awarding $\alpha_6$ points in every manipulation vote.
  \end{enumerate}
 \end{itemize}
 
 Therefore, as long as we only use setup votes that give less than $\alpha_1$ points to $B$, we know that these votes will not be deleted. Since $B_4$ and $B_5$ are in the $\alpha_3$-segment of the 3DM votes, $B_4$ and $B_5$ will not lose points from the delete action, and therefore, all positions giving fewer than $\alpha_3$ points in the manipulation votes will be filled with candidates $B$, $B_4$, and $B_5$.
 
 It therefore remains to show how we can increase the relative points of $B$ sufficiently. Lowering the relative points of $B$ to the exact required amount and adjusting the points of the remaining candidates can be done as usual (see, e.g., the proof of Theorem~\ref{theorem:bribery hardness alpha1 > alpha 3 > alpha 4}), note that we do not need votes that give $\alpha_1$ points to $B$ for this.
 
 To increase the points of $B$, we simply transform the given 3DM instance into an $F$-3DM instance for a suitably chosen $F$ with an application of Proposition~\ref{prop:f3dm np complete}. This increases the relative points of $B$ gained by the 3DM votes to $3F\card X(\alpha_1-\alpha_3)$. Since we need to increase $B$'s score by $\card X(\alpha_1-\alpha_3+\alpha_3-\alpha_6)=\card X(\alpha_1-\alpha_6)$, choosing $F=1+\ceilfrac{\alpha_1-\alpha_6}{\alpha_1-\alpha_3}$ suffices.
\end{proof}

\subsubsection{Proof of Theorem~\ref{theorem:approval generalization bribery}}

\thmapprovalgeneralizationbribery*

\begin{proof}
Clearly, it suffices to consider the case where $\alpha^m_4>\alpha^m_{2n}$, and the value $\card M$ satisfies the condition. We can also assume, without loss of generality, that $m$ is large enough such that $f(m)$ uses at least three different coefficients. We let $X\cup Y\cup Z=\set{s_1,\dots,s_{3k}}$. For each $c\in X\cup Y\cup Z$, let $i(c)$ denote the unique index $j$ with $c=s_{j}$. We choose a value $\ell$ such that $\ell\ge3$ and $(\alpha^m_4-\alpha^m_{m-3k})>k(\alpha^m_\ell-\alpha^m_{\ell+3k-1})$. We first prove that such an $\ell$ exists, a matching one then can be found in polynomial time since $f$ is polynomial-time uniform.

For this, we first choose $x$ as the number of sequential blocks of length $3k$ starting at position $4$, such that the last block ends before the position $m-3k$, i.e., $x=\lfloor\frac{m-4}{3k}\rfloor-1$. We now choose $\ell$ as the start of the block with the minimal difference between the coefficient at position $\ell$ and the coefficient at position $\ell+3k-1$, i.e., we choose $\ell\in\mathbb N$ such that $\ell=3+x'\cdot 3k$, and $(\alpha^m_\ell-\alpha^m_{\ell+3k-1})$ is minimal. Then, since between position $4$ and position $m-3k$, we have at least $x$ of these blocks, the difference between $\alpha^m_4$ and $\alpha^m_{m-3k}$ is at least $x$ times the difference inside the block starting at position $\ell$, i.e.,

$$\alpha^m_4-\alpha^m_{m-3k}\ge x(\alpha^m_\ell-\alpha^m_{l-3k-1}).$$

Note that the left-hand side of this inequality is in fact strictly positive, since we know that $\alpha^m_4>\alpha^m_{\frac23m}$, and $\frac23m\ge m-3k$, as $m=3\card M$, $k=\card X$, and $\card M\ge\card X^2+2\card X+2$. To obtain the required inequality $(\alpha^m_4-\alpha^m_{m-3k})>k(\alpha^m_\ell-\alpha^m_{\ell+3k-1})$, it therefore suffices to show that $x>k$. For this, note that by choice of $\card M$ and since $m=3\card M$, and $\card X=k$, we have that

$$
\begin{array}{lll}
 x & =   & \left\lfloor\frac{m-4}{3k}\right\rfloor-1 \\
   & \ge & \frac{m-4}{3k} -2 \\
   & =   & \frac{3\card M-4}{3\card X}-2 \\
   & \ge & \frac{3(\card X^2+2\card X+2)}{3\card X}-2 \\
   & =   & \frac{3\card X^2+6\card X+6}{3\card X}-2 \\
   & =   & \card X+2+\frac{2}{\card X}-2 \\
   & >   & \card X \\
   & =   & k,
\end{array}
$$

as required. Hence $\ell$ as chosen above satisfies the requirements.

Now, let a 3DM-instance $M\subseteq X\times Y\times Z$ be given. We denote $\card X$ with $k$. We construct an $f$-bribery instance as follows:

The candidate set is $X\cup Y\cup Z\cup p\cup B\cup D$, where $p$ is the preferred candidate, $B$ is a set of $n_b:=m-(\ell+3k-1)$ blocking candidates, and $D$ is a set of $\ell$ dummy candidates. Let $B=\set{b_1,\dots,b_{n_b}}$. Note that by construction, the number of candidates is exactly $m$ as required, and also $\card B\neq\emptyset$, $\card D\ge3$. Without loss of generality, we assume that $\alpha^m_m=0$.

For each $(x,y,z)\in M$, we introduce a vote (we will again call these votes \emph{3DM votes}):

$$b_1>x>y>z>\RESTcandidates>S_{xyz}>p,$$

where $\RESTcandidates$ contains the remaining blocking candidates and all dummy candidates, and $S_{xyz}$ contains each $c\in X\cup Y\cup Z\setminus\set{x,y,z}$ in the $i(c)$-th position, with dummy candidates taking the positions of $x$, $y$, and $z$. Let $r(c)$ denote the position in which $c$ is positioned if $c$ is among the first $4$. (I.e., $c$ gains $\alpha^m_{r(c)}$ points from the votes introduced for a tuple containing $c$, and $\alpha^m_{m-3k+i(c)-1}$ from the remaining 3DM votes.) For a candidate $b\in B$ with $b\neq b_1$, let $i(b)$ denote the position of $b$ in the 3DM votes (i.e., $b$ gets $\alpha^m_{i(b)}$ points from each of these votes).

We set up the (relative) scores of the candidates as follows:

\begin{itemize}
 \item $\score{p}=0$,
 \item $\score{c}=k\cdot\alpha^m_1+\alpha^m_{r(c)}+(k-1)\alpha^m_{m-3k+i(c)-1}-k\alpha^m_{\ell+i(c)-1}$
 \item $\score{b_1}=2k\cdot\alpha^m_1$,
 \item $\score{b_i}=k\alpha^m_1+k\cdot\alpha^m_{i(b_i)}-k\alpha^m_{m-i+1}$ for $i\ge2$,
 \item for all dummy candidates, the scores are so low that they cannot win the election with at most than $k$ bribes.
\end{itemize}

We now show that the 3DM instance is positive if and only if $p$ can be made a winner in the $f$-election with bribing at most $k$ voters, assuming that the above scores can be realized by additional setup votes that will never be deleted by the controller.

First assume that the 3DM instance is positive, i.e., there is a cover $C\subseteq M$ with $\card C=k$. Then the following bribery action is successful:

\begin{itemize}
 \item for each $(x,y,z)\in C$, we delete the vote introduced for $(x,y,z)$,
 \item we add $k$ votes of the form
 $$p>\RESTcandidates>s_1>\dots>s_{3k}>b_{n_b}>b_{n_b-1}>\dots>b_2>b_1.$$
\end{itemize}

Note that in the manipulation votes, a candidate $c\in X\cup Y\cup Z$ is voted in position $m-n_b-3k+i(c)=\ell+i(c)-1$.

We show that none of the non-dummy candidates beats $p$ after the bribery action:

\begin{itemize}
 \item $p$ is voted last in every deleted voter and voted first in every manipulation voter, $p$ gains $k\alpha^m_1$ points, therefore $\score p=k\alpha^m_1$ points.
 \item since $C$ is a cover, for each $c$, one vote featuring $c$ and $(k-1)$ votes not featuring $c$ are removed. Therefore, $c$ loses $\alpha^m_{r(c)}+(k-1)\alpha^m_{m-3k+i(c)-1}$ points from the delete action. Since $c$ is voted in position $\ell+i(c)-1$ in each of the manipulation votes, $c$ gains $k\alpha^m_{\ell+i(c)-1}$ points from the manipulation votes. The final score of $c$ is therefore $k\cdot\alpha^m_1+\alpha^m_{r(c)}+(k-1)\alpha^m_{m-3k+i(c)-1}-k\alpha^m_{\ell+i(c)-1}-(\alpha^m_{r(c)}+(k-1)\alpha^m_{m-3k+i(c)-1})+k\alpha^m_{\ell+i(c)-1}= k\cdot\alpha^m_1$, hence $c$ ties with $p$ 
 \item $b_1$ loses $\alpha^m_1$ points with each of the $k$ deletions, and receives $0$ points from the manipulation votes. Therefore, $b_1$'s final score is $k\alpha^m_1$, also tieing with $p$.
 \item for $b\in B\setminus\set{b_1}$, $b_j$ loses $\alpha^m_{i(b)}$ points from each deletion, and gains $\alpha^m_{m-i+1}$ points from each manipulation vote. Therefore, the final score of $b$ is $k\alpha^m_1+k\alpha^m_{i(b)}-k\alpha^m_{m-i+1}-k\alpha^m_{i(b)}+k\alpha^m_{m-i+1}=k\alpha^m_1$ points, hence $b$ also ties with $p$.
\end{itemize}

For the other direction, assume that $p$ can be made a winner by bribing at most $k$ of the 3DM votes. Clearly, every bribed voter will vote $p$ first after the bribery, therefore, $p$'s final score is $k\alpha^m_1$. Since $b_1$ may not beat $p$, it follows that $b_1$ must gain $0$ points from the manipulation votes. It also follows that there are exactly $k$ bribed voters, since $b_1$ must lose $k\alpha^m_1$ points in order not to beat $p$. Similarly, in order not to beat $p$, each $b_i$ for $i\ge2$ may only gain $k\alpha^m_{m-i+1}$ points from the manipulation votes. Therefore (allowing for swaps between the $b_i$ candidates that cancel each other out, and exchanging positions corresponding to coefficients with the same value), we can without loss of generality assume that the manipulation votes vote candidate $b_i$ in position $m-i+1$. Therefore, the last $n_b$ many positions in all manipulation votes are taken by candidates from $B$, and thus, in each of these votes, each candidate $c$ gains at least $\alpha^m_{m-n_b}$ points. Since $n_b=m-(\ell+3k-1)$, this means that each $c$ gains at least $\alpha^m_{\ell+3k-1}$ points in each manipulation votes. 

Let $C$ be the set of all tuples $(x,y,z)$ such that the vote corresponding to $(x,y,z)$ is deleted. We claim that $C$ is a cover. By the above, we know that $\card C=k$. Now assume that there is some $c\in X\cup Y\cup Z$ that is not covered by $C$. Then, $c$ is voted in position $\alpha^m_{m-3k+i(c)-1}$ in each of the removed votes, and hence loses $k\alpha^m_{m-3k+i(c)-1}$ from the deletion of votes. Due to the above, we know that $c$ gains at least $k\alpha^m_{\ell+3k-1}$ points from each manipulation vote. Therefore, $c$'s final score is at least

$$
\begin{array}{lll}
    & k\cdot\alpha^m_1+\alpha^m_{r(c)}+(k-1)\alpha^m_{m-3k+i(c)-1}-k\alpha^m_{\ell+i(c)-1}-k\alpha^m_{m-3k+i(c)-1}+k\alpha^m_{\ell+3k-1} \\
=   & k\cdot\alpha^m_1+\alpha^m_{r(c)}-\alpha^m_{m-3k+i(c)-1}-k\alpha^m_{\ell+i(c)-1}+k\alpha^m_{\ell+3k-1} \\
\ge & k\cdot\alpha^m_1+(\alpha^m_4-\alpha^m_{m-3k+i(c)-1})-k(\alpha^m_{\ell+i(c)-1}-\alpha^m_{\ell+3k-1}) \\
\ge & k\cdot\alpha^m_1+(\alpha^m_4-\alpha^m_{m-3k})-k(\alpha^m_{\ell}-\alpha^m_{\ell+3k-1}) \\
\end{array}
$$

Since $(\alpha^m_4-\alpha^m_{m-3k})>k(\alpha^m_\ell-\alpha^m_{\ell+3k-1})$, this value exceeds $k\alpha^m_1$, which is the final score of $p$. Therefore, $c$ beats the preferred candidate $p$, and hence the bribery action is not successful, we have a contradiction.

It remains to show that we can add setup vote ensuring that the relative points are as above, and that will never be deleted by the controller. Note that the candidate $b_1$ initially beats $p$ with a headstart of $2k\alpha^m_1$ points. In the manipulation votes, $p$ can gain at most $k\alpha^m_1$ points against $b_1$ (by voting $p$ first and $b_1$ last in all of the votes). Therefore, $b_1$ must lose $k\alpha^m_1$ points from the delete actions, and $p$ may not lose any points from the delete actions. In other words, each of the $k$ deleted votes must give $\alpha^m_1$ points to $b_1$, and $0$ points to $p$. Therefore, it suffices to construct setup votes that result in the above points, and which never have, \emph{at the same time}, $p$ in a position giving $0$ points and $b_1$ in a position giving $\alpha^m_1$ points.

For this, we modify the construction from the implementation Lemma:

\begin{itemize}
 \item If, in the Lemma, the score of $b_1$ is raised by $\alpha^m_1$ (relative to $p$), then we simply repeat one of the 3DM votes, which also has the desired effect. (Note that the controller can never delete two copies of the same 3DM vote, as then she will not be able to construct a cover, which is required as seen above). The side-effect for the other candidates can then be undone by modifying their reletive points (compared to $p$).
 \item If, in the Lemma, the score of some other candidate is changed (relative to $p$), we proceed as follows: We choose $v_0$ as the vote $b_1>p>\RESTcandidates$. Then:
 \begin{itemize}
   \item Since $f$ uses at least $3$ coefficients for $m$ candidates, we know that via rotation of $v_0$, we never produce a vote that gives $\alpha^m_1$ points to $b_1$ and $0$ points to $p$.
   \item Since we are not increasing the score of $b_1$ by $\alpha^m_1$, we never perform a swap that lets $p$ get $0$ points and $b_1$ get $\alpha^m_1$ points.
 \end{itemize}
\end{itemize}

Therefore, none of the setup votes give $0$ points to $p$ and $\alpha^m_1$ points to $b_1$, which concludes the proof.

\end{proof}

\subsubsection{Proof of Theorem~\ref{theorem:veto generalization bribery}}

\thmvetogeneralizationbribery*

\begin{proof}
 The proof is obtained from the proof of the CCDV version (Theorem~\ref{theorem:veto generalization}) in a similar way as the proof of Theorem~\ref{theorem:approval generalization bribery} is obtained from its CCDV version, Theorem~\ref{theorem:approval generalization ccdv}.
 
 Now let a 3DM-instance $M\subseteq X\times Y\times Z$ be given; let $m=3\card M$ be the number of candidates, let $k=\card X$. We write $X\cup Y\cup Z$ as $\set{s_1,\dots,s_{3k}}$, and for $c\in X\cup Y\cup Z$, with $i(c)$ we denote the unique index $i$ with $s_i=c$, and $r(c)$ is $1$, $2$, or $3$, depending on whether $c\in X$, $c\in Y$ or $c\in Z$.
 
 We first choose an $\ell\ge1$ such that $(\alpha^m_{3k+1}-\alpha^m_{m-3})>k(\alpha^m_\ell-\alpha^m_{\ell+3k-1})$ and $n_b\ge 1$, where $n_b=m-\ell-3k+1$. We show that such an $\ell$ exists, the value then can be found in polynomial time since $f$ is polynomial-time uniform. We proceed in a similar way as in the proof of Theorem~\ref{theorem:approval generalization bribery}:
 
 Let $x$ be the number of sequential blocks of length $3k$ that start after position $1$, such that after the start of the last block, there are still at least $3k-3$ positions left, i.e., $x=\left\lceil\frac{m-3k}{3k}\right\rceil-4=\left\lceil\frac{m}{3k}\right\rceil-5$. Now choose $\ell$ of the form $\ell=1+x'3k$ with $1\leq x'\leq x$ such that $\alpha^m_\ell-\alpha^m_{\ell+3k}$ is minimal. Then, $(\alpha^m_{3k+1}-\alpha^m_{m-3})\ge x(\alpha^m_\ell-\alpha^m_{\ell+3k-1})$. 
 
 We first show that the left-hand side of this inequality is strictly positive. Since $\alpha^m_{\frac23m}>\alpha^m_{m-3}$, it suffices to show that $3k+1\leq\frac23m$, i.e., $3\card X+1\leq 2\card M$ (recall that $m=3\card M$). This follows easily since $\card M\ge\card X^3$. To show the required inequality, it now suffices to show that $x>k$. We obtain this as follows:
 
 $$
 \begin{array}{ccccccccccc}
  x & = & \left\lceil\frac{m}{3k}\right\rceil-5 & \ge & \frac{m}{3k}-6 & = & \frac{3\card M}{3\card X}-6 \\
    & = & \frac{\card M}{\card X}-6 & \ge & \frac{\card X^3}{\card X}-6 & \ge & \card X^2-6.
 \end{array}
 $$
 
 Since $\card X=k$, this value exceeds $k$ for sufficiently large $\card X$. We therefore have found a value $\ell$ as required.
 
 Without loss of generality, we assume that $\alpha^m_m=0$, and that the gcd of $\alpha^m_1,\dots,\alpha^m_m$ is $1$.
  
 We now construct the $f$-bribery instance as follows:
 \begin{itemize}
  \item the set of candidates is $X\cup Y\cup Z\cup\set p\cup B\cup D$, where $B$ is a set of $n_b$ blocking candidates and $D$ a set of dummy candidates such that the total number of candidates is $3\card M$,
  \item the preferred candidate is $p$,
  \item for each $(x,y,z)\in M$, we introduce a vote $b_1>S_{xyz}>\RESTcandidates>z>y>x>p$, where $S_{xyz}$,
  \item we introduce additional setup votes (see below) ensuring that the relative points of the candidates are as follows:
  \begin{itemize}
    \item $\score p=0$,
    \item $\score{b_1}=2k\alpha^m_1$,
    \item $\score{b_i}=k\alpha^m_1+k(\alpha^m_1-\alpha^m_{m-i+1})$ for $i\ge2$ 
    \item $\score c=k\alpha^m_1+\alpha^m_{m-r(c)}+(k-1)\alpha^m_{1+i(c)}-k\alpha^m_{\ell+i(c)-1}$.
  \end{itemize}
 \end{itemize}

 In a similar way as in the proof of Theorem~\ref{theorem:approval generalization bribery}, we can show that $p$ can be made a winner of the election by bribing at most $k$ of the above-introduced voters if and only if the 3DM instance is positive. In particular, the controller needs to bribe exactly $k$ voters, since $b_1$ must lose $k\alpha^m_1$ points, and $p$ must win $k\alpha^m_1$ points. The controller can make $p$ win the election by bribing the votes corresponding to a cover, and then letting every bribed voter vote
 
 $$p>\RESTcandidates>s_1>s_2>\dots>s_{3k}>b_{n_b}>\dots>b_1.$$
 
 Now assume that the controller can make $p$ win the election by bribing at most $k$ (and hence, without loss of generality, with exactly $k$) voters. With the same mechanism as used in the proof of Theorem~\ref{theorem:approval generalization bribery}, we know that the candidates $b_1,\dots,b_{n_b}$ must occupy the last $n_b$ positions in every manipulator vote. Therefore, each candidate $c\in X\cup Y\cup Z$ obtains at least $k\alpha^m_{\ell+3k-1}$ points from the $k$ manipulator votes.
 
 Now assume that the bribed voters do not correspond to a cover. Since $k$ voters are bribed, then there is some candidate $c$ such that at least two votes voting $c$ in one of the last four positions are deleted. Therefore, $c$ loses at most $(k-2)\alpha^m_{1+i(c)}+2\alpha^m_{m-r(c)}$ points from the delete action, and gains at least $k\alpha^m_{\ell+3k-1}$ points from the manipulation votes.
 
 Therefore, the final score of $c$ is at least
 
 $$
 \begin{array}{cc}
     & k\alpha^m_1+\alpha^m_{m-r(c)}+(k-1)\alpha^m_{1+i(c)}-k\alpha^m_{\ell+i(c)-1}-2\alpha^m_{m-r(c)}-(k-2)\alpha^m_{1+i(c)}+k\alpha^m_{\ell+3k-1} \\
  =  & k\alpha^m_1-\alpha^m_{m-r(c)}+\alpha^m_{1+i(c)}-k(\alpha^m_{\ell+i(c)-1}-\alpha^m_{\ell+3k-1}) \\
 \ge & k\alpha^m_1+(\alpha^m_{1+i(c)}-\alpha^m_{m-4})-k(\alpha^m_\ell-\alpha^m_{\ell+3k-1}) \\
 \ge & k\alpha^m_1+(\alpha^m_{3k+1}-\alpha^m_{m-4})-k(\alpha^m_\ell-\alpha^m_{\ell+3k-1}).
 \end{array}
 $$
 
 Since $p$ has at most $k\alpha^m_1$ points after the bribery and $(\alpha^m_{3k+1}-\alpha^m_{m-4})>k(\alpha^m_\ell-\alpha^m_{k+3k-1})$ due to the choice of $\ell$, this means that $c$ beats $p$, and hence we have a contradiction.
 
 It remains to show how to construct the setup votes to ensure that the scores are as above. This can be done with the exact same mechanism as in the proof of Theorem~\ref{theorem:approval generalization bribery}.
\end{proof}

\subsubsection{Proof of Theorem~\ref{theorem:bribery fixed coefficients additional case}}

As discussed above, compared to CCDV, we need an additional case to cover all ``many coefficients''-cases of the bribery problem, the above-stated Theorem~\ref{theorem:bribery fixed coefficients additional case}. We now prove this result.

\theorembriberyfixedcoefficientsadditionalcase*

\begin{proof}
 We write $f$ as $(\alpha_1,\alpha_2,\alpha_3,0,\dots,0,-\alpha_4,-\alpha_5,-\alpha_6)$ for $\alpha_3>0$ and $\alpha_4\leq\alpha_5\leq\alpha_6\ge0$. We reduce from 3DM as follows: Let $M\subseteq X\times Y\times Z$ be given. Again, for $c\in X\cup Y\cup Z$, we define $r(c)$ as $1$, $2$, or $3$, depending on whether $c\in X$, $c\in Y$, or $c\in Z$. We construct an instance of $f$-bribery as follows:
 
 \begin{itemize}
  \item the set of candidates is $X\cup Y\cup Z\cup p\cup\set{b_1,b_2,b_3}\cup D$ for a sufficiently large set $D$ of dummy candidates.
  \item for each $(x,y,z)$ we add a vote $x>y>z>\dots>d_1,d_2>p$ for dummy candidates $d_1$ and $d_2$.
  \item the preferred candidate is $p$, the controller's budget is $\card X$.
  \item we additionally add setup votes to ensure the following relative scores:
  
  \begin{itemize}
    \item $\score p=-(\alpha_1+\alpha_6)$,
    \item $\score c=\alpha_{r(c)}$,
    \item $\score{b_i}=k\alpha_{m-i+1}$,
    \item the score of the dummy candidates is so low that they cannot win the election.
  \end{itemize}
 \end{itemize}
 
 If the 3DM instance is positive, then the controller can remove the $\card X$ many votes corresponding to the cover, and let all bribed voters vote $p>d_1>d_2>\RESTcandidates>b_3>b_2>b_1$, this lets $p$ gain $k(\alpha_1+\alpha_6)$ points, each $c$ loses $\alpha_{r(c)}$ points, and each $b_i$ loses $k\alpha_{m-i+1}$ points, hence all candidates tie and $p$ wins the election.
 
 For the other direction, if the controller can bribe at most $k$ votes of the above to ensure that $p$ wins the election, then $p$ has at most $0$ points after the bribery. By construction, the candidates $b_1$, $b_2$, and $b_3$ must occupy the last $3$ positions in each manipulation vote, therefore no $c\in X\cup Y\cup Z$ can lose points from the manipulation votes. Therefore, the removed votes must form a cover of $X\cup Y\cup Z$ to ensure that each $c\in X\cup Y\cup Z$ loses $\alpha_{r(c)}$ points and thus ties with $p$.
 
 It remains to show how to construct the corresponding setup votes. This follows in the same way as in the proof of Theorem~\ref{theorem:approval generalization ccdv}: We add enough dummy candidates such that in the votes resulting from the application of the construction lemma, at most $2$ of the first $3$ positions are filled with relevant candidates. Since the points are ``tight,'' this means that the controller cannot afford to remove a vote not introduced for a 3DM-tuple as above.
\end{proof}

\subsubsection{Proof of Theorem~\ref{theorem:two coefficients case bribery}}

We now give the proof of Theorem~\ref{theorem:two coefficients case bribery}, classifying the complexity of the bribery problem for all generators that only use two coefficients.

\theoremtwocoefficientscasebribery*

\begin{proof}
 Since $f$ only uses two coefficients, we can without loss of generality assume that for each $m$, $f(m)$ is of the form $(1^{a(m)}0^{b(m)})$ for some functions $a,b\colon\mathbb N\rightarrow\mathbb N$. Note that $a(m)+b(m)=m$. Since $f$ is a pure generator, the functions $a$ and $b$ are monotone. Since $\alpha^m_3>\alpha^m_4$, we know that for sufficiently large $m$, we have that $a(m)\ge3$ and $b(m)\ge4$. If $a(m)\leq 3$ for all $m$, then, for sufficiently large $m$, $f$ is $3$-approval, analogously if $b(m)\leq 4$ for all $m$ then $f$ is ultimately equivalent to $4$-veto. In both cases, \NP-hardness of $f$-bribery follows from Theorem~\ref{theorem:lin bribery results list}. Therefore, we can assume that, for sufficiently large $m$, $a(m)\ge4$ and $b(m)\ge5$.
 
 We reduce from 3DM as follows. Let $M\subseteq X\times Y\times Z$ be a 3DM-instance, let $k=\card X$. We use the candidate set $X\cup Y\cup Z\cup p\cup D\cup\set{B}$, i.e., the number of candidates is $m=3k+2+\card D$, where $D$ is a set of $9k+6$ dummy candidates, i.e., $m=9k+8$. Without loss of generality, we assume $a(m)\ge4$ and $b(m)\ge5$. We make a case distinction.
 
 First consider the case that $b(m)\ge 3k+1$. In this case, we treat the generator as a variant of approval, since there are enough ``bad positions'' in the vote to place ``most of the candidates.'' We set up the election as follows:
 
 \begin{itemize}
  \item for each $(x,y,z)\in M$, we introduce a vote $b_1>x>y>z>D>X\cup Y\cup Z\setminus{x,y,z}>p$; again, we call these votes 3DM-votes.
  \item we set up the points such that, relative to $p$,
  \begin{itemize}
    \item $\score p=0$,
    \item $\score b_1=2k$,
    \item $\score c=k+1$ for each $c\in X\cup Y\cup Z$,
    \item the dummy candidates do not have enough points to win the election,
    \item the controller's budget is $k$.
  \end{itemize}
  \end{itemize}
  It is easy to see that the preferred candidate $p$ can be made a winner of the election with at most $k$ bribes if and only if the 3DM instance is positive:
  
  If the instance is positive, then the controller can bribe a set of voters corresponding to a cover, and let each manipulator vote $p>\dots>X\cup Y\cup Z>b_1$. Then $p$ gains $k$ points, $b_1$ loses $k$ points, and each $c\in X\cup Y\cup Z$ loses $1$ points, hence all non-dummy candidates tie and $p$ wins the election.
  
  For the other direction, clearly each $c\in X\cup Y\cup Z$ must lose $1$ points, hence the deleted votes form a cover.
  
  The setup votes can easily be constructed by repeating 3DM votes to lower the score of $p$ relative to other candidates, and then using votes having $p$ as the first candidate and the remaining candidates in a suitable order to adjust their score. Clearly, only (copies of) the 3DM votes will be deleted, since $p$ is optimally positioned in the other votes.

 Now consider the case that $b(m)<3k+1$. Then $a(m)=m-b(m)>m-3k-1=9k+8-3k-1=6k+7$. We use the following reduction:
 
 \begin{itemize}
  \item for each $(x,y,z)\in M$, we introduce a vote $b_1>X\cup Y\cup Z\setminus{x,y,z}>d_1>\dots>d_{b(m)}>\dots>x>y>z>p$, again these votes are called 3DM votes. Note that the candidates $d_1,\dots,d_{b(m)}$ each get $1$ point from each 3DM vote.
  \item let $d_1,\dots,d_{b(m)}$ be the first $b(m)$ dummy candidates (these exist, since $b(m)<3k+1$),
  \item we set up the points such that, relative to $p$,
  \begin{itemize}
    \item $\score p=0$,
    \item $\score c=k-1$ for each $c\in X\cup Y\cup Z$,
    \item $\score d_i=2k$ for each $i\leq b(m)$,
    \item the remaining dummy candidates do not have enough points to win the election,
    \item the controller's budget is $k$.
  \end{itemize}
 \end{itemize}
 
 Again, $p$ can be made a winner of the election with at most $k$ bribes if and only if the 3DM instance is positive:
 
 If the instance is positive, the controller bribes a set of votes corresponding to the cover, and lets each manipulator vote $p>\dots>d_{b(m)}>\dots>d_1$. Then:
 
 \begin{itemize}
  \item $p$ gains $k$ points, and has $k$ points in the end,
  \item each $c\in X\cup Y\cup Z$ loses $k-1$ point from the delete action and gains $k$ points from the manipulation votes, and thus ties with $p$,
  \item each $d_i$ for $i\leq b(m)$ loses $k$ points and hence ties with $p$.
 \end{itemize}
 
 Therefore, $p$ wins the election. For the other direction, note that with deleting $k$ 3DM votes, the candidates $b_i$ for $i\leq b(m)$ still have $k$ points each, and thus must be places in the $0$-point segment of the manipulation votes, in particular, each $c\in X\cup Y\cup Z$ gets $k$ points from the manipulation votes. Now assume that the deleted votes to not correspond to a cover. Since $k$ votes must be deleted to ensure that $b_1$ does not beat $p$, this means that there is some $c\in X\cup Y\cup Z$ such that two votes not giving any point to $c$ are removed. Therefore, $c$ loses only at most $k-2$ points in the delete action, and thus (since $c$ also gains $k$ points from the manipulation votes) ends with $k+1$ points, beating $c$. This is a contradiction.
 
 The setup votes are constructed in the exact same way as in the above case $b(m)\ge3k+1$.
\end{proof}

\subsubsection{Proof of Corollary~\ref{corollary:many coefficients bribery hardness}}

The following combines the previous results to obtain the hardness result for all ``many coefficients'' bribery cases.

\corollarymanycoefficientsbriberyhardness*

\begin{proof}
 If $f$ only uses two coefficients, then the claim follows from Theorem~\ref{theorem:two coefficients case bribery}. Therefore, assume that $f$ uses at least three coefficients for the remainder of the proof.
 
 We reduce from 3DM. Hence, let a 3DM-instance $M\subseteq X\times Y\times Z$ be given. We apply Proposition~\ref{prop:f3dm np complete} to ensure that $\card M\ge\card X^3$ and $\card M\ge\card X^2+2\card X+2$. Let $m=3\card M$. Without loss of generality, assume that $m$ is large enough such that $f(m)$ uses $3$ different coefficients, and such that $\alpha^m_3>\alpha^m_{m-3}$ (note that when this inequality is true for any $m$, it remains true for all larger $m$).
 
 We now make a case distinction:
 \begin{itemize}
  \item If $\alpha^m_4>\alpha^m_{\frac23m}$, then we apply the reduction from Theorem~\ref{theorem:approval generalization bribery}.
  \item If $\alpha^m_{\frac23m}>\alpha^m_{m-4}$, then we apply the reduction from Theorem~\ref{theorem:veto generalization bribery}.
 \end{itemize}
 
 Otherwise, we have that $\alpha^m_4=\alpha^m_{m-4}$. Since $\alpha^m_3>\alpha^m_{m-4}$, this implies $\alpha^m_3>\alpha^m_4=\alpha^m_{m-4}$, and $f$ is of the form $(\alpha_1,\alpha_2,\alpha_3,\alpha_4,\dots,\alpha_4,\alpha_5,\alpha_6,\alpha_7)$ with $\alpha_3>\alpha_4$. Therefore, the result follows from Theorem~{theorem:bribery fixed coefficients additional case}.
\end{proof}

\end{appendix}

\end{document}